\documentclass[11pt]{article}
\usepackage[english]{babel}
\usepackage{amsfonts,amsmath}
\usepackage{amssymb}
\usepackage{graphicx}
\usepackage{a4wide}
\newtheorem{theorem}{Theorem}

\newtheorem{corollary}[theorem]{Corollary}
\newtheorem{definition}[theorem]{Definition}
\newtheorem{lemma}[theorem]{Lemma}
\newtheorem{proposition}[theorem]{Proposition}
\newtheorem{remark}[theorem]{Remark}
\newenvironment{proof}[1][Proof]{\noindent\textbf{#1.} }{\ \rule{0.5em}{0.5em}}
\numberwithin{equation}{section}
\numberwithin{theorem}{section}
\newcommand{\z}{\mathfrak{z}}

\begin{document}
\date{}
\title{Non-Isothermal Boundary in the Boltzmann Theory and Fourier Law}
\author{R. Esposito$ ^1$, Y. Guo$ ^2$, C. Kim$ ^3$ and R. Marra$ ^4$}
\maketitle
\footnotetext [1]{International Research Center M\&MOCS, Univ. dell'Aquila, Cisterna di Latina, (LT) 04012 Italy}
\footnotetext [2]{Division of Applied Mathematics, Brown University, Providence, RI 02812, U.S.A.}
\footnotetext [3]{Department of Pure Mathematics and Mathematical Statistics, University of Cambridge, Cambridge\\ CB3 0WA, UK}
\footnotetext [4]{Dipartimento di Fisica and Unit\`a INFN, Universit\`a di
Roma Tor Vergata, 00133 Roma, Italy. }

\begin{abstract}
In the study of the heat transfer in the Boltzmann theory, \ the basic
problem is to construct solutions to the following steady problem:
\begin{eqnarray}
v\cdot \nabla _{x}F &=&\frac{1}{\text{K}_{\text{n}}}Q(F,F),\text{ \ \ \ \ \ }%
(x,v)\in \Omega \times \mathbf{R}^{3},  \label{boltzmann} \\
\hskip -1cmF(x,v)|_{n(x)\cdot v<0} &=&\mu _{\theta }\int_{n(x)\cdot v^{\prime
}>0}F(x,v^{\prime })(n(x)\cdot v^{\prime })dv^{\prime},\quad x\in
\partial \Omega \text{,}  \label{fdiffuse}
\end{eqnarray}%
where $\Omega $ is a bounded domain in $\mathbf{R}^{d}$, $1\leq d\leq 3$, ${\text{K}_{\text{n}}}$ is the Knudsen number and $\mu _{\theta }=\frac{1}{2\pi \theta ^{2}(x)}\exp [-\frac{|v|^{2}}{%
2\theta (x)}]$ is a Maxwellian with non-constant (non-isothermal)  wall temperature $%
\theta (x)$. Based on new constructive coercivity estimates for both steady and dynamic cases, for $|\theta -\theta_{0}|\leq \delta \ll 1$ and any fixed value of  ${\text{K}_{\text{n}}}$, we  construct a unique solution $F_s$ to (\ref{boltzmann}) and (\ref{fdiffuse}), continuous away from the grazing set and exponentially asymptotically stable.
This solution  is a genuine non equilibrium stationary solution differing from  a local equilibrium Maxwellian. As an application of our results we establish the expansion $F_s=\mu_{\theta_0}+\delta F_{1}+O(\delta ^{2})$ and we prove that, if the Fourier law holds,  the temperature contribution associated to $F_1$ must be linear, in the slab geometry. This contradicts available numerical simulations, leading to the prediction of breakdown of the Fourier law in the kinetic regime.
\end{abstract}

\bigskip
\tableofcontents

\vskip 1cm

\section{Introduction and notation}

According to the Boltzmann equation (\ref{boltz}),
a rarefied gas
confined in a bounded domain, in contact with a thermal reservoir
modeled by (\ref{fdiffuse})
at a given constant temperature $\theta $ (isothermal), has an
equilibrium state described by the Maxwellian
\begin{equation*}
\text{constant}\times \text{\textrm{e}}^{-|v|^{2}/2\theta },
\end{equation*}%
and it is well known \cite{U1,CIP,Gui,DV,V2,Guo08, V} that such an equilibrium
is reached exponentially fast, at least if the initial state is close to the
above Maxwellian in a suitable norm.

If the temperature $\theta $ at the boundary is not uniform in (\ref%
{fdiffuse}), such a statement is not true and the existence of stationary
solutions and the rate of convergence require a much more delicate analysis,
because rather complex phenomena are involved. For example, suppose that the
domain is just a slab between two parallel plates at fixed temperatures $%
\theta _{-}$ and $\theta _{+}$ with $\theta _{-}<\theta _{+}$.  Then, one
expects that a stationary solution is reached, where there is a steady flow
of heat from the hotter plate to the colder one. The approach to the
stationary solution may involve convective motions, oscillations and
possibly more complicated phenomena. Even the description of the stationary
solution is not obvious, and the relation between the heat flux and the
temperature gradient, (e.g. the Fourier law (\ref{fourierlaw1})), is not
\textit{a priori} known. A first answer to such questions can be given
confining the analysis to the small Knudsen number regime (i.e., K$_{\text{n}%
}\rightarrow 0$ in (\ref{boltzmann})), where the particles undergo a large
number of collision per unit time and a hydrodynamic regime is established.
In this case, it can be formally shown, using expansion techniques \cite%
{C,CC}, that the lowest order in K$_{\text{n}}$  is a  Maxwellian local equilibrium
and the evolution is ruled by macroscopic equations, such as the
Navier-Stokes equations. In particular, the heat flux vector $q$ turns out to be
proportional to the gradient of temperature, as predicted by the Fourier law
\begin{equation}\label{fourierlaw1} q=-\kappa(\theta)\nabla_x \theta\end{equation}
with the heat conductivity  $\kappa (\theta )$ depending on the interaction potential. This was first obtained by Maxwell and Boltzmann \cite{M, B}
which relates the macroscopic heat flow to the microscopic potential of
interaction between the molecules. The rigorous proof of such a statement
was given in \cite{ELM1,ELM2} in the case of the slab geometry and provided
that $\theta _{+}-\theta _{-}$ is sufficiently small (uniformly in the
Knudsen number K$_{\text{n}}$). This is a special case of a problem which
has received recently a large attention in the Statistical Mechanics
community, the derivation of the Fourier law from the microscopic deterministic evolutions
ruled by the Newton or Schr\"{o}dinger equation or from stochastic models \cite{BLR, O, BOS, ALS}. 

The aim of this paper is to analyze the thermal conduction phenomena in the
kinetic regime. This problem was studied in the slab geometry, for small Knudsen numbers in \cite{ELM1,ELM2} and for large Knudsen number (K$_{\text{n}}\rightarrow \infty $) in \cite{Yu}.
Here we are interested in a general domain and in a regime where  the Knudsen number K$_{\text{n}}$ is neither
small nor large. In this regime, the only construction of solutions to (\ref%
{boltzmann}) and (\ref{fdiffuse}) we are aware of,  was achieved in \cite{AN} in a slab for
large $\theta _{+}-\theta _{-}$, with $L^1$ techniques closer in the spirit to the DiPerna-Lions renormalized
solutions \cite{DL, DL2} (see also \cite{V} and references quoted therein). However, the uniqueness and stability of such $L^{1}$
solutions are unknown. To develop  the quantitative analysis we have in mind,
the theory of the solutions close to the equilibrium  (\cite{U1,U2,CIP,Guo08}) is better suited. For this reason we confine ourselves to
temperature profiles at the boundary which do not oscillate too much. More
precisely, we will assume that the temperature $\theta (x)$ on $\partial
\Omega $ is given by $\theta (x)=\theta _{0}+\delta \vartheta (x)$ with $%
\delta $ a small parameter and $\vartheta (x)$ a prescribed bounded function
on $\partial \Omega $ such that
\begin{equation*}
\sup_{x\in \partial \Omega }|\vartheta (x)|\leq 1.
\end{equation*}%
This will allow us to use perturbation arguments in the neighborhood of the
equilibrium at the\textit{\ uniform }temperature $\theta _{0}$.

We shall consider the Boltzmann equation
\begin{equation}
\partial _{t}F+v\cdot \nabla F=\frac{1}{\text{K}_{\text{n}}}Q(F,F),
\label{boltz}
\end{equation}%
with $F(t,x,v)$ the probability density that a particle of the gas at time $t
$ is in a small cell of the phase space $\Omega \times \mathbf{R}^{3}$
centered at $(x,v)$.  Here $\Omega $ is a bounded domain in $\mathbf{R}^{d}$,
$d=1,2,3$ with a smooth boundary $\partial \Omega $.  The function $F$ is
required to be a positive function on $\Omega \times \mathbf{R}^{3}$ such
that $\iint_{\Omega \times \mathbf{R}^{3}}Fdxdv$ is fixed for any $t$.  The right
hand side of (\ref{boltz}), $Q(F,G)$, is the Boltzmann collision operator
(non-symmetric)
\begin{eqnarray}
Q(F,G) &=&\int_{\mathbf{R}^{3}}dv_{\ast }\int_{\mathbf{S}^{2}}d\omega
B(v-v_{\ast },\omega )F(v_{\ast }^{\prime })G(v^{\prime })  \notag \\
&&-\int_{\mathbf{R}^{3}}dv_{\ast }\int_{\mathbf{S}^{2}}d\omega B(v-v_{\ast
},\omega )F(v_{\ast })G(v)  \notag \\
&\equiv &Q_{\text{gain}}(F,G)-Q_{\text{loss}}(F,G),  \label{QFG}
\end{eqnarray}%
where $B(v,\omega )=|v|^{\gamma }q_{0}\left( \omega \cdot \frac{v}{|v|}%
\right) $ with $0\leq \gamma \leq 1$ (hard potential), $0\leq
q_{0}(\omega \cdot \frac{v}{|v|})\leq C|\omega \cdot \frac{v}{|v|}|$
(angular cutoff) is the collision cross section and $v^{\prime },v_{\ast
}^{\prime }$ are the incoming velocities in a binary elastic collision with
outgoing velocities $v,v_{\ast }$ and impact parameter $\omega $:
\begin{equation}
v^{\prime }=v-\omega \lbrack (v-v_{\ast })\cdot \omega ],\quad v_{\ast
}^{\prime }=v_{\ast }+\omega \lbrack (v-v_{\ast })\cdot \omega ].
\label{elast}
\end{equation}%

The contact of the gas with thermal reservoirs is described by suitable
boundary conditions. We confine ourselves to the simplest interesting case
of the diffuse reflection (\ref{fdiffuse}), although more general boundary
data could be studied \cite{CIP}. On $\partial \Omega $ (supposed to be a $%
C^{1}$-smooth surface with external normal $n(x)$ well defined in each point
$x\in \partial \Omega $) we assume the condition: 
\begin{equation}
F(t,x,v)=\mu _{\theta }(x,v)\int_{n(x)\cdot v^{\prime }>0}F(t,x,v^{\prime
})\{n(x)\cdot v^{\prime }\}dv^{\prime },  \label{diffbc}
\end{equation}%
for $x\in \partial \Omega $ and $n(x)\cdot v<0$, where $\mu _{\theta }$ is the
Maxwellian at temperature $\theta $,
\begin{equation}
\mu _{\theta }(x,v)=\frac{1}{2\pi \theta ^{2}(x)}\exp \left[ -\frac{|v|^{2}}{%
2\theta (x)}\right] ,  \label{maxw}
\end{equation}%
normalized so that
\begin{equation}
\int_{n(x)\cdot v>0}\mu _{\theta }(x,v)\{n(x)\cdot v\}dv=1.  \label{norm}
\end{equation}

Throughout this paper, $\Omega $ is a connected and bounded domain in $%
\mathbf{R}^{d}$, for $d=1,2,3$ and the velocity $\bar{v}\in \mathbf{R}^{d}$
and $\hat{v}\in \mathbf{R}^{3-d}$ such that
\begin{equation}
v=(v_{1},\cdots ,v_{d},v_{d+1},\cdots ,v_{3})=(\bar{v},\hat{v}).  \label{v}
\end{equation}%
We denote the phase boundary in the phase space $\Omega \times \mathbf{R}^{3}
$ as $\gamma =\partial \Omega \times \mathbf{R}^{3}$, and split it into the
outgoing boundary $\gamma _{+}$, the incoming boundary $\gamma _{-}$, and
the grazing boundary $\gamma _{0}$:
\begin{eqnarray*}
\gamma _{+} &=&\{(x,v)\in \partial \Omega \times \mathbf{R}^{3}\ :\
n(x)\cdot v>0\}, \\
\gamma _{-} &=&\{(x,v)\in \partial \Omega \times \mathbf{R}^{3}\ :\
n(x)\cdot v<0\}, \\
\gamma _{0} &=&\{(x,v)\in \partial \Omega \times \mathbf{R}^{3}\ :\
n(x)\cdot v=0\}.
\end{eqnarray*}%
The backward \textit{exit time} $t_{\mathbf{b}}(x,v)$ is defined for $%
(x,v)\in \overline{\Omega}\times \mathbf{R}^{3}$
\begin{equation}
t_{\mathbf{b}}(x,v)=\inf \{\ t\geq 0:x-t\bar{v}\in \partial \Omega \},
\label{backexit}
\end{equation}%
and $x_{\mathbf{b}}(x,v)=x-t_{\mathbf{b}}(x,v)\bar{v}\in \partial \Omega $.
Furthermore, we define the singular grazing boundary $\gamma _{0}^{\mathbf{S}%
}$, a subset of $\gamma _{0}$, as:
\begin{equation}
\gamma _{0}^{\mathbf{S}}\ =\ \{(x,v)\in \gamma _{0}\ :\ t_{\mathbf{b}%
}(x,-v)\neq 0\ \ \text{and }\ t_{\mathbf{b}}(x,v)\neq 0\},  \label{graze}
\end{equation}%
and the discontinuity set in $\overline{\Omega}\times \mathbf{R}^{3}$:
\begin{equation}
\mathfrak{D}=\gamma _{0}\cup \{(x,v)\in \overline{\Omega}\times \mathbf{R}^{3}\
:\ (x_{\mathbf{b}}(x,v),v)\in \gamma _{0}^{\mathbf{S}}\}.  \label{Dset}
\end{equation}
We will use the short notation $\mu _{\delta }$ for the Maxwellian
\begin{equation}
\mu _{\delta }(x,v)=\mu _{\theta _{0}+\delta \vartheta (x)}(v)=\frac{1}{2\pi
\lbrack \theta _{0}+\delta \vartheta (x)]^{2}}\exp \left[ -\frac{|v|^{2}}{%
2[\theta _{0}+\delta \vartheta (x)]}\right] .  \label{mudelta}
\end{equation}%
Moreover, to denote the global Maxwellian at temperature $\theta_0$, $\mu _{\theta _{0}}$, we will simply use the symbol $\mu $:
\begin{equation*}
\mu \equiv \mu _{\theta _{0}}. 
\end{equation*}%
Since (\ref{norm}) is valid for all $\delta $, we have
\begin{equation}
\int_{n(x)\cdot v>0}\mu(v) \{n(x)\cdot v\}dv=1.  \label{munorm}
\end{equation}
We denote by $L$ the standard linearized Boltzmann operator
\begin{eqnarray}
Lf &=&-\frac{1}{\sqrt{\mu }}[Q(\mu ,\sqrt{\mu }f)+Q(\sqrt{\mu }f,\mu )]=\nu
(v)f-Kf  \notag \\
&=&\nu (v)f-\int_{\mathbf{R}^{3}}\mathbf{k}(v,v_*)f(v_*)dv_*,  \label{linBol}
\end{eqnarray}%
with the collision frequency $\nu (v)\equiv \iint_{\mathbf{R}^{3}\times
\mathbf{S}^{2}}B(v-v_{\ast },\omega )\mu (v_{\ast })d\omega dv_{\ast }\sim
\{1+|v|\}^{\gamma }$ for $0\leq \gamma \leq 1$.  Moreover, we set
\begin{equation}
\Gamma (f_{1},f_{2})=\frac{1}{\sqrt{\mu }}Q(\sqrt{\mu }f_{1},\sqrt{\mu }%
f_{2})\equiv \Gamma _{\text{gain}}(f_{1},f_{2})-\Gamma _{\text{loss}%
}(f_{1},f_{2}).\label{nonBol}
\end{equation}%
Finally, we define
\begin{equation}
P_{\gamma }f(x,v)=\sqrt{\mu (v)}\int_{n(x)\cdot v^{\prime }>0}f(x,v^{\prime
})\sqrt{\mu (v^{\prime })}(n(x)\cdot v^{\prime })dv^{\prime }.
\label{pgamma}
\end{equation}
Thanks to (\ref{munorm}), $P_{\gamma
}f $, viewed as function on $\{v\in \mathbf{R}^3\ |\ v\cdot n(x)>0\}$ for any fixed $x\in \partial\Omega$,  is a $L_{v}^{2}$-projection with respect to the measure $|n(x)\cdot v|$
for any boundary function $f$ defined on $\gamma _{+}$.

We denote $\Vert \,\cdot \,\Vert _{\infty }$ either the $L^{\infty
}(\Omega \times \mathbf{R}^{3})-$norm or the $L^{\infty }(\Omega )-$norm in the bulk,
while $|\,\cdot \,|_{\infty }$ is either the $L^{\infty }(\partial \Omega
\times \mathbf{R}^{3})-$norm or the $L^{\infty }(\partial \Omega )-$norm at the boundary.
Also we adopt the Vinogradov notation: $X\lesssim Y$ is equivalent to $|X|\leq CY$
where $C$ is a constant not depending on $X$ and $Y$.  We subscript this to
denote dependence on parameters, thus $X\lesssim _{\alpha }Y$ means $|X|\leq
C_{\alpha }Y$.  Denote $\langle v\rangle =\sqrt{1+|v|^{2}}$.  Our main results
are as follows.

\begin{theorem}
\label{main1} There exists $\delta _{0}>0$ such that for $%
0<\delta <\delta _{0}$ in (\ref{mudelta}) and for all $M>0$, there exists a non-negative
solution $F_{s}=M\mu +\sqrt{\mu }f_{s}\geq 0$ with $\iint_{\Omega \times
\mathbf{R}^{3}}f_{s}\sqrt{\mu }dxdv=0$ to the steady problem (\ref{boltzmann}%
) and (\ref{fdiffuse}) such that for all $0\leqq \zeta <\frac{1}{4+2\delta }%
,\ \beta >4$,
\begin{equation}
\|\langle v\rangle ^{\beta }e^{\zeta |v|^{2}}f_{s}\|_{\infty }+|\langle
v\rangle ^{\beta }e^{\zeta |v|^{2}}f_{s}|_{\infty }\lesssim \delta .
\label{inftymain}
\end{equation}%
If $M\mu +\sqrt{\mu }g_{s}$ with $\iint_{\Omega \times \mathbf{R}^{3}}g_{s}%
\sqrt{\mu }dxdv=0$ is another solution such that, for $\beta >4$
\begin{equation}
\|\langle v\rangle ^{\beta }g_{s}\|_{\infty }+|\langle v\rangle ^{\beta
}g_{s}|_{\infty }\ll 1,  \notag
\end{equation}%
then $f_{s}\ \equiv \ g_{s}$.  Furthermore, if $\theta (x)$ is continuous on $%
\partial \Omega $ then $F_{s}$ is continuous away from $\mathfrak{D}$.  In
particular, if $\Omega $ is convex then $\mathfrak{D}=\gamma _{0}$.  On the
other hand, if $\Omega $ is not convex, we can construct a continuous
function $\vartheta (x)$ on $\partial \Omega $ in (\ref{temper}) with $%
|\vartheta |_{\infty }\leq 1$ such that the corresponding solution $F_{s}$
in not continuous.
\end{theorem}

We stress that the solution $F_s$ is a genuine non equilibrium steady solution. Indeed, it is not a local Maxwellian because it would not satisfy equation (\ref{boltzmann}),  nor a global Maxwellian because it would not satisfy the boundary condition  (\ref{fdiffuse}). 

We also remark that in addition to the wellposedness property, the continuity
property in this theorem is the first step to understand higher regularity
of $F_{s}$ in a convex domain. The robust $L^{\infty }$ estimates used in
the proof enable us to establish the following $\delta -$expansion of $F_{s}$%
, which is crucial for deriving the necessary condition of the Fourier law (%
\ref{fourierlaw1}). In the rest of this paper we assume the normalization $M=1$.

\begin{theorem}
\label{main2} Let
\begin{equation}
\mu _{\delta }=\mu +\delta \mu _{1}+\delta ^{2}\mu _{2}+\cdots ,
\end{equation}
with $\int \mu _{i}d\gamma =0$ for all $i$ from (\ref{norm}), and set
\begin{equation*}
F_{s}=\mu +\sqrt{\mu }f_{s}.
\end{equation*}
Then there exist $f_{1},f_{2},...,f_{m-1}$ with $\|\langle v\rangle ^{\beta
}e^{\zeta |v|^{2}}f_{i}\|_{\infty }\lesssim 1$, for $0\leqq \zeta <\frac{1}{4
},\beta >4$, such that the following $\delta$-expansion is valid
\begin{equation*}
f_{s}=\delta f_{1}+\delta ^{2}f_{2}+\cdots +\delta ^{m}f_{m}^{\delta },
\end{equation*}
with $\|\langle v\rangle ^{\beta }e^{\zeta |v|^{2}}f_{m}^{\delta }\|_{\infty
}\lesssim 1$ for $0\leqq \zeta <\frac{1}{4+2\delta },\ \beta >4$.  In
particular, $f_{1}$ satisfies
\begin{eqnarray}
v\cdot \nabla _{x}f_{1}+\frac{1}{\text{K}_{\text{n}}}Lf_{1} &=&0,  \label{f1}
\\
f_{1}|_{\gamma _{-}} &=&\sqrt{\mu (v)}\int_{n(x)\cdot v^{\prime
}>0}f_{1}(x,v^{\prime })\sqrt{\mu (v^{\prime })}\{n(x)\cdot v^{\prime
}\}dv^{\prime }+\frac{\mu _{1}}{\sqrt{\mu }}.  \notag
\end{eqnarray}
\end{theorem}

We have the following dynamical stability result:

\begin{theorem}
\label{main3} For every fixed $0\leqq \zeta <\frac{1}{4+2\delta },\ \beta >4$
, there exist $\lambda >0$ and $\varepsilon _{0}>0$, depending on $\delta _{0}
$, such that if
\begin{equation}
\|\langle v\rangle ^{\beta }e^{\zeta |v|^{2}}[f(0)-f_{s}]\|_{\infty
}+|\langle v\rangle ^{\beta }e^{\zeta |v|^{2}}[f(0)-f_{s}]|_{\infty }\leq
\varepsilon _{0},\label{epsilon0}
\end{equation}
then there exists a unique non-negative solution {$F(t)=\mu +
\sqrt{\mu }f(t)\geq 0$} to the dynamical problem (\ref{boltz}) and (\ref
{diffbc}) such that
\begin{eqnarray*}
&&\ \ \|\langle v\rangle ^{\beta }e^{\zeta |v|^{2}}[f(t)-f_{s}]\|_{\infty
}+|\langle v\rangle ^{\beta }e^{\zeta |v|^{2}}[f(t)-f_{s}]|_{\infty } \\
&&\ \lesssim \ e^{-\lambda t}\big\{\|\langle v\rangle ^{\beta }e^{\zeta
|v|^{2}}[f(0)-f_{s}]\|_{\infty }+|\langle v\rangle ^{\beta }e^{\zeta
|v|^{2}}[f(0)-f_{s}]|_{\infty }\big\}.
\end{eqnarray*}%
If the domain is convex, $\theta (x)$ is continuous on $\partial \Omega $ and moreover
 $F_{0}(x,v)$ is continuous away from $\gamma _{0}$ and satisfies the
compatibility condition
\begin{equation}
F_{0}(x,v)=\mu ^{\theta }(x,v)\int_{n(x)\cdot v^{\prime
}>0}F_{0}(x,v^{\prime })(n(x)\cdot v^{\prime })dv^{\prime },  \notag
\end{equation}%
then $F(t,x,v)$ is continuous away from $\gamma _{0}$.
\end{theorem}

The asymptotic stability of $F_{s}$ further justifies the physical importance
of such a steady state solution. We remember that, when $\delta >0$, then $F_{s}
$ fails to be a global Maxwellian or even a local Maxwellian, and its stability analysis marks a
drastic departure from relative entropy approach (e.g. \cite{DV}). Moreover,
such an asymptotic stability plays a crucial role in our proof of
non-negativity of $F_{s}$.

\vskip .2cm
An important consequence of Theorem \ref{main2} is the Corollary below which
specializes the result to the case of a slab $-\frac{1}{2}\leq x\leq \frac{1%
}{2}$  between two parallel plates kept at temperatures $\theta (-\frac{1}{2}%
)=1-\delta $ and $\theta (\frac{1}{2})=1+\delta $.

\begin{corollary}
\label{fourierlaw0}Let $x\in \Omega =[-\frac{1}{2},\frac{1}{2}]$ and  $%
f_{s}$ be the solution to
\begin{eqnarray}
v_{1}\partial _{x}f_{s}+\frac{1}{\text{K}_{\text{n}}}Lf_{s} &=&\frac{1}{%
\text{K}_{\text{n}}}\Gamma (f_{s},f_{s}),\text{ \ \ \ \ \ \ \ }-\frac{1}{2}%
<x<\frac{1}{2},  \label{slab} \\
{{\sqrt{\mu(v)}}}f_{s}(-\frac{1}{2},v) &=&\frac{1}{2\pi \lbrack 1-\delta ]^{2}}\exp \left[ -%
\frac{|v|^{2}}{2[1-\delta ]}\right] \int_{u_{1}<0}u_{1}{\sqrt{\mu(u)}}f_{s}(-\frac{1}{2}%
,u)du,\text{ \ \ \ }v_{1}>0,  \notag \\
{\sqrt{\mu(v)}}f_{s}(\frac{1}{2},v) &=&\frac{1}{2\pi \lbrack 1+\delta ]^{2}}\exp \left[ -%
\frac{|v|^{2}}{2[1+\delta ]}\right] \int_{u_{1}>0}u_{1}{\sqrt{\mu(u)}}f_{s}(\frac{1}{2}%
,u)du,\text{ \ \ \ \ \ \ }v_{1}<0.  \notag
\end{eqnarray}%
where $v_{1}$ is the component in the direction $x$ of the velocity $v$.

Then $f_{1}$, the first order in $\delta $ correction to $\mu $ according to
the expansion of Theorem \ref{main2}, is the unique solution to
\begin{eqnarray}
&& v_{1}\partial _{x}f_{1}+\frac{1}{\text{K}_{\text{n}}}Lf_{1} =0,\quad-\frac{1}{2}<x<\frac{1}{2},  \label{slablinear} \\
&&f_{1}(-\frac{1}{2},v) =\sqrt{\mu (v)}\int_{u_{1}<0}f_{1}(-\frac{1}{2},u)%
\sqrt{\mu (u)}\{-u_{1}\}du+\frac{\mu _{1}(-\frac{1}{2},v)}{\sqrt{\mu (v)}},%
\text{ \ \ \ }v_{1}>0,  \notag \\
&&f_{1}(\frac{1}{2},v) =\sqrt{\mu (v)}\int_{u_{1}>0}f_{1}(\frac{1}{2},u)%
\sqrt{\mu (u)}\{u_{1}\}du+\frac{\mu _{1}(\frac{1}{2},v)}{\sqrt{\mu (v)}},
\text{\ \ \ \ \ \ \ \ \ \ \ \ \ }v_{1}<0.  \notag
\end{eqnarray}
\end{corollary}

In order to establish a criterion of validity of the Fourier law, let us
remember that the temperature associated to the stationary solution $F_{s}$
is given by
\begin{equation}
\theta _{s}(x)=\frac{1}{3\rho _{s}}\int_{\mathbf{R}%
^{3}}|v-u_{s}|^{2}F_{s}(x,v)dv,  \label{thetas}
\end{equation}%
where $\rho_s=\int_{\mathbf{R}^{3}}F_x(x,v)dv$ and $u_{s}(x)=\rho_s^{-1}\int_{\mathbf{R}^{3}}vF_{s}(x,v)dv%
$.  The heat flux associated to the distribution $F_{s}$
is the vector field defined as
\begin{equation}
q_{s}(x)=\frac{1}{2}\int_{\mathbf{R}%
^{3}}(v-u_{s}(x))|v-u_{s}(x)|^{2}F_{s}(x,v)dv.  \label{qs}
\end{equation}

We state the Fourier law in the following formulation: \textit{there is a positive $%
C^{1}$ function $\kappa (\theta )$, the heat conductivity, such that (\ref%
{fourierlaw1}) is valid for }$F_{s}$\textit{. }

\begin{theorem}
\label{fourierlaw} If the Fourier law holds for $F_{s}$, then $\theta
_{1}(x)\equiv \frac{1}{3}\int_{\mathbf{R}^{3}}|v-u_{1}|^{2}f_{1}$ is a
linear function over $-\frac{1}{2}\leq x\leq \frac{1}{2}$.
\end{theorem}

Available numerical simulations, Figure 1, \cite{OAS} indicate the linearity is clearly
violated  for all finite Knudsen number K$_{\text{n}}$ ($=k$ in Figure 1).
\begin{figure}[tbp]
\begin{center}
\includegraphics[
width=1.0\textwidth]{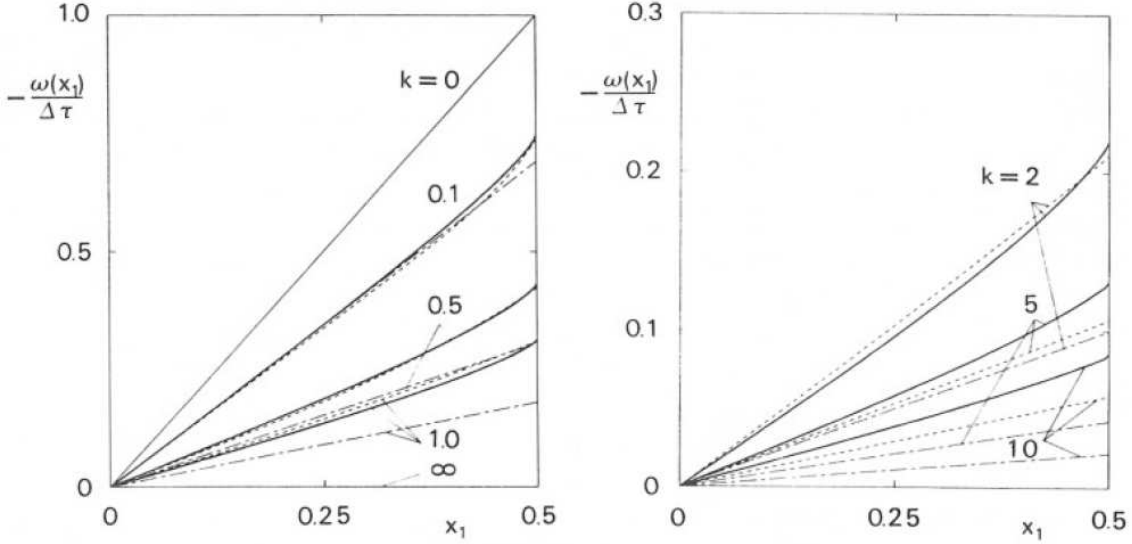}
\end{center}
\caption{Nonlinearity of $\protect\theta _{1}$ in a kinetic regime}
\end{figure}

We therefore predict that the general Fourier law (\ref{fourierlaw1}%
) is invalid and inadequate in the kinetic regime with finite Knudsen
number K$_{\text{n}}$.  It is necessary to at least solve the linear
Boltzmann equation (\ref{f1}) to capture the heat transfer in the Boltzmann
theory. 

\vskip .3cm
Without loss of generality, we may assume $\theta _{0}=1$ and K$_{\text{n}}=1
$ throughout the rest of the paper. Therefore in the rest of this paper we assume
\begin{equation}
\theta (x)=1+\delta \vartheta (x).  \label{temper}
\end{equation}

The key difficulty in the study of the steady Boltzmann equation lies in the
fact that the usual entropic estimate for $\int F\ln F$, coming from $\partial
_{t}F$, is absent. The only a-priori estimate is given by the entropy
production $\int Q(F,F)\ln F$, which is very hard to use \cite{AN}. In the
context of small perturbation of $\mu $, only the linearized dissipation
rate $\|(\mathbf{I}-\mathbf{P})f\|_{\nu }^{2}$ is controlled \footnote{We  denote by $\mathbf{P}f$ the projection of $f$ on the null space of $L$.}, and the key is
to estimate the missing hydrodynamic part $\mathbf{P}f$, in term of $(%
\mathbf{I}-\mathbf{P})f$.  This is a well-known basic question in the
Boltzmann theory. Motivated by the studies of collisions in a plasma, a new
\textit{nonlinear energy method} in high Sobolev norms was initiated in the
Boltzmann study, to estimate $\mathbf{P}f$ (\cite{G1}) in terms of $(\mathbf{I%
}-\mathbf{P})f$.  For $F=\mu +\sqrt{\mu }f$, remembering the definitions   (\ref{linBol}) of the
linearized Boltzmann operator and (\ref{nonBol}) of the quadratic nonlinear collision operator, the Boltzmann equation (\ref%
{boltz}) can be rewritten for the perturbation $f$ \ as%
\begin{equation}
\text{\ }\left\{ \partial _{t}+v\cdot \nabla _{x}+L\right\} f=\Gamma (f,f),
\label{gboltzmann}
\end{equation}%
 \textit{If} the operator $L$ \textit{were} positive definite,
then global solutions (for small $f$) could be constructed ``easily'' for (\ref%
{gboltzmann}). However, $L$ is only semi-positive
\begin{equation}
\langle Lf,f\rangle \gtrsim \|\{\mathbf{I}-\mathbf{P}\}f\|_{\nu }^{2},\text{
}  \label{pg}
\end{equation}%
where $\|\cdot \|_{\nu }$ is the $\nu$-weighted $L^{2}$ norm. The kernel
(the hydrodynamic part) is given by
\begin{equation}
\mathbf{P}f\equiv \{a_{f}(t,x)+v\cdot b_{f}(t,x)+\frac{(|v|^{2}-3)}{2}%
c_{f}(t,x)\}\sqrt{\mu }.  \label{P}
\end{equation}%
Note that we use slight different definitions of $a_{f}(t,x)$ and $c_{f}(t,x)
$ from \cite{G1} to capture the crucial total mass constraint. The so-called
`hydrodynamic part' of $f$, $\mathbf{P}f$, is the $L_{v}^{2}$-projection on
the kernel of $L$, for every given $x$.  The novelty of such energy method is
to show that $L$ is indeed \textit{positive definite} for small solutions $f$
to the nonlinear Boltzmann equation (\ref{gboltzmann}). The key macroscopic
equations for  $\xi_f=(a_f,b_{f},c_f)$ connect $\mathbf{P}f$ and $\{\mathbf{I}-\mathbf{P}\}f$
via the Boltzmann equation as in page 621 of \cite{G1}:%
\begin{equation}
\Delta \xi_{f}\ \sim \ \partial ^{2}\{\mathbf{I}-\mathbf{P}\}f\ +\ \text{%
higher order terms,}  \label{ellipticity}
\end{equation}%
where $\partial ^{2}$ denotes some second order differential operator in
\cite{G1}. Such hidden ellipticity in $H^{k}$ ($k\geq 1$) implies that the
hydrodynamic part $\mathbf{P}f$, missing from the lower bound of $L$, can be
controlled via the microscopic part $\{\mathbf{I}-\mathbf{P}\}f$, so that $L$
can control the full $\|f\|_{\nu }^{2}$ from (\ref{pg}). Such nonlinear
energy framework has led to resolutions to several  open problems in the
Boltzmann theory \cite{G1,G2,GS}.

It should be noted that all of these results deal with idealized periodic
domains, in which the solutions can remain smooth in $H^{k}$ for $k\geq 1$.
Of course, a gas is usually confined within a container, and its interaction
with the boundary plays a crucial role both from physical and mathematical view
points. Mathematically speaking, the phase boundary $\partial \Omega \times
\mathbf{R}^{3}$ is always \textit{characteristic} but not \textit{uniformly
characteristic} at the grazing set $\gamma _{0}=\{(x,v):x\in \partial \Omega
,~$and $v\cdot n(x)=0\}$.  In particular, many of the natural physical
boundary conditions create singularities in general domains (\cite{Guo08,Kim}), for which the high Sobolev estimates break down in the crucial
elliptic estimates (\ref{ellipticity}). Discontinuities are expected to be
created at the boundary, and then propagate inside a non-convex domain.
Therefore completely new tools need to be developed.

A new $L^{2}$--$L^{\infty }\ $framework is developed in \cite{Guo08} (see
\cite{G3} for a short summary of the method) to resolve such a difficulty in
the Boltzmann theory, which leads to resolution of asymptotic stability of
Maxwellians for specular reflection in an analytic and convex domain, and
for diffuse reflection (with uniform temperature!) in general domains (no
convexity is needed).  We remark that the non-convex domains occur naturally
for non-isothermal boundary (e.g. two non flat separated boundaries).
Furthermore, the solutions to the boundary problems are shown to be \textit{%
continuous} if the domain $\Omega $ is strictly convex. Different $%
L^{2}-L^{\infty}$ methods have been used in \cite{AEMN1,AEMN2,AN1, ELM1, ELM2} in particular geometries.

The new $L^{2}-L^{\infty }$ framework introduced in \cite{Guo08} has two parts:

$L^{2}$\textit{\ Positivity:} Assume the wall temperature $\theta $ is
constant in (\ref{diffbc}). It suffices to establish the following
finite-time estimate%
\begin{equation}
\int_{0}^{1}\|\mathbf{P}g(s)\|_{\nu }^{2}ds\lesssim \left\{ \int_{0}^{1}\|\{%
\mathbf{I}-\mathbf{P\}}g(s)\|_{\nu }^{2}+\text{boundary contributions}%
\right\} .  \label{pi-p}
\end{equation}%
The natural attempt is to establish $L^{2}$ estimate for $\xi_{f}$ from the
macroscopic equation (\ref{ellipticity}). However, this is very challenging
due to the fact that only $f$ has trace in the sense of Green's identity
(Lemma \ref{green}), neither $\mathbf{P}f$ nor $\{\mathbf{I}-\mathbf{P}\}f$ $%
\ $even make sense on the boundary of a general domain. Instead, the proof of (\ref{pi-p}) given in \cite{Guo08}
is based on a delicate contradiction argument because it is difficult to
estimate $\mathbf{P}f$ via $\{\mathbf{I}-\mathbf{P}\}f$ directly in a $L^{2}$
setting in the elliptic equation (\ref{ellipticity}), in the presence of
boundary conditions.  The heart of this argument, lies in an exact
computation of $\mathbf{P}f$ which leads to the contradiction. As a result,
such an indirect method fails to provide a constructive estimate of (\ref%
{pi-p}) with explicit constants.

$L^{\infty }$\textit{\ Bound: }The method of characteristics can
bootstrap the $L^{2}$ bound into a point-wise bound to close the nonlinear
estimate. Let $U(t)$ be the semigroup generated by $v\cdot\nabla_x+L$ and $G(t)$ the semigroup generated by $v\cdot\nabla_x+\nu$, with the prescribed boundary conditions.  By \textit{two iterations, }one can establish:
\begin{equation}
U(t)=G(t)+\int_{0}^{t}G(t-s_{1})KG(s_{1})ds_{1}+\int_{0}^{t}%
\int_{0}^{s_{1}}G(t-s_{1})KG(s_{1}-s)KU(s)dsds_{1}.  \label{duhamel2}
\end{equation}%
From the compactness property of $K$, the main contribution in (\ref{duhamel2})
is roughly
\begin{equation}
\int_{0}^{t}\int_{0}^{s_{1}}\int_{v^{\prime },v^{\prime \prime }\text{bounded%
}}|f(s,X_{\mathbf{cl}}(s;s_{1},X_{\mathbf{cl}}(s_{1};t,x,v),v^{\prime
}),v^{\prime \prime })|dv^{\prime }dv^{\prime \prime }dsds_{1}.  \label{k-k}
\end{equation}%
where $X_{\mathbf{cl}}(s;t,x,v)$ denotes the generalized characteristics
associated with specific boundary condition. A change of variable from $%
v^{\prime }$ to $X_{\mathbf{cl}}(\cdot )$ would transform the $v^{\prime }$
and $v^{\prime \prime }$-integration in (\ref{k-k}) into \thinspace $x$ and $%
v$ integral of $f$, which decays from the $L^{2}$ theory. The key is to
check if
\begin{equation}
\det \left\{ \frac{X_{\mathbf{cl}}(s;s_{1},X_{\mathbf{cl}}(s_{1};t,x,v),v^{%
\prime })}{dv^{\prime }}\right\} \neq 0.  \label{vidav}
\end{equation}%
is valid. Without boundary, $X_{\mathbf{cl}}(s;s_{1},X(s_{1};t,x,v),v^{%
\prime })$ is simply $x-(t-s_{1})v-(s_{1}-s)v^{\prime }$, and $\det \left\{
\frac{X_{\mathbf{cl}}(\cdot )}{dv^{\prime }}\right\} \neq 0$ most of the
time. For specular or diffuse reflections, each type of characteristic
trajectories $X_{\mathbf{cl}}$ repeatedly interact with the boundary. To
justify (\ref{vidav}), various delicate arguments were invented to overcome
different difficulties, and analytic boundary and convexity are needed for
the specular case.

Since its inception, this new $L^{2}-L^{\infty }$ approach has already led
to new results in the study of relativistic Boltzmann equation \cite{S}, in
hydrodynamic limits of the Boltzmann theory(\cite{GJ} \cite{GJJ} \cite{SS}),
in stability in the presence of a large external field (\cite{EGM} \cite{K2}%
).

Our current study of non-isothermal boundary ($\theta $ is non-uniform) is
naturally based on such a $L^{2}-L^{\infty }\,\ $framework. The main new
difficulty in contrast to \cite{Guo08}, however, is that the presence of a
non-constant temperature creates non-homogeneous terms in both the linear
Boltzmann (steady and unsteady) as well as in the boundary condition, so
that the exact computation, crucial to the $L^{2}$ estimate (\ref{pi-p}),
breaks down, and the $L^{2}-L^{\infty }$ scheme \cite{Guo08} collapses.

The main technical advance in this paper is the development of a {\it direct} (constructive) and
{ robust} approach to establish (\ref{pi-p}) in the presence of a non-uniform
temperature diffuse boundary condition (\ref{fdiffuse}). Instead of using
these macrosopic equations (e.g. (\ref{ellipticity})), whose own meaning
is doubtful in a bounded domain, we resort to the basic Green's identity
for the transport equation and choose proper test functions to recover
ellipticity estimates for $a$, $b$, $c$ and hence $\mathbf{P}f$  directly. In light of the
energy identity, $\{1-P_{\gamma }\}f$ is controlled at the boundary $\gamma
_{+}$, but not $P_{\gamma }f$ in (\ref{pgamma}). The essence of the method
is to choose a test function which can eliminate the $P_{\gamma }f$
contribution at the boundary, and to control the $a$, $b$, $c$ component of $\mathbf{P%
}f$ respectively in the bulk \textit{at the same time}. The choice of the test function
for $c$ is rather direct: we set
\begin{equation*}
\psi _{c}=(|v|^{2}-\beta _{c})\sqrt{\mu }v\cdot \nabla _{x}\phi _{c}(x),
\end{equation*}%
 for some constant $\beta _{c}$ to be determined, with $-\Delta \phi _{c}=c$ and $\phi _{c}=0$ on $\partial \Omega$.  On the other hand, the test
function for $b$ is rather delicate. In fact, two different sets of test
functions have to be constructed:%
\begin{eqnarray*}
(v_{i}^{2}-\beta _{b})\sqrt{\mu }\partial _{j}\phi _{b}^{j},\text{ \ \ \ }%
i,j &=&1,...d, \\
|v|^{2}v_{i}v_{j}\sqrt{\mu }\partial _{j}\phi _{b}^{i},\text{ \ }i &\neq &j,
\end{eqnarray*}%
with $-\Delta \phi _{b}^{j}=b_{j}$ and $\phi _{b}^{j}=0$ on $\partial \Omega
$.  In particular, a unique constant $\beta _{c}$ can be chosen
to deduce the estimates for $b$, thanks to the special structure of the
transport equation and the diffuse boundary condition. The choice of the test
function for $a$ requires special attention. It turns out that in order to
eliminate the contribution $P_{\gamma }f$, we need to choose the test function
\begin{equation*}
(|v|^{2}-\beta _{a})v\cdot \nabla _{x}\phi _{a}\sqrt{\mu },
\end{equation*}%
with $-\Delta \phi _{a}=a$ and Neumann boundary condition%
\begin{equation*}
\frac{\partial }{\partial n}\phi _{a}=0\quad\text{ on }\partial \Omega .
\end{equation*}%
This is only possible if the total mass of  $f$ is zero, i.e.,
\begin{equation}
\iint_{\Omega \times \mathbf{R}^{3}}f(x,v)\sqrt{\mu }dxdv\equiv\sqrt{2\pi}\int_{\Omega }a(x)dx=0.  \label{a=0}
\end{equation}%
This illustrates the importance of the mass constraint, which unfortunately is
not valid for the steady problem (\ref{boltzmann}) and (\ref{fdiffuse}). So we are forced to use a penalization procedure  (see below) to deal with it. The key lemma which delivers the basic estimates is Lemma \ref{steadyabc}.
Furthermore, in the dynamical case, also the time derivatives $\partial _{t}\phi _{c}$, $\partial
_{t}\phi _{b}^{j}$ and $\partial _{t}\phi _{a}$ need to be controlled in
negative Sobolev spaces. This is possible due to the special structure
of the Boltzmann equation as well as the diffuse boundary condition. Once
again, the total mass zero condition (\ref{a=0}) is essential. This is the key to prove the crucial Lemma \ref{dabc}. Even though
this new unified procedure can be viewed as a `weak version' of the
macroscopic equations (\ref{ellipticity}), the estimates we obtain via this
approach are more general. For instance, the (\ref{ellipticity}) was only
valid for dimension  {$\ge 2$}, but the new estimates are valid for any dimension.
There seems to be a very rich structure in the linear Boltzmann equation.

To bootstrap such a $L^{2}$ estimate into a $L^{\infty }$ estimate, we
define the stochastic cycles for the generalized characteristic lines
interacting with the boundary:

\begin{definition}[Stochastic Cycles]
\label{diffusecycles}Fixed any point $(t,x,v)$ with $(x,v)\notin \gamma
_{0}$, let 
$(t_{0},x_{0},\bar{v}_{0})$ $=(t,x,\bar{v})$. For $\bar{v}_{k+1}$such that $
\bar{v}_{k+1}\cdot n(x_{k+1})>0$, define the $(k+1)$-component of the
back-time cycle as
\begin{equation}
(t_{k+1},x_{k+1},\bar{v}_{k+1})=(t_{k}-t_{\mathbf{b}}(x_{k},\bar{v}_{k}),x_{%
\mathbf{b}}(x_{k},\bar{v}_{k}),\bar{v}_{k+1}).  \label{diffusecycle}
\end{equation}%
Set
\begin{eqnarray*}
X_{\mathbf{cl}}(s;t,x,\bar{v}) &=&\sum_{k}\mathbf{1}_{[t_{k+1},t_{k})}(s)%
\{x_{k}+(s-t_{k})\bar{v}_{k}\}, \\
\bar{V}_{\mathbf{cl}}(s;t,x,\bar{v}) &=&\sum_{k}\mathbf{1}%
_{[t_{k+1},t_{k})}(s)\bar{v}_{k},\text{ \ \ \ }V_{\mathbf{cl}}(s;t,x,\bar{v}%
)=(\bar{V}_{\mathbf{cl}}(s;t,x,\bar{v}),\hat{v}).
\end{eqnarray*}%
Define $\mathcal{V}_{k+1}=\{v\in \mathbf{R}^{3}\ |\ \bar{v}\cdot
n(x_{k+1})>0\}$, and let the iterated integral for $k\geq 2$ be defined as
\begin{equation}
\int_{\Pi _{j=1}^{k-1}\mathcal{V}_{j}}\dots\Pi _{j=1}^{k-1}d\sigma _{j}\equiv
\int_{\mathcal{V}_{1}}\dots\left\{ \int_{\mathcal{V}_{k-1}}d\sigma
_{k-1}\right\} d\sigma _{1},  \label{sigma}
\end{equation}%
where $d\sigma _{j}=\mu (v)(n(x_{j})\cdot \bar{v})dv$ is a probability
measure.
\end{definition}

We note that the $\bar{v}_{j}$'s ($j=1,2,...$) are all independent
variables, and $t_{k},x_{k}$ depend on $t_{j},x_{j},\bar{v}_{j}$ for $j\leq
k-1$. However, the phase space $\mathcal{V}_{j}$ implicitly depends on $(t,x,%
\bar{v},\bar{v}_{1},\bar{v}_{2},...\bar{v}_{j-1})$. Our method is to use the
Vidav's two iterations argument(\cite{Vi}) and estimate the $L^{\infty }$-norm
along these stochastic cycles with corresponding phase spaces $\Pi
_{j=1}^{k-1}\mathcal{V}_{j}$. The key is to estimate measures of various
sets in $\Pi _{j=1}^{k-1}\mathcal{V}_{j}$ in Lemma \ref{k}. We designed an
abstract and unified iteration (\ref{iteration}), which is suitable for both steady
and unsteady cases. New precise estimate (\ref{rho12}) of non-homogeneous
terms resulting from the non-constant temperature are obtained in Prop. \ref%
{linfty}. Based on a delicate change of variables in Lemma \ref{tbsmooth},
such an estimate is crucial in the proof of formation of singularity for a
non-convex domain in Theorem \ref{main1}.

In order to keep the mass zero condition and to start iterating scheme, it
is essential to introduce a penalization $\varepsilon $ to solve the problem%
\begin{equation}
\varepsilon f^{l+1}+v\cdot \nabla _{x}f^{l+1}+Lf^{l+1}=\Gamma (f^{l},f^{l}).
\notag
\end{equation}
The presence of $\varepsilon $ ensures the critical zero mass condition (\ref%
{a=0}).   It is important to note that our $L^{\infty }$ estimate are
intertwined with $L^{2}$ at every step of the approximations, which ensures
the preservation of the continuity for a convex domain. The continuity
properties of our final solutions follows from the $L^{\infty }$ limit at
every step. \ Moreover, the proof of continuity away from the singular set $%
\mathfrak{D}$ in (\ref{Dset}) in a general domain is consequence of a
delicate result for $Q_{\text{gain}}$ in \cite{Kim}.

To illustrate the subtle nature of our construction, we remark that for
the natural positivi\-ty-preserving scheme:
\begin{eqnarray}
\varepsilon f^{l+1}+v\cdot \nabla _{x}f^{l+1}+\nu f^{l+1}-Kf^{l} &=&\Gamma _{%
\text{gain}}(f^{l},f^{l})-\Gamma _{\text{loss}}(f^{l},f^{l+1})
\label{positivescheme} \\
F^{l+1} &=&\mu _{\theta }(x,v)\int_{n(x)\cdot v^{\prime
}>0}F^{l}\{n(x)\cdot v^{\prime }\}dv^{\prime },  \notag
\end{eqnarray}%
we are unable to prove the convergence, due to breakdown of $Lf^{l+1}$,
whence the mass zero constraint (\ref{a=0}) fails to be satisfied. Consequently, we are unable to
prove $F_{s}\geq 0$ in our construction. Such a positivity is only proven
via the dynamical asymptotic stability of $F_{s}$, in which the initial
positivity plus the choice of a small time interval are crucial to
guarantee the convergence of  the analog of (\ref{positivescheme}) in the dynamical setting.

Our estimates are robust and allow us to expand our steady state $F_{s}$ in
terms of $\delta $, the magnitude of the perturbation. This leads to the
first order precise characterization of $F_{s}$ by $f_{1}$, which satisfies (%
\ref{f1}).

It should be pointed out the our rather complete study of the
non-isothermal boundary for the Boltzmann theory for $\delta \ll 1$ forms a
mathematically solid foundation for influential work in applied physics and
engineering such as in \cite{So1}\cite{So2}, where the existence of such
steady solutions to the Boltzmann equation is a starting point, but without
mathematical justification. We expect our solutions as well as our new
estimates would lead to many new developments along this direction.

\vskip .2cm

The plan of the paper is the following. In next section we present some background material and in particular a version of the Ukai Trace Theorem and the Green identity as well as a new $L^1$ estimate at the boundary in Lemma 2.3. Section 3 is devoted to the construction of $L^2$ solutions to the stationary linearized problem. In particular, we prove Lemma \ref{steadyabc} which provides the basic estimate of the $L^2$ norms of  $\mathbf{P}f$ in terms of the $L^2$ norm of $(\mathbf{I}-\mathbf{P})f$.
In Section 4, after introducing an abstract iteration scheme suitable for proving $L^\infty$ bounds, we prove, in Proposition \ref{linfty}, the existence of the solution to the linearized problem in $L^\infty$ and related bounds. In Section 5 we combine the results of the previous sections to construct the stationary solution to the Boltzmann equation and discuss its regularity properties. In particular, we give the proof of the $\delta$-expansion and use it to establish a necessary condition for the validity of the Fourier law. In Section 6 we extend the $L^2$ estimates to the time dependent problem. Section 7 contains the extension of the $L^\infty$ estimates to the time dependent problem and the proof of the exponential asymptotic stability of the stationary solution. From this we then obtain its positivity.

\bigskip
\section{Background}

In this section we state basic preliminaries. First we shall clarify the notations of functional spaces and norms: we use $\| \cdot \|_p$ for both of the $L^p(\Omega\times\mathbf{R}^3)$ norm and the $L^p(\Omega)$ norm, and $( \, \cdot\,,\, \cdot \, )$ for the standard $L^2(\Omega\times\mathbf{R}^3)$ inner product. Moreover we denote $\|\cdot \|_{\nu}\equiv \|\nu^{1/2}\cdot \|_2$ and $\|f\|_{H^k}= \|f\|_{2}+ \sum_{i=1}^{k}\|\nabla_x^i f \|_{2}$.  For the phase boundary integration, we define $d\gamma = |n(x)\cdot v|dS(x)dv$ where $dS(x)$ is the surface measure and define $|f|_p^p =  \int_{\gamma} |f(x,v)|^p d\gamma \equiv\int_{\gamma} |f(x,v)|^p  $ and the corresponding space as $L^p(\partial\Omega\times\mathbf{R}^3;d\gamma)=L^p(\partial\Omega\times\mathbf{R}^3)$.  Further $|f|_{p,\pm}= |f \mathbf{1}_{\gamma_{\pm}}|_p$ and $|f|_{\infty,\pm}= |f \mathbf{1}_{\gamma_{\pm}}|_{\infty}$.  We also use $|f|_p^p = \int_{\partial\Omega} |f(x)|^p dS(x)\equiv\int_{\partial\Omega} |f(x)|^p $.  Denote $f_{\pm}=f_{\gamma_{\pm}}$.  Recall that $x\in\overline{\Omega}\subset \mathbf{R}^d$ for $d=1,2,3$, and $v\in \mathbf{R}^3$.  For $\mathcal{A}\in \mathbf{R}^d$, the notation $v\cdot \mathcal{A}$ means $\sum_{i=1}^d v_i \mathcal{A}_i$.

In the next lemma, we prove that the trace of $f$ is well-defined locally for a certain class of $f$:
\begin{lemma}[Ukai Trace Theorem]\label{ukai}
Define
\begin{equation}
\gamma^{\varepsilon}  \ \equiv \ \{ (x,v)\in\gamma : |n(x)\cdot v| \geq \varepsilon, \ |v|\leq \frac{1}{\varepsilon}\}.\label{ukaitrace}
\end{equation}
Then
\begin{eqnarray}
| f\mathbf{1}_{\gamma^{\varepsilon}}  |_1 &\lesssim_{\varepsilon,\Omega}&  \| f  \|_{1}+\|\{v\cdot\nabla_x\} f \|_1  , \ \ \ \ \ \ \ \ \ \ \ \ \ \ \ \ \ \    (\text{steady})\nonumber\\
\int_s^t| f\mathbf{1}_{\gamma^{\varepsilon}}(\tau)  |_1 d\tau  &\lesssim_{\varepsilon,\Omega}&  \int_s^t\| f(\tau) \|_{1} d\tau + \int_s^t\|\{\partial_t + v\cdot\nabla_x\} f(\tau)\|_{1}d\tau  ,\nonumber\\ && \ \ \  \ \ \   \ \ \ \ \ \ \ \ \ \ \ \ \ \  \ \ \ \ \ \ \ \ \ \ \ \ \ \ \ \   \ \ \ \ \ \ \ \  \  (\text{dynamic})\nonumber
\end{eqnarray}
for all $0\leq s\leq t$.
\end{lemma}
\begin{proof}
Recall the notation (\ref{v}) and the definition of backward exit time $t_{\mathbf{b}}(x,v)=t_{\mathbf{b}}(x,\bar{v})$ and $x_{\mathbf{b}}(x,v)=x_{\mathbf{b}}(x,\bar{v})=x-t_{\mathbf{b}}(x,\bar{v})\bar{v}\in\partial\Omega $ in (\ref{backexit}).

In the steady case, from \cite{CIP}, page 247, we have the following identity:
\begin{eqnarray}
\iint_{\Omega\times\mathbf{R}^3} f(x,v) dxdv &=& \int_{\gamma_+} \int_{t_{\mathbf{b}}(x,v)}^0 f(x-s\bar{v},v) ds d\gamma\notag\\
&&+\int_{\gamma_-} \int^{t_{\mathbf{b}}(x,-v)}_0 f(x+s\bar{v},v) ds d\gamma.\notag
\end{eqnarray}
Assume $\|f\|_1, \ \|\{v\cdot\nabla_x\} f \|_1 <  \infty$.  Then $f(x+s\bar{v},v)$ is an absolutely continuous function of $s$ for fixed $(x,v)$, so that, by the fundamental theorem of calculus
\begin{eqnarray*}
f(x,v) &=& f(x-s^{\prime}\bar{v} ,v) + \int_{-s^{\prime}}^{0}  v \cdot \nabla_x f(x+\tau \bar{v},v)d\tau\nonumber \\
 && \ \ \ \ \ \ \ \ \ \ \ \ \ \ \ \ \ \ \ \ \ \ \ \ \ \ \ \ \ \ \ \ \ \ \ \text{for } \ (s^{\prime},x,v)\in[0,t_{\mathbf{b}}(x,v)]\times\gamma_+,\\
f(x,v) &=& f(x+s^{\prime}\bar{v} ,v) + \int_{0}^{s^{\prime}}  v \cdot \nabla_x f(x+\tau \bar{v},v)d\tau\nonumber \\
 && \ \ \ \ \ \ \ \ \ \ \ \ \ \ \ \ \ \ \ \ \ \ \ \ \ \ \ \ \ \ \ \ \ \ \ \text{for } \ (s^{\prime},x,v)\in[0,t_{\mathbf{b}}(x,-v)]\times\gamma_-,
\end{eqnarray*}
and therefore, for both cases,
\begin{eqnarray}
|f(x,v)| \leq    |f(x-s^{\prime}\bar{v},v)|  +   \int_{ t_{\mathbf{b}}(x,v)}^0 | v \cdot \nabla_x f(x-\tau \bar{v},v)| d\tau  ,\label{ukai_1}\\
|f(x,v)|  \leq    |f(x+s^{\prime}\bar{v},v)|  +   \int^{t_{\mathbf{b}}(x,-v)}_0 | v \cdot \nabla_x f(x+\tau\bar{ v},v)| d\tau .\label{ukai_2}
\end{eqnarray}
On the other hand, for $x \in \partial \Omega$ assume that $\partial\Omega$ is locally parameterized by $\xi : \{y\in \mathbf{R}^d : |x-y|< \delta\} \rightarrow \mathbf{R}$ so that
\begin{eqnarray*}
\sup_{\substack{  y \in \partial\Omega \\|x-y|< \delta}} \frac{|(x-y)\cdot n(x)|}{|x-y|^2} \ \leq \ \max_{\substack{  y \in \partial\Omega \\|x-y|< \delta}}  |\nabla_x^2 \xi(y)|.
\end{eqnarray*}
By the compactness of $\Omega$ (and $\partial\Omega$), uniformly in $x$, we have $|(x-y)\cdot n(x)|\leq C_{\Omega }|x-y|^{2}$ for
all $y\in \partial \Omega $.  Taking inner product of $%
x-x_{\mathbf{b}}(x,v)=t_{\mathbf{b}}(x,v)\bar{v}$ with $n(x)$, we get
\begin{equation*}
t_{\mathbf{b}}(x,v)|v\cdot n(x )|=|( x-x_{\mathbf{b}}(x,v) )\cdot n(x )|\le C_{\Omega }|x-x_{\mathbf{b}}(x,v)|^{2}=
C_{\Omega}|v|^{2}|t_{\mathbf{b}}(x,v)|^{2}.
\end{equation*}%
Therefore we deduce
\begin{equation}
t_{\mathbf{b}}(x,v) \geq \frac{|n(x)\cdot v|}{C_{\Omega}|v|^2}.
\label{tlower}
\end{equation}
If $(x,v)\in\gamma^{\varepsilon} \cap \gamma_{+}$ or $(x,v)\in\gamma^{\varepsilon} \cap \gamma_-$ then, by the definition of $\gamma^{\varepsilon}$ and (40) in \cite{Guo08},
\begin{equation}
 \min\{ t_{\mathbf{b}}(x,v), \ t_{\mathbf{b}}(x,-v)\} \geq \frac{|n(x)\cdot v|}{C_{\Omega}|v|^2}\geq C_{\Omega}^{-1} \varepsilon^3.\nonumber
\end{equation}
First we integrate (\ref{ukai_1}), for $(x,v)\in\gamma^{\varepsilon} \cap \gamma_+$, and then for $s^{\prime}\in [0,t_{\mathbf{b}}(x,v)]$.  Similarly integrate (\ref{ukai_2}), for $(x,v)\in\gamma^{\varepsilon} \cap\gamma_-$, and then for $s^{\prime}\in [0,t_{\mathbf{b}}(x,-v)]$ to have
\begin{eqnarray}
C_{\Omega}^{-1}\varepsilon^3 |  f \mathbf{1}_{\gamma^{\varepsilon}}  |_{1} \leq C \left\{ \| f \|_1 +   \| v\cdot\nabla_x f \|_1
\right\}. \nonumber
\end{eqnarray}

In the dynamic case, from \cite{CIP} in page 247, we have the following identity :
\begin{eqnarray}
&&\hskip -.7cm\iiint_{[0,t]\times\Omega\times\mathbf{R}^3} f(s,x,v) ds dx dv = \int_{\gamma_+} \int_{\max\{0,t-t_{\mathbf{b}}(x,v)\}}^{t} f(s,x-(t-s)\bar{v},v)  ds d\gamma\nonumber\\
&&\hskip 3.5cm+ \int_{\gamma_-} \int_{0}^{\min\{t,t_{\mathbf{b}}(x,-v)\}}
f(s,x+s\bar{v},v)  ds d\gamma. \ \ \ \ \ \ \ \ \ \ \ \ \label{ukai_int}
\end{eqnarray}
Assume again $\|f\|_1, \ \|\{\partial_t+ v\cdot\nabla_x\} f \|_1 <  \infty$.  Tthen $f(s,x+s\bar{v},v)$ is absolutely continuous of $s$ for fixed $x,v$, so that, by the fundamental theorem of calculus
\begin{eqnarray}
f(s,x,v)&=& f(s-s^{\prime},x-s^{\prime}\bar{v},v) + \int^s_{s-s^{\prime}} v\cdot\nabla_x f (\tau,x-(s-\tau)\bar{v},v ) d\tau \ \ \nonumber\\
&& \ \ \ \ \ \ \ \ \ \ \ \ \ \ \ \      \text{for} \ (s,s^{\prime},x,v)\in[0,t]\times [\max\{0,s-t_{\mathbf{b}}(x,v)\},s]\times\gamma_+ ,\nonumber\\
f(s,x,v)&=& f(s+s^{\prime},x+ s^{\prime} \bar{v},v) + \int^{s+s^{\prime}}_{s} v\cdot\nabla_x f (\tau,x+\tau \bar{v},v ) d\tau \ \ \nonumber\\
&& \ \ \ \ \ \ \ \ \ \ \ \ \ \ \ \     \text{for} \ (s,s^{\prime},x,v)\in[0,t]\times [0,\min\{t-s, t_{\mathbf{b}}(x,-v) \}]\times\gamma_-, \nonumber
\end{eqnarray}
and hence, for both cases,
\begin{eqnarray}
|f(s,x,v)|  &\leq&   |f(s-s^{\prime},x-s^{\prime}\bar{v},v)|  +  \int_{s-s^{\prime}}^s |v\cdot \nabla_x f(\tau,x-(s-\tau)\bar{v},v)|  d\tau
 ,\nonumber\\
|f(s,x,v)|  &\leq&   |f(s+s^{\prime},x+s^{\prime}\bar{v},v)|  +  \int^{s+s^{\prime}}_s |v\cdot \nabla_x f(\tau,x+\tau \bar{v},v)|  d\tau .\nonumber
\end{eqnarray}
Then the rest part of proof is exactly the same as in the steady case.
\end{proof}

\vskip .3cm
The following Green's identities are important in this paper.
\begin{lemma}{\label{green}}For the steady case, assume that  $f(x,v), \ g(x,v)\in L^2(\Omega\times\mathbf{R}^3)$,  $v \cdot \nabla_x f, \ v \cdot \nabla_x g \in L^2(\Omega\times\mathbf{R}^3)$ and $f_{\gamma}, g_{\gamma}\in L^2(\partial\Omega\times\mathbf{R}^3)$.  Then
\begin{eqnarray}
\iint_{\Omega\times\mathbf{R}^3} \{ v\cdot \nabla_x f \} g + \{v\cdot \nabla_x g\} f \ dv dx= \int_{\gamma} f g d\gamma.\label{steadyGreen}
\end{eqnarray}
For the dynamic case, assume that  $f(t,x,v), \ g(t,x,v)\in L^{\infty}([0,T];L^2(\Omega\times\mathbf{R}^3))$, $\partial_t f+ v \cdot \nabla_x f$, $\partial_t g+ v \cdot \nabla_x g \in L^2([0,T] \times \Omega\times\mathbf{R}^3)$ and $f_{\gamma}, g_{\gamma}\in L^2([0,T]\times\partial\Omega\times\mathbf{R}^3)$.  Then, for almost all $t,s\in [0,T]$,
\begin{eqnarray}
&& \ \ \int_s^t \iint_{\Omega\times\mathbf{R}^3} \{ \partial_t f + v\cdot \nabla_x f\} g  \ dv dx + \int_s^t \iint_{\Omega\times\mathbf{R}^3} \{ \partial_t g + v\cdot \nabla_x g\} f \ dv dx \notag\\
& &= \iint_{\Omega\times\mathbf{R}^3} f(t)g(t) dv dx - \iint_{\Omega\times\mathbf{R}^3} f(s)g(s) dv dx + \int_s^t \int_{\gamma} fg d\gamma d\tau.\label{dynamicGreen}
\end{eqnarray}
\end{lemma}
\begin{proof}
See the proof in Chapter 9 of \cite{CIP}.
\end{proof}

\vskip .3cm
\begin{lemma}\label{tbsmooth}
Let $\Omega \subset \mathbf{R}^d$ and recall the notation (\ref{v}). If
\begin{eqnarray}
\bar{v}\cdot n(x_{\mathbf{b}}(x,\bar{v})) <0, \notag
\end{eqnarray}
then $t_{\mathbf{b}}(x,\bar{v})$ and  $x_{\mathbf{b}}(x,\bar{v})$ are smooth functions of $(x,\bar{v})$ so that
\begin{eqnarray}
&&\nabla_x t_{\mathbf{b}} = \frac{n(x_{\mathbf{b}})}{\bar{v}\cdot n(x_{\mathbf{b}})}, \ \ \ \ \ \ \ \  \ \ \nabla_{\bar{v}} t_{\mathbf{b}}= -\frac{t_{\mathbf{b}} n(x_{\mathbf{b}})}{\bar{v}\cdot n(x_{\mathbf{b}})},\label{tb}\\
&&\nabla_x x_{\mathbf{b}}=I - \nabla_x t_{\mathbf{b}}\otimes v, \ \ \ \ \nabla_{\bar{v}} x_{\mathbf{b}}=-t_{\mathbf{b}}I - \nabla_{\bar{v}} t_{\mathbf{b}}\otimes v.
\notag
\end{eqnarray}
For $d=2,3$, if $x\in\partial\Omega$
\begin{eqnarray}
\int_{\mathbf{S}^{d-1}} |f(x-t_{\mathbf{b}}(x,\mathfrak{u})\mathfrak{u}| |n(x)\cdot \mathfrak{u}| d\mathfrak{u}
\lesssim \int_{\partial\Omega} |f(y)|dS(y),\label{change1}
\end{eqnarray}
and if $ x\in \Omega$ and $\{y \in \mathbf{R}^d : |x-y|< \epsilon|\} \cap \partial\Omega = \emptyset$
\begin{eqnarray}
\int_{\mathbf{S}^{d-1}} |f(x-t_{\mathbf{b}}(x,\mathfrak{u})\mathfrak{u})| d\mathfrak{u}
\lesssim_{\epsilon} \int_{\partial\Omega} |f(y)|dS(y),\label{change2}
\end{eqnarray}
\end{lemma}
\begin{proof}
Assume that $\partial\Omega$ is locally parameterized by $\xi: \mathbf{R}^d \rightarrow \mathbf{R}$ and $\Omega$ is locally $\{x\in \mathbf{R}^d : \xi(x)<0\}$ so that
\[
\xi (x_{\mathbf{b}}(t,x,\bar{v}))=\xi (x-t_{\mathbf{b}}(x,\bar{v})\bar{v}) =0.
\]%
By the
implicit function theorem, taking $x_{j}$ and $\bar{v}_{j}$ derivative
respectively,
\begin{eqnarray*}
\partial _{j}\xi -\sum_{k=1}^{d}\partial _{k}\xi \frac{\partial t_{\mathbf{b}}}{%
\partial x_{j}}\bar{v}_{k} =0,& \text{ so that }& \frac{\partial t_{\mathbf{b}}}{%
\partial x_{j}}=\frac{n_{j}(x_{\mathbf{b}})}{\bar{v}\cdot n(x_{\mathbf{b}}) }, \\
\partial _{j}\xi t_{\mathbf{b}}+\sum_{k=1}^{d}\partial _{k}\xi \frac{\partial t_{\mathbf{b}}}{%
\partial \bar{v}_{j}}\bar{v}_{k}  = 0,& \text{ so that }& \frac{\partial t_{\mathbf{b}}}{%
\partial \bar v_{j}}=-t_{\mathbf{b}}\frac{n_{j}(x_{\mathbf{b}})}{\bar{v}\cdot n(x_{\mathbf{b}}) }.
\end{eqnarray*}%
We then have
\begin{eqnarray*}
\frac{\partial (x_{\mathbf{b}})_i}{\partial x_{j}} = \delta _{ij}-\frac{\bar{v}
_{i}n_{j}(x_{\mathbf{b}})}{\bar{v}\cdot n(x_{\mathbf{b}})}, \ \ \ \ \
\frac{\partial (x_{\mathbf{b}})_{i}}{\partial \bar{v}_{j}}  = -t_{\mathbf{b}}\left\{\delta _{ij}-%
\frac{\bar{v}_{i}n_{j}(x_{\mathbf{b}})}{ \bar{v}\cdot n(x_{\mathbf{b}})}\right\}.
\end{eqnarray*}%

Now we prove (\ref{change1}) for $d=3$. Without loss of generality we may assume that $\partial_3 \xi(x_{\mathbf{b}})\neq 0$. Using the spherical coordinates, $\mathfrak{u}=(\sin\theta \cos\phi, \sin\theta \sin\phi, \cos\theta)\in\mathbf{S}^{2}$ and
\begin{eqnarray*}
\xi(x_1-t_{\mathbf{b}}(x,\mathfrak{u})\sin\theta\cos\phi, x_2 -t_{\mathbf{b}}(x,\mathfrak{u})\sin\theta\cos\phi, x_3 -t_{\mathbf{b}}(x,\mathfrak{u})\cos\theta) = 0.
\end{eqnarray*}
By the implicit function theorem we compute
\begin{eqnarray*}
&&\frac{\partial }{\partial\theta}t_{\mathbf{b}}(x,\mathfrak{u}(\theta,\phi)) = - t_{\mathbf{b}}(x,\mathfrak{u}(\theta,\phi))\frac{-\sin\theta + \partial_1 \xi(x_{\mathbf{b}})\cos\theta \cos\phi + \partial_2\xi(x_{\mathbf{b}}) \cos\theta\sin\phi}{\mathfrak{u}(\theta,\phi)\cdot n(x_{\mathbf{b}}) |\nabla_x \xi(x_{\mathbf{b}})|},\\
&&\frac{\partial }{\partial\phi}t_{\mathbf{b}}(x,\mathfrak{u}(\theta,\phi)) = t_{\mathbf{b}}(x,\mathfrak{u}(\theta,\phi))\frac{\partial_1 \xi(x_{\mathbf{b}}) \sin\theta \sin \phi -\partial_2 \xi(x_{\mathbf{b}}) \sin\theta \cos\phi}{\mathfrak{u}(\theta,\phi)\cdot n(x_{\mathbf{b}}) |\nabla_x \xi(x_{\mathbf{b}})|}.
\end{eqnarray*}
Since the Jacobian matrix is
\begin{multline*}
\text{Jac}\left\{\frac{\partial ((x_{\mathbf{b}})_1,(x_{\mathbf{b}})_2)}{\partial(\theta,\phi)}\right\} \\=\left(\begin{array}{ccc}  -\frac{\partial t_{\mathbf{b}}}{\partial\theta}\sin\theta\cos\phi -t_{\mathbf{b}}\cos\theta\cos\phi & -\frac{\partial t_{\mathbf{b}}}{\partial\phi} \sin\theta \cos\phi + t_{\mathbf{b}}\sin\theta \sin\phi \\ -\frac{\partial t_{\mathbf{b}}}{\partial \theta} \sin\theta \sin\phi -t_{\mathbf{b}} \cos\theta \sin\phi & -\frac{\partial t_{\mathbf{b}}}{\partial \phi} \sin\theta \sin \phi - t_{\mathbf{b}} \sin\theta \cos\phi \end{array}\right),
\end{multline*}
we have
\begin{eqnarray}
\det \left(\frac{\partial((x_{\mathbf{b}})_1,(x_{\mathbf{b}})_2)}{\partial(\theta,\phi)}\right) \geq \frac{(t_{\mathbf{b}})^2 \sin\theta|\partial_{3}\xi(x_{\mathbf{b}})|}{|n(x_{\mathbf{b}})\cdot \mathfrak{u}| |\nabla\xi(x_{\mathbf{b}})|},\label{det1}
\end{eqnarray}
and hence
\begin{eqnarray}
&&\int_{  \mathfrak{u} \in \mathbf{S}^2 ,  \mathfrak{u}^{\prime}\sim \mathfrak{u}   } |f(x-t_{\mathbf{b}}(x,\mathfrak{u}^{\prime})\mathfrak{u}^{\prime})| |n(x)\cdot \mathfrak{u}^{\prime}| d\mathfrak{u}^{\prime} \notag\\
&\leq& \iint_{\substack{\theta^{\prime} \in [0,\pi), \phi^{\prime} \in [0,2\pi)\\
(\theta^{\prime}, \phi^{\prime})\sim (\theta,\phi)
}} |f(x-t_{\mathbf{b}}(x,\mathfrak{u}^{\prime})\mathfrak{u}^{\prime})|
 |n(x)\cdot \mathfrak{u}^{\prime}|
 \sin\theta^{\prime} d\phi^{\prime} d\theta^{\prime},\label{xb}
 \end{eqnarray}
where $\mathfrak{u}^{\prime} = (\sin\theta^{\prime} \cos\phi^{\prime}, \sin\theta^{\prime} \sin\phi^{\prime}, \cos\theta^{\prime})$. Now we apply a change of variables
\begin{eqnarray}
(\theta^{\prime},\phi^{\prime}) \mapsto ((x_{\mathbf{b}}(x,\mathfrak{u}^{\prime}))_1, (x_{\mathbf{b}}(x,\mathfrak{u}^{\prime}))_2)
\equiv ((x_{\mathbf{b}}^{\prime})_1,(x_{\mathbf{b}}^{\prime})_2),
\end{eqnarray}
and use (\ref{det1}) to further bound (\ref{xb}) as
 \begin{eqnarray}
 \iint_{\substack{(x_{\mathbf{b}}^{\prime})_1 \sim (x_{\mathbf{b}})_1\\ (x_{\mathbf{b}}^{\prime})_2 \sim (x_{\mathbf{b}})_2 } }|f(x_{\mathbf{b}}^{\prime})|
 \frac{|n(x)\cdot \mathfrak{u}(x_{\mathbf{b}}^{\prime})| |n(x_{\mathbf{b}}^{\prime})\cdot \mathfrak{u}(x_{\mathbf{b}}^{\prime})|}{|t_{\mathbf{b}}(x,\mathfrak{u}(x_{\mathbf{b}}^{\prime}))|^2}
\frac{|\nabla \xi(x_{\mathbf{b}}^{\prime})|}{|\partial_3 \xi(x_{\mathbf{b}}^{\prime})|}d(x_{\mathbf{b}}^{\prime})_1 d(x_{\mathbf{b}}^{\prime})_2.\label{xb1}
\end{eqnarray}
Notice the surface measure of $\partial\Omega$ is $dS=\displaystyle{\frac{|\nabla \xi(x_{\mathbf{b}}^{\prime})|}{|\partial_3 \xi(x_{\mathbf{b}}^{\prime})|}}d(x_{\mathbf{b}}^{\prime})_1 d(x_{\mathbf{b}}^{\prime})_2$ if $\partial_3 \xi(x_{\mathbf{b}}^{\prime}) \neq 0$.
Denote $\mathfrak{u}(x_{\mathbf{b}}^{\prime})\equiv \displaystyle{\frac{x-x_{\mathbf{b}}^{\prime}}{|x-x_{\mathbf{b}}^{\prime}|}}$ and we use (\ref{tlower}) as
\begin{eqnarray*}
t_{\mathbf{b}}(x,\mathfrak{u}(x_{\mathbf{b}}^{\prime})) \gtrsim_{\Omega} |n(x)\cdot \mathfrak{u}(x_{\mathbf{b}}^{\prime}) |, \ \ \
t_{\mathbf{b}}(x,\mathfrak{u}(x_{\mathbf{b}}^{\prime})) \gtrsim_{\Omega} |n(x_{\mathbf{b}}^{\prime})\cdot \mathfrak{u}(x_{\mathbf{b}}^{\prime}) |,
\end{eqnarray*}
to bound (\ref{xb1}) by
\begin{eqnarray*}
\iint_{y\in\mathbf{S}^2, y\sim x_{\mathbf{b}}} |f(y)| dS(y).
\end{eqnarray*}

For $d=2$, using $\mathfrak{u}=(\cos\theta,\sin\theta)$ if $\partial_2 \xi(x_{\mathbf{b}})\neq 0$ we can compute
\begin{eqnarray*}
 {\frac{\partial}{\partial \theta}}{t_{\mathbf{b}}(x,\mathfrak{u}(\theta))} = t_{\mathbf{b}}(x, \mathfrak{u}(\theta))\frac{ n(x_{\mathbf{b}}) \cdot (\sin\theta,-\cos\theta)}{n(x_{\mathbf{b}})\cdot \mathfrak{u}(\theta)},
\end{eqnarray*}
and 
\begin{eqnarray}
\det \left( \frac{\partial (x_{\mathbf{b}})_1}{\partial\theta} \right) \geq \frac{t_{\mathbf{b}}(x,\mathfrak{u}(\theta))|\partial_2 \xi(x_{\mathbf{b}})|}{|n(x_{\mathbf{b}})\cdot \mathfrak{u}(\theta)| |\nabla \xi(x_{\mathbf{b}})|}.\label{det2}
\end{eqnarray}
The rest of proof of (\ref{change1}) is same (even simpler!) as the $d=3$ case.

For (\ref{change2}), since there is a lower bound for $t_{\mathbf{b}}(x,\mathfrak{u})\geq \epsilon$ for $\{|x-y|< \epsilon\}\cap \partial\Omega = \emptyset$ and $\mathfrak{u}\in \mathbf{S}^{d-1}$, using (\ref{det1}) and (\ref{det2}), it is easy to prove (\ref{change2}).
\end{proof}

\bigskip
\section{$L^{2}$ Estimate}
The main purpose of this section is to prove the following:
\begin{proposition}
\label{linearl2} Assume
\begin{equation}
\iint_{\Omega\times\mathbf{R}^3}g(x,v) \sqrt{\mu }\ dxdv=0 , \ \ \ \int_{\gamma _{-}}r\sqrt{\mu }d\gamma=0.
\label{constraint}
\end{equation}%
Then there exists a unique solution to
\begin{equation}
v\cdot \nabla _{x}f+L  f=g , \ \ \ f_{  {-} }=P_{\gamma }f+r,  \label{linearf}
\end{equation}%
such that  $\iint_{\Omega\times\mathbf{R}^3} f\sqrt{\mu}\ dxdv=0$ and
\begin{equation*}
\|f\|_{\nu }+|f|_{2}\lesssim \|g\|_{2}+|r|_{2}.
\end{equation*}
\end{proposition}

\noindent For the proof of Proposition \ref{linearl2} we need several lemmas. We start with the simple transport equation with a penalization term:
\begin{lemma}
\label{incoming}For any $\varepsilon >0$, there exists a unique solution to
\begin{equation*}
\varepsilon f+v\cdot \nabla _{x}f=g,\text{ \ \ \  }f_{{-}}=r,
\end{equation*}%
so that
\begin{eqnarray*}
\|f\|_{2}+|f |_{2} &\lesssim_{\varepsilon }&\|g\|_{2}+|r|_{2}, \\
\|\langle v \rangle^{\beta}e^{\zeta|v|^2} f\|_{\infty }+|\langle v \rangle^{\beta}e^{\zeta|v|^2}  f |_{\infty } &\lesssim_{\varepsilon} &\| \langle v \rangle^{\beta}e^{\zeta|v|^2}  g\|_{\infty
}+| \langle v \rangle^{\beta}e^{\zeta|v|^2}  r|_{\infty },
\end{eqnarray*}
for all $\beta\geq 0 , \zeta\geq0$.  Moreover if $g$ and $r$ are continuous away from the grazing set $\gamma _{0}$, then $f$ is continuous away from $\mathfrak{D}$.  In particular, if $ \ \Omega   $ is convex then $\mathfrak{D}=\gamma_0$.
\end{lemma}

\begin{proof}
The existence of $f$ and $L^{\infty }$-bound follow from  integration along the
characteristic lines of $\frac{dx}{ds}=\bar{v} \in \mathbf{R}^d$, and $\frac{dv}{ds}=0$ (Recall (\ref{v}) for notations). More precisely, setting $h(x,v)=\langle v \rangle^{\beta} e^{\zeta |v|^2} f(x,v)$, the integrated form of the equation for $h$ is:
\begin{eqnarray*}
 h(x,v) &=& \mathbf{1}_{t> t_{\mathbf{b}}(x,v)} \Big\{ h(x-t\bar{v},v)e^{-\varepsilon t}\\ & & \ \ \ \  \ \ \ \ \  \ \ \ \ \ +\int_0^t \langle v \rangle^{\beta} e^{\zeta |v|^2}g(x-(t-s)\bar{v},v) e^{-\varepsilon(t-s)} ds\Big\}\\
  &&+\mathbf{1}_{t\leq t_{\mathbf{b}}(x,v)} \Big\{ \langle v \rangle^{\beta} e^{\zeta |v|^2}r(x-t\bar{v},v)e^{-\varepsilon t}\\ & & \ \ \ \ \ \ \ \ \ \  \ \ \ \ \  + \int_{t-t_{\mathbf{b}}(x,v)}^t\langle v \rangle^{\beta} e^{\zeta |v|^2} g(x-(t-s)\bar{v},v) e^{-\varepsilon(t-s)} ds\Big\},
\end{eqnarray*}
where $t_{\mathbf{b}}(x,v)$ is defined in (\ref{backexit}). We prove the $L^{\infty}-$bound by choosing a large $t=t(\varepsilon)$ such that
\begin{eqnarray*}
| h(x,v)| \ \lesssim \ \frac{1}{2} \{\| h \|_{\infty}+|h|_{\infty}\} + \| \langle v \rangle^{\beta} e^{\zeta |v|^2} r \|_{\infty}+ {\varepsilon}^{-1} \| \langle v \rangle^{\beta} e^{\zeta |v|^2} g \|_{\infty}.
\end{eqnarray*}
In order to prove the continuity, let $(x,v)\in \overline{\Omega}\times\mathbf{R}^3 \backslash \mathfrak{D}$.  Then by the definition of $\mathfrak{D}$, $n(x_{\mathbf{b}}(x,v))\cdot v <0$ and hence $t_{\mathbf{b}}(x,v)$ is smooth by Lemma \ref{tbsmooth}. Therefore, if $g$ and $ r$ are continuous,  $f(x,v)$ is continuous at $(x,v)\in \bar\Omega \times\mathbf{R}^3 \backslash \mathfrak{D}$.

Suppose now that $\Omega$ is convex. In order to show that $\mathfrak{D}=\gamma_0$, since $\mathfrak{D}\supset\gamma_0$ is true for any $\Omega$, it suffices to show that $\mathfrak{D}\subset\gamma_0$.  Since $\Omega$ is convex $t_{\mathbf{b}}(x,v)=0=t_{\mathbf{b}}(x,-v)$ for $x\in\partial\Omega, \ v\neq 0$.  Therefore $\gamma_0^{\mathbf{S}}=\partial\Omega\times\{0\}$.  Hence, if $(x_{\mathbf{b}}(x,v),v)\in\gamma_0^{\mathbf{S}}$ then $v=0$ and $x_{\mathbf{b}}(x,0)=x\in\partial\Omega$ and $(x,v)\in\gamma_0^{\mathbf{S}}$.

\noindent The $L^{2}$-estimate and the uniqueness follow from Green's formula, since $
\|f\|_{2}\lesssim \|\langle v \rangle^{\beta} e^{\zeta |v|^2}f\|_{\infty }$.
\end{proof}

\vskip .3cm
In next lemma we add to the penalized transport equation a suitably cut-offed linearized Boltzmann operator. Moreover we include a {\it reduced} diffuse reflection boundary condition, with the purpose of setting up a contracting map argument.  We have
\begin{lemma}
\label{lineareb}For any $\varepsilon >0, \ m>0$, and for any integer $j>0$, there exists a unique solution to
\begin{equation}
 \varepsilon f+ v\cdot \nabla _{x}f+L_mf = g,\text{ \ \ \ \ \ \ \ \ }%
f_{{-}}=(1-\frac{1}{j})P_{\gamma }f +r,  \label{elinear}
\end{equation}%
with $L_m$ the linearized Boltzmann operator corresponding to the cut-offed cross section $B_m=\min\{B,m\}$.  Moreover,  uniformly in $j$, we have
\begin{equation*}
\|f\|_{\nu }+|f|_{2}\lesssim _{\varepsilon ,m}\|g\|_{2}+|r|_{2}.
\end{equation*}
Finally the limit $f$ as $j\to \infty$ of the sequence $\{f^j\}$ exists and solves uniquely
\begin{equation}
 \varepsilon f+ v\cdot \nabla _{x}f+L_mf = g,\text{ \ \ \ \ \ \ \ \ }%
f_{{-}}=P_{\gamma }f +r . \label{elinear1}
\end{equation}
\end{lemma}

\begin{proof}
Denote $L_m=\nu_m -K_m$.  For any $j$, we apply Lemma \ref{incoming} to the following double-iteration in both $j$ and $\ell$:
\begin{eqnarray}
  \varepsilon f^{\ell+1}+v\cdot \nabla _{x}f^{\ell+1}+\nu_m f^{\ell+1}-K_mf^{\ell}  &=&g ,
\label{approximate} \\
f_{{-}}^{\ell+1} &=&(1-\frac{1}{j})P_{\gamma }f^{\ell}+r,
\notag
\end{eqnarray}
{with $f^0=0$.}
\newline {\it Step 1}. We first fix $m$ and $j$ and take $\ell\rightarrow \infty $.

From Green's identity,
\begin{eqnarray*}
&&\varepsilon \|f^{\ell+1}\|^{2}_2+\frac{1}{2}| f^{\ell+1} |_{2,+}^{2}+\|f^{\ell+1}\|_{\nu _{m}}^{2} \\
&&\ \ \ \ \ \ \ \ \ =(K_mf^{\ell},f^{\ell+1})+\frac{1}{2}\vert (1-\frac{1}{j}%
)P_{\gamma }f^{\ell}+r\vert _{{2,-}}^{2}+(f^{\ell+1},g).
\end{eqnarray*}%
From $(K_m(f^{\ell}+f^{\ell+1}),f^{\ell}+f^{\ell+1})\leq (\nu_m
(f^{\ell}+f^{\ell+1}),f^{\ell}+f^{\ell+1})$, we deduce
\begin{equation*}
(K_mf^{\ell},f^{\ell+1})\leq (\nu_m f^{\ell},f^{\ell+1}).
\end{equation*}%
Moreover, there is $C_j$ such that
\begin{equation}
\vert (1-\frac{1}{j})P_{\gamma }f^{\ell}+r\vert
_{{2,-}}^{2}\leq \vert (1-\frac{1}{j})P_{\gamma }f^{\ell}\vert _{{2,-}}^{2}+\frac{1}{2j^{2}}|P_{\gamma
}f^{\ell}|_{2,-}^{2}+C_{j}|r|_{2}^{2}.\label{Cj}
\end{equation}%
We derive from $|P_{\gamma }f^{\ell}|_{2,-}^{2}\leq |f ^{\ell}|_{2,+}^{2}$ and $\| \cdot \|_{2}\geq \frac{1}{m}\|\cdot \|_{\nu_m}$,
\begin{eqnarray*}
&&\left\{ \frac{\varepsilon }{m}+1\right\} \|f^{\ell+1}\|_{\nu_m }^{2}+\frac{1}{2}%
|f^{\ell+1}|_{2,{+}}^{2}\leq \ \frac{1}{2}\|f^{\ell}\|_{\nu_m }^{2}+\frac{1}{2}%
\|f^{\ell+1}\|_{\nu_m }^{2}\\
&&+\frac{1}{2}(1-\frac{2}{j}+\frac{3}{2j^{2}}%
)|f^{\ell}|_{2,{+}}^{2}+\frac{1}{2}C_j|r|_{2}^{2} +\frac{\varepsilon }{2m}\|f^{\ell+1}\|_{\nu_m
}^{2}+4\varepsilon ^{-1}\|g\|_{2}^{2}.
\end{eqnarray*}%
Since $\frac{\varepsilon }{2m}+1-\frac{1}{2}  >  \frac{1}{2}$ and $1-\frac{2}{j}+%
\frac{3}{2j^{2}}<1$, by iteration over $\ell$, for some $\eta_{\varepsilon
,m,j}<1, $
\begin{equation*}
\|f^{\ell+1}\|_{\nu _{m}}^{2}+|f^{\ell+1}|_{2,{+}}^{2}\leq \eta_{\varepsilon
,m,j}\{\|f^{\ell}\|_{\nu_{m} }^{2}+|f^{\ell}|_{2,+}^{2}\}+C_{\varepsilon
,m,j}\{|r|_{2}^{2}+\|g\|_{2}^{2}\}.
\end{equation*}%
Taking the difference of $f^{\ell+1}-f^{\ell}$ we conclude that $f^{\ell}$ is a Cauchy
sequence. We take $\ell\rightarrow \infty $ to obtain $f^{j}$ as a solution
to the equation%
\begin{equation}
 \varepsilon f^{j}+v\cdot \nabla _{x}f^{j}+L_mf^{j} = g,\text{ \ \ \ \ \ }%
f_{{-}}^{j}=(1-\frac{1}{j})P_{\gamma }f^{j}+r.
\label{jequation}
\end{equation}
\newline {\it Step 2}. We take $j\rightarrow \infty $ for $f^{j}$.

By Green's identity we obtain uniformly in $j$,  %for any $\eta >0,$%
\begin{equation*}
\varepsilon \|f^{j}\|_{2}^{2}+(L_m f^{j},f^{j})+\frac{1}{2}|f ^{j}|_{2,+}^{2}-\frac{1}{2}|P_{\gamma }f^{j}+r|_{2,-}^{2}=\int f^{j}g.
\end{equation*}%
We rewrite for any $\eta >0$
\begin{eqnarray}
\frac{1}{2}|P_{\gamma }f^{j}+r|_{2,-}^{2} &= &\frac{1}{2}%
|P_{\gamma }f^{j}|_{2,-}^{2}+\frac{1}{2}|r|_{2}^{2}+\int_{\gamma
_{-}}P_{\gamma }f^{j} \ r \ d\gamma\notag\\
&\leq &\frac{1}{2}|P_{\gamma }f^{j}|_{2,-}^{2}+C_{\eta
}|r|_{2}^{2}+\eta |P_{\gamma }f^{j}|_{2,-}^{2},\label{eta}
\end{eqnarray}%
so that from $\int f^{j}g\leq \frac{\varepsilon }{2}\|f^{j}\|_{2}^{2}+C_{%
\varepsilon }\|g\|_{2}^{2}$ and from the spectral gap of $L_m$, we have
\begin{equation}
\frac{\varepsilon }{2}\|f^{j}\|_{2}^{2}+\|(\mathbf{I}-\mathbf{P}%
)f^{j}\|_{\nu_m }^{2}+\frac{1}{2}|(1-P_{\gamma })f^{j}|_{2,+}^{2}\leq C_{\eta
,\varepsilon }\{|r|_{2}^{2}+\|g\|_{2}^{2}\}+\eta |P_{\gamma }f^{j}|_{2,-}^{2}.  \label{epsilonenergy}
\end{equation}%
But from the equation we have
\begin{equation*}
v\cdot \nabla _{x}[f^{j}]^{2}=-2\varepsilon \lbrack f^{j}]^{2}-2f^{j}L_m(\mathbf{I}-%
\mathbf{P})f^{j}+2f^{j}g.
\end{equation*}%
Taking absolute value and integrating on $\Omega\times\mathbf{R}^3$,  from (\ref{epsilonenergy}) we have
\begin{eqnarray*}
\|v\cdot \nabla _{x}(f^{j} )^{2}\|_{1}&\leq& C_{\varepsilon }\{\|f^j\|_{2}^{2}+\|(%
\mathbf{I}-\mathbf{P})f^{j}\|_{\nu_{m} }^{2}+\|g\|_{2}^{2}\}\\
&\leq& C_{\eta
,\varepsilon }\{|r|_{2}^{2}+\|g\|_{2}^{2}\}+\eta C_{\varepsilon }|P_{\gamma
}f^{j}|_{2,-}^{2}.
\end{eqnarray*}%
Hence, by Lemma \ref{ukai}, for any $\gamma ^{\varepsilon^{\prime} }$ in (\ref{ukaitrace}) away from $\gamma_{0},$
we have
\begin{equation}
|f ^{j} \ \mathbf{1}_{\gamma^{\varepsilon^{\prime}}}|_{2 }^{2}\leq C_{\varepsilon ,\eta ,\varepsilon^{\prime}
}\{|r|_{2}^{2}+\|g\|_{2}^{2}\}+\eta C_{\varepsilon ,\varepsilon^{\prime} }|P_{\gamma
}f^{j}|_{2,-}^{2}.  \label{tracefk}
\end{equation}%
From (\ref{pgamma}) we can write $P_{\gamma }f^j =z_{\gamma}(x)\sqrt{\mu}$ for a suitable function $z_\gamma(x)$ and, from $|P_{\gamma}f^j \mathbf{1}_{\gamma^{\varepsilon^{\prime}}}|_2\lesssim |f ^{j} \ \mathbf{1}_{\gamma^{\varepsilon^{\prime}}}|_{2 } < +\infty$, for $\varepsilon^{\prime}$ small
\begin{eqnarray}
|P_{\gamma}f^j \ \mathbf{1}_{\gamma^{\varepsilon^{\prime}}}|_2^2 &=&
\int_{\partial\Omega} |z_{\gamma}(x)|^2 \int_{ |n(x)\cdot v|\geq \varepsilon^{\prime},|v|\leq \frac{1}{\varepsilon^{\prime}} } \mu(v)|v\cdot n(x)|
dvdx\notag\\
&\geq &  \int_{\partial\Omega } |z_{\gamma}(x)|^2 dx\times \frac{1}{2} \int_{\mathbf{R}^3}  \mu(v)|v\cdot n(x)| dv\notag\\
&=& \frac{1}{2} |P_{\gamma}f^j|_2^2,\label{Pgamma}
\end{eqnarray}
where we used the fact
\begin{eqnarray}
&&\int_{|n(x)\cdot v|\leq \varepsilon^{\prime}} \mu(v) |n(x)\cdot v| dv \leq
\int_{-\varepsilon^{\prime}}^{\varepsilon^{\prime}}e^{-v_{\|}^2} \ |v_{\|}| dv_{\|}
\int_{\mathbf{R}^2} e^{-{|v_{\bot}|^2}/{2}} dv_{\bot} \leq C\varepsilon^{\prime},\label{epsilonprime1}\\
&&\int_{|v|\geq 1/\varepsilon^{\prime}} \mu(v) |n(x)\cdot v| dv \leq C\varepsilon^{\prime}.\notag
\end{eqnarray}
Therefore we conclude
\begin{eqnarray}
\frac{1}{2}|P_{\gamma }f^{j}|_{2}^{2}-|(1-P_{\gamma })f^{j}|_{2,+}^{2} &\leq
&|P_{\gamma }f ^{j}\mathbf{1}_{\gamma^{\varepsilon^{\prime}}}|_{2}^{2}-|(1-P_{\gamma })f^{j}\mathbf{1}_{\gamma^{\varepsilon^{\prime}}}|_{2,+}^{2}\label{tracep}\\
&\lesssim& |f ^{j} \mathbf{1}_{\gamma^{\varepsilon^{\prime}}} |_{2  }^{2}\notag\\
&\leq&  C_{\varepsilon ,\eta ,\varepsilon^{\prime} }\{|r|_{2}^{2}+\|g\|_{2}^{2}\}+\eta
C_{\varepsilon ,\varepsilon^{\prime} }|P_{\gamma }f ^{j}|_{2,-}^{2}.  \notag
\end{eqnarray}%
Adding $4\times (\ref{epsilonenergy})$ to $(\ref{tracep})$, we
obtain:%
\begin{eqnarray*}
&& \ \ 2\varepsilon \|f^{j}\|_{2}^{2}+|(1-P_{\gamma })f ^{j}|_{2,+}^{2}+4 \|(\mathbf{I}-\mathbf{P})f^{j}\|_{\nu_{m } }^{2}+\frac{1}{2}%
|P_{\gamma }f^j|_{2}^{2}\\
&&\leq \ C_{\varepsilon ,\eta
}\{|r|_{2}^{2}+\|g\|_{2}^{2}\}+\eta (1+ C_{\varepsilon ,\varepsilon^{\prime}
})|P_{\gamma }f ^{j}|_{2,-}^{2}.
\end{eqnarray*}%
Choosing $\eta $ small and taking the weak limit $j\rightarrow \infty$,  we complete the
proof of the lemma.
\end{proof}
\vskip .3cm
Next lemma states the crucial $L^2$ bound for ${\mathbf P}f$.  It will provide uniform in $\varepsilon$ estimates which allow to take the limit $\varepsilon\to 0$ in (\ref{elinear1}).
\begin{lemma}
\label{steadyabc}Let $f$ be a solution, in the sense of (\ref{weakformulation}) below, to
\begin{equation}
v\cdot \nabla _{x}f+L f=g,\text{ \ \ \ \ \ \ \ }f_{ {-}}=P_{\gamma
}f+r,\label{linearf2}
\end{equation}%
with
\begin{equation*}
\iint_{\Omega\times\mathbf{R}^3}f\sqrt{\mu }dxdv
= \iint_{\Omega\times\mathbf{R}^3}   g\sqrt{\mu }dxdv=\int_{\gamma_-}r \sqrt{\mu}d\gamma=0,
\end{equation*}%
then we have
\begin{equation*}
\|\mathbf{P}f\|_{\nu }^{2}\lesssim \|(\mathbf{I}-\mathbf{P})f\|_{\nu
}^{2}+\|g\|_{2}^{2}+|(1-P_{\gamma })f|_{2,+}^{2}+|r|_{2 }^{2}.
\end{equation*}
\end{lemma}

\vskip .3cm
\begin{proof}
 The Green's identity (\ref{steadyGreen}) provides the following weak version of (\ref{linearf2}):
\begin{equation}
\int_{\gamma }\psi fd \gamma-\iint_{\Omega \times \mathbf{R}^{3}} {v}\cdot \nabla
_{x}\psi f=-\iint_{\Omega \times \mathbf{R}^{3}}\psi L  (\mathbf{I}-\mathbf{P%
})f+\iint_{\Omega \times \mathbf{R}^{3}} \psi g.  \label{weakformulation}
\end{equation}%
Recall that $\mathbf{P}f=\{a+v\cdot b+c[\frac{|v|^{2}}{2}-\frac{3}{2}]\}\sqrt{\mu
}$ on $\Omega\times\mathbf{R}^3$.  The key of the proof is to choose suitable $H^1$ test functions $\psi$ to estimate $a,b$ and $c$ in (\ref{cestimate}), (\ref{stiijjkiki}), (\ref{stijijiijj}), (\ref{aestimate}) thus concluding the proof of  Proposition \ref{linearl2}.
\newline

\noindent {\it Step 1}. {Estimate of} \ ${c}$

To estimate $c$, we first choose the test function
\begin{equation}
\psi=\psi_c \equiv(|v|^{2}-\beta_c )\sqrt{\mu }v \cdot \nabla_x \phi _{c}(x),
\label{phic}
\end{equation}%
where
\begin{equation*}
-\Delta_x \phi_{c}(x)=c(x)  ,\ \ \ \phi _{c}|_{\partial \Omega }=0,
\end{equation*}%
and $\beta_c$ is a constant to be determined. From the standard elliptic estimate, we have
\begin{equation*}
\|\phi _{c}\|_{H^{2} }\lesssim \|c\|_{2
}.
\end{equation*}%
With the choice (\ref{phic}), the right hand side of (\ref{weakformulation}) is bounded by%
\begin{equation}
\|c\|_{2
}\big\{\|(\mathbf{I}-\mathbf{P})f\|_{2}+\|g\|_{2}\big\}.  \label{rweak}
\end{equation}
We have
\begin{equation*}
{v}\cdot \nabla _{x}\psi_c =\sum_{i,j=1}^{d}(|v|^{2}-\beta_c )v_{i}\sqrt{\mu }%
v_{j}\partial _{i j}\phi _{c}(x),
\end{equation*}%
so that the left hand side of (\ref{weakformulation}) takes the form, for $i=1,\cdots,d,$
\begin{eqnarray}
&&\int_{\partial \Omega \times \mathbf{R}^{3}}(n(x)\cdot v)(|v|^{2}-\beta_c )\sqrt{\mu }%
\sum_{i=1}^d v_{i}\partial _{i}\phi _{c}f\notag\\
&&-\iint_{\Omega \times \mathbf{R}%
^{3}}(|v|^{2}-\beta_c )\sqrt{\mu }\Big\{\sum_{i,j=1}^d v_{i}v_{k}\partial _{ij}\phi
_{c} \Big\}f .  \label{lweak}
\end{eqnarray}%
We decompose
\begin{eqnarray}
f_{\gamma } &=&P_{\gamma }f +\mathbf{1}_{\gamma
_{+}}(1-P_{\gamma })f+\mathbf{1}_{\gamma _{-}}r,\text{ \ \  \ \ \ \ \ \ \ \ \ \ \ \ \ \ \ \ \ on
}\gamma,  \label{bsplit} \\
f &=&\Big\{a+v\cdot b+c[\frac{|v|^{2}}{2}-\frac{3}{2}]\Big\}\sqrt{\mu }+(\mathbf{I}-%
\mathbf{P})f,\text{ \ \ \ \ \  on }\Omega\times\mathbf{R}^3. \label{insidesplit}
\end{eqnarray}%
We shall choose $\beta_c $ so
that, for all $i$
\begin{equation}
\int (|v|^{2}-\beta_c )v_{i}^{2}\mu(v) dv=0.  \label{beta}
\end{equation}%
Since $\mu= \frac{1}{2\pi}e^{-\frac{|v|^2}{2}}$ the desired value of $\beta_c$  is $\beta_c=5$.
Because of the choice of $\beta_c$, there is no $a$ contribution in the bulk and no $P_{\gamma }f$ contribution at
the boundary in (\ref{lweak}).

Therefore, substituting (\ref{bsplit}) and (\ref{insidesplit}) into (\ref%
{lweak}), since the $b$ terms and the off-diagonal $c$ terms also vanish by oddness in $v$ in the bulk, the left hand side of (\ref{weakformulation}) becomes%
\begin{eqnarray*}
&& \sum_{i=1}^d \int_{\gamma }(|v|^{2}-\beta_c ){v_{i}}^{2}n_{i}\sqrt{\mu }\partial _{i}\phi
_{c}[(1-P_{\gamma })f \mathbf{1}_{\gamma _{+}}+r\mathbf{1}_{\gamma _{-}}]
\\
&&-\sum_{i=1}^d \int_{\mathbf{R}^3}(|v|^{2}-\beta_c )v_{i}^{2} (\frac{|v|^{2}}{2}-\frac{3}{2})\mu(v)dv \int_{\Omega } \partial
_{ii}{\phi _{c}(x)c(x)} dx \\
&&-\sum_{i=1}^d \iint_{\Omega \times \mathbf{R}^{3}}(|v|^{2}-\beta_c )v_{i}\sqrt{\mu }%
( {v}\cdot \nabla_x )\partial _{i}\phi _{c}(\mathbf{I}-\mathbf{P})f.
\end{eqnarray*}%
For $\beta_c =5$, $\int_{\mathbf{R}^3}(|v|^{2}-\beta_c )v_{i}^{2} (\frac{|v|^{2}}{2}-\frac{3}{2})\mu(v)dv= 10\pi\sqrt{2\pi}$.
Therefore, we obtain from (\ref{rweak})
\begin{multline*}
- 10\pi\sqrt{2\pi} \int_{\Omega }\Delta_x \phi _{c}(x)c(x)
\lesssim \ \|c\|_{2 }\{|(1-P_{\gamma })f |_{2,+}+\|(\mathbf{I}-\mathbf{P}%
)f\|_{2}+\|g\|_{2}+|r|_{2}\},
\end{multline*}%
where we have used the elliptic estimate and the trace estimate:
\begin{equation*}
|\nabla_x \phi _{c}|_{2 }\lesssim \|\phi _{c}\|_{H^{2} }\lesssim \|c\|_{2}.
\end{equation*}%
Since $-\Delta_x \phi_c = c$, from (\ref{phic}) we obtain
\begin{equation*}
\|c\|_{2 }^{2} \ \lesssim \ \big\{|(1-P_{\gamma })f |_{2,+}+\|(\mathbf{I}-\mathbf{P}%
)f\|_{2}+\|g\|_{2}+|r|_{2}\big\}\|c\|_{2 },\end{equation*}%
and hence
\begin{equation}
\|c\|_{2 }^2 \ \lesssim \ |(1-P_{\gamma })f |_{2,+}^{2}+\|(\mathbf{I}-\mathbf{P}%
)f\|_{2}^{2}+\|g\|_{2}^{2}+|r|_{2}^{2}.  \label{cestimate}
\end{equation}

\noindent{\it Step 2.} {Estimate of} \ $ {b}$

We shall establish the estimate of $b$ by estimating   $(\partial_i \partial_j \Delta^{-1}b_j) b_i$ for all $i,j=1,\dots, d$, and  $(\partial_j\partial_j \Delta^{-1} b_i) b_i$ for $i\neq j$.

We fix $i,j$.  To estimate $\partial_i \partial_j \Delta^{-1}b_j b_i$ we choose as test function in (\ref{weakformulation})
\begin{equation}
\psi = \psi^{i,j}_b\equiv (v_{i}^{2}-\beta_ b)\sqrt{\mu }\partial _{j}\phi _{b}^{j}, \quad i,j=1,\dots, d,  \label{phibj}
\end{equation}%
where $\beta_b$ is a constant to be determined, and
\begin{equation}
-\Delta_x \phi _{b}^{j}(x)=b_{j}(x)  , \ \ \ \phi_b^{j}|_{\partial\Omega}=0.\label{jb}
\end{equation}%
From the standard elliptic estimate%
\begin{equation*}
\|\phi _{b}^{j}\|_{H^{2} }\lesssim \|b\|_{2 }.
\end{equation*}%
Hence the right hand side of (\ref{weakformulation}) is now bounded by%
\begin{equation}
 \|b\|_{2 }\big\{\|(\mathbf{I}-\mathbf{P})f\|_{2}+\|g\|_{2}\big\}.
\label{rweakbj}
\end{equation}%
Now substitute (\ref{bsplit}) and (\ref{insidesplit}) into the left hand side of (\ref{weakformulation}). Note that $(v_{i}^{2}-\beta_b  )\{n(x)\cdot  {v}\}\mu $ is odd in $v$, therefore $P_{\gamma
}f $ contributions to (\ref{weakformulation}) vanishes. Moreover, by (\ref{insidesplit}), the $a, c$ contributions to (\ref{weakformulation}) also vanish by oddness. Therefore the left hand side of (\ref{weakformulation}) takes the form %
\begin{eqnarray}
&&\int_{\partial \Omega \times \mathbf{R}^{3}}(n(x)\cdot  {v})(v_{i}^{2}-\beta_b  )\sqrt{\mu }%
\partial _{j}\phi _{b}^{j}f-\iint_{\Omega \times \mathbf{R}%
^{3}}(v_{i}^{2}-\beta_b  )\sqrt{\mu }\{\sum_{l}v_{l} \partial_{l j}\phi
_{b}^{j}\}f  \notag \\
&=&\int_{\partial \Omega \times \mathbf{R}^{3}}(n(x)\cdot  {v})  (v_{i}^{2}-\beta_b  )
\sqrt{\mu }\partial _{j}\phi _{b}^{j}[(1-P_{\gamma })f +r]\mathbf{1}%
_{\gamma _{+}}  \label{sti_6} \\
&&-\int_{\Omega } { {  \int_{\mathbf{R}%
^{3}}\sum_{l}(v_{i}^{2}-\beta_b )v_{l}^{2}\mu \partial _{lj}\phi
_{b}^{j}(x)b_{l}dv } } dx \label{sti_8} \\
&&-\iint_{\Omega \times \mathbf{R}^{3}}\sum_{l}(v_{i}^{2}-\beta_b  )v_{l}\sqrt{\mu }%
\partial _{l j} \phi _{b}^{j}(x)(\mathbf{I}-\mathbf{P})f.  \notag
\end{eqnarray}%
Furthermore, since $\mu (v)=\frac{1}{2\pi}\prod_{i=1}^3  e^{-\frac{|v_i|^2}{2}}$ we can
choose $\beta_b >0$ such that for all $i,$
\begin{equation}
\int_{\mathbf{R}^{3}}[(v_{i})^{2}-\beta_b ]\mu (v) dv=  \int_{%
\mathbf{R}}[v_{1}^{2}-\beta_b]  e^{-\frac{|v_1|^2}{2}} dv_1=0.  \label{alpha}
\end{equation}%
We remark that the choice (\ref{alpha}) also plays a crucial rule in the dynamical estimate (\ref{i_7}). Since $\mu (v)=\frac{1}{2\pi}e^{-\frac{|v|^{2}}{2}}$, the desired value is $\beta_b=1$.

Note that for such chosen $\beta_b
, $ and for $i\neq k$, by an explicit computation
\begin{eqnarray*}
 \int (v_{i}^{2}-\beta_b )v_{k}^{2}\mu dv  &=& \int (v_{1}^{2}-\beta_b
)v_{2}^{2}  \frac{1}{2\pi}e^{-\frac{|v_1|^2}{2}} e^{-\frac{|v_2|^2}{2}} e^{-\frac{|v_3|^2}{2}} dv\\
&=&  \int_{\mathbf{R}}
(v_{1}^{2}-\beta_b ) e^{-\frac{|v_1|^2}{2}} dv_1=0, \\
 \int (v_{i}^{2}-\beta_b )v_{i}^{2}\mu dv & = &  \int_{\mathbf{R}}
[v_{1}^{4}-\beta_b v_{1}^{2}]e^{-\frac{|v_1|^2}{2}}dv_{1}= 2\sqrt{2\pi}\neq 0.
\end{eqnarray*}%
Therefore, (\ref{sti_8}) becomes, by (\ref{phibj}),
\begin{eqnarray*}
&&-\iint_{\Omega\times\mathbf{R}^3} (v_{i}^{2}-\beta_b )v_{i}^{2}\mu dv \partial _{ij}\phi^j
_{b}(x)b_{i}+\sum_{k(\neq i)}\underbrace{[\int_{\mathbf{R}^3} (v_{i}^{2}-\beta_b )v_{k}^{2}\mu
 ]}_{=0}\int_{\Omega}\partial _{kj}\phi _{b}^j (x)b_{k}  \\
&&=2\sqrt{2\pi}\int_{\Omega }(\partial
_{i}\partial _{j}\Delta ^{-1}b_{j})b_{i}.
\end{eqnarray*}%
Hence we have the following estimate for all $i,j$, by (\ref{rweakbj}):
\begin{eqnarray}
 \left|\int_{\Omega }\partial _{i}\partial _{j}\Delta ^{-1}b_{j}b_{i}\right|
 \lesssim  |(1-P_{\gamma })f |_{2,+}^{2}+\|(\mathbf{I}-\mathbf{P}%
)f\|^{2}_{2}+ \|g\|_{2}^{2}+|r|_{2}^{2} +\varepsilon \|b\|_{2}^{2}.  \label{stiijjkiki}
\end{eqnarray}%
In order to estimate $ \partial _{j}(\partial
_{j}\Delta ^{-1}b_{i})b_{i}$ for $i\neq j$, we choose as test function in (\ref{weakformulation})
\begin{equation}
\psi =%\varphi^{i,j}_b\equiv
|v|^{2}v_{i}v_{j}\sqrt{\mu }\partial _{j}\phi _{b}^{i}(x),\quad i\neq j,
\label{phibij}
\end{equation}%
where $\phi_b^i$ is given by (\ref{jb}). Clearly, the right hand side of (\ref{weakformulation}) is again bounded by (\ref{rweakbj}%
). We substitute again (\ref{bsplit}) and (\ref{insidesplit}) into the left hand side of (\ref{weakformulation}). The $P_{\gamma}f$ contribution and $a,c$ contributions vanish again due to oddness. Then the left hand side of (\ref{weakformulation}) becomes
\begin{eqnarray}
&&\int_{\partial \Omega \times \mathbf{R}^{3}}\{n\cdot  {v}\}|v|^{2}v_{i}v_{j}\sqrt{\mu }%
\partial _{j}\phi _{b}^{i}f-\iint_{\Omega \times \mathbf{R}%
^{3}}|v|^{2}v_{i}v_{j}\sqrt{\mu }\{\sum_{k}v_{k}\partial _{kj}\phi
_{b}^{i}\}f  \notag \\
&=&\int_{\partial \Omega \times \mathbf{R}^{3}} \{n\cdot  {v}\} |v|^{2} v_{i}v_{j}
\sqrt{\mu }\partial_j\phi ^{i}_b[(1-P_{\gamma })f +r]\mathbf{1}_{\gamma _{+}}
\label{stij_6} \\
&&- {\iint_{\Omega \times \mathbf{R}^{3}}|v|^{2}v_{i}^{2}v_{j}^{2}%
\mu \lbrack \partial _{ij}\phi _{b}^{i}b_{j}+\partial _{jj}\phi
_{b}^{i}(x)b_{i}]}  \label{stij_8} \\
&&-\iint_{\Omega \times \mathbf{R}^{3}}|v|^{2}v_{i}v_{j}v_{k}\sqrt{\mu }%
\partial _{kj}\phi _{b}^{i}(x)[\mathbf{I}-\mathbf{P}]f.  \label{stij_9}
\end{eqnarray}%
Note that (\ref{stij_8}) is evaluated as
\begin{equation*}
7\sqrt{2\pi} \int_{\Omega }\{(\partial _{i}\partial _{j}\Delta
^{-1}b_{i})b_{j}+(\partial _{j}\partial _{j}\Delta ^{-1}b_{i})b_{i}\}.
\end{equation*}%
Furthermore, by $(\ref{jb})$, $|\partial _{j}\phi _{b}^{i}|_{2}\lesssim
\|\phi _{b}^{i}\|_{H^{2}}\lesssim \|b\|_{2}$, so that
\begin{equation*}
(\ref{stij_6})+(\ref{stij_9})\lesssim \|b\|_{2}\big\{|(1-P_{\gamma
})f |_{2,+} +|r|_2 +\|(\mathbf{I}-\mathbf{P})f\|_{2} \big\}.
\end{equation*}
Combining (\ref{stiijjkiki}), we have the following estimate for $i\neq j,$
\begin{eqnarray}
&&\left|\int_{\Omega} \partial _{j}\partial _{j}\Delta ^{-1}b_{i}b_{i}\right|\   \notag \\
&\lesssim &\ \left|\int_{\Omega} \partial _{i}\partial _{j}\Delta
^{-1}b_{i}b_{j}\right|+ |(1-P_{\gamma })f |_{2,+}^{2}+ \|(\mathbf{I}-\mathbf{P})f\|^{2}_{2}+\|g\|_{2}^{2}+ |r|_2^2 +\varepsilon \|b\|_2^2  \notag \\
&\lesssim &  |(1-P_{\gamma })f |_{2,+}^{2} +\|(\mathbf{I}-\mathbf{P})f\|^{2}_{2}+\|g\|_{2}^{2}+ |r|_2^2 +\varepsilon \|b\|_2^2.
\label{stijijiijj}
\end{eqnarray}%
Moreover, by (\ref{stiijjkiki}), for $i=j=1,2,\dots,d,$
\begin{eqnarray}
&& \ \ \left|\int_{\Omega} \partial _{j}\partial _{j}\Delta ^{-1}b_{j}b_{j}\right| \   \notag \\
&&  \lesssim \ |(1-P_{\gamma })f |_{2,+}^{2} +\|(\mathbf{I}-\mathbf{P})f\|^{2}_{2}+\|g\|_{2}^{2}+ |r|_2^2 +\varepsilon \|b\|_2^2.\label{stjjjj}
\end{eqnarray}
Combining (\ref{stijijiijj}) and (\ref{stjjjj}), we sum over $j=1,2,\dots,d$, to obtain, for all $i=1,2,\dots,d,$
\begin{eqnarray}
\| b_i \|_2 \lesssim \ |(1-P_{\gamma })f |_{2,+}^{2} +\|(\mathbf{I}-\mathbf{P})f\|^{2}_{2}+\|g\|_{2}^{2}+ |r|_2^2.\label{b_i}
\end{eqnarray}

\noindent{\it Step 3.} {Estimate of} \ ${a}$

The estimate for $a$ is more delicate because it requires the zero mass condition
$$\iint_{\Omega\times\mathbf{R}^3} f \sqrt{\mu}dx dv= 0=\int_{\Omega} a dx.$$  We choose a test function%
\begin{eqnarray}
 \psi =\psi_a  \equiv  (|v|^{2}-\beta_{a} )v\cdot\nabla_x\phi _{a}\sqrt{\mu }
 = \sum_{i=1}^d(|v|^{2}-\beta_{a} )v_{i}\partial _{i}\phi _{a}\sqrt{\mu } , \label{phia}
\end{eqnarray}%
where
\begin{equation*}
-\Delta_x \phi _{a}(x)=   a(x) ,\text{ \ \ \ \   }\frac{%
\partial }{\partial n}\phi_a|_{\partial\Omega} =0.
\end{equation*}%
It follows from the elliptic estimate with $\int_{\Omega}a =0$ that we have
\begin{equation*}
\|\phi_{{a}} \|_{H^{2} }\lesssim \|a\|_{2 }.
\end{equation*}%
Since $\int_{\mathbf{R}^3}(\frac{|v|^2}{2}-\frac{3}{2})(v_i)^2 \mu(v) dv\neq 0$, we can choose $\beta_a>0 $ so that, for all $i,$
\begin{equation}
\int_{\mathbf{R}^{3}}(|v|^{2}-\beta_a )(\frac{|v|^{2}}{2}-\frac{3}{2}%
)(v_{i})^{2}\mu(v) =0.\label{betaalpha}
\end{equation}%
Since $\mu(v)=\frac{1}{2\pi}e^{-\frac{|v|^2}{2}}$, the desired value is $\beta_a =10$.
Plugging $\psi_a $ into (\ref{weakformulation}) and its right hand side is again bounded by%
\begin{equation*}
 \|a\|_{2}\big\{\|(\mathbf{I}-\mathbf{P})f\|_{2}+\|g\|_{2}\big\}.
\end{equation*}%
By (\ref{bsplit}) and (\ref{insidesplit}), since the $c$ contribution vanishes in (\ref{weakformulation}) due to our choice of $\beta_a $ and the $b$ contribution vanishes in (\ref{weakformulation}) due to the oddness, the right hand side of (\ref{weakformulation}%
) takes the form of
\begin{eqnarray}
&&\sum_{i=1}^d\int_{\gamma } \{ n \cdot  {v}\} (|v|^{2}-\beta_{a} )v_{i}\sqrt{\mu }\partial _{i}\phi
_{a}(x) [P_{\gamma }f+(I-P_{\gamma })f\mathbf{1}_{\gamma _{+}}+r%
\mathbf{1}_{\gamma _{+}}] \ \ \ \ \ \ \label{aboundary} \\
&&-\sum_{i,k=1}^d \iint_{\Omega\times\mathbf{R}^3} (|v|^{2}-\beta_a )v_{i}v_{k}\partial _{ik}\phi _{a}(x)a(x)\mu(v) \label{abulk}\\
&&-\sum_{i,k=1}^d\iint_{\Omega\times\mathbf{R}^3} (|v|^{2}-\beta_a )v_{i}v_{k}\partial _{ik}\phi _{a}(x)(\mathbf{I}-%
\mathbf{P})f.\label{a2bulk}
\end{eqnarray}%
We make an orthogonal decomposition at the boundary,
\begin{equation*}
 {v}_i=(  {v}\cdot n)n_i +  ({v}_{\perp })_i = v_n n_i + ({v}_{\perp })_i .
\end{equation*}%
The contribution of $P_{\gamma}f=z_{\gamma}(x)\sqrt{\mu}$ in (\ref{aboundary}) is
\begin{eqnarray*}
&&\int_{\gamma}(|v|^{2}-\beta_a ) {v}\cdot \nabla_x \phi _{a}(x) {v}_{n}\mu(v) z_{\gamma}(x)  \\
&=&\int_{\gamma }(|v|^{2}-\beta_a ) {v}_{n} \frac{\partial \phi _{a}}{%
\partial n} {v}_{n}\mu(v) z_{\gamma}(x) \\
&&+ \int_{\gamma }(|v|^{2}-\beta_a ) v_{\bot}\cdot \nabla_x \phi_a  {v}_{n}\mu(v) z_{\gamma}(x) .
\end{eqnarray*}%
The crucial choice of (\ref{betaalpha}) makes the first term vanish due to the Neumann boundary condition, while the second term also
vanishes due to the oddness of $ ({v}_{\bot})_i  {v}_n$ for all $i$.
Therefore, (\ref{aboundary}) and (\ref{a2bulk}) are bounded by%
\begin{equation*}
\|a\|_{2 }\big\{\|(\mathbf{I}-%
\mathbf{P})f\|_{2}+ |(1-P_{\gamma })f |_{2,+} +|r|_{2} \big\}.
\end{equation*}%
The second term (\ref{abulk}), for $k\neq i$ vanishes due to the oddness. Hence we only have the $k=i$ contribution:
\begin{equation*}
\sum_{i=1}^d\iint_{\Omega \times \mathbf{R}^{3}}(|v|^{2}-\beta_a )(v_{i})^{2}\mu \partial
_{ii}\phi _{a}a.
\end{equation*}%
Using $-\Delta_x \phi_{a}=a$ we obtain
\begin{equation}
\|a\|_{2 }^{2}\lesssim \|(\mathbf{I}-\mathbf{P})f\|_{2}^{2}+|(1-P_{\gamma
})f |_{2,+}^{2}+|r|_{2}^{2}+\|g\|_{2}^{2}.\label{aestimate}
\end{equation}
\end{proof}

We close this section by proving Proposition \ref{linearl2}.
\newline

\begin{proof}[Proof of the Proposition \ref{linearl2}]

We keep $m$ fixed and take $\varepsilon \rightarrow 0$ in Lemma %
\ref{lineareb} by using Lemma \ref{steadyabc} which obviously holds even with the additional penalization term.  Indeed $\iint_{\Omega\times\mathbf{R}^3} g\sqrt{\mu} = 0 = \int_{\gamma_-} r \sqrt{\mu} d\gamma$. Hence we have $\varepsilon \iint_{\Omega\times\mathbf{R}^3} f^{\varepsilon} \sqrt{\mu}=0$ and therefore, for any $\varepsilon>0$,
\[
\iint_{\Omega\times\mathbf{R}^3} f^{\varepsilon}\sqrt{\mu} dx dv =0.
\]
We first obtain%
\begin{equation}
\|\mathbf{P}f^{\varepsilon }\|_{2}^{2}\lesssim \|(\mathbf{I}-\mathbf{P}%
)f^{\varepsilon }\|_{2}^{2}+|(1-P_{\gamma })f^{\varepsilon
}|_{2,+}^{2}+|r|_{2}^{2}+\|g\|_{2}^{2}+\varepsilon \|f^{\varepsilon
}\|_{2}^{2},  \label{pfe}
\end{equation}%
On the other hand from Green's identity:
\begin{equation*}
\varepsilon \|f^{\varepsilon }\|_{2}^{2}+(L_{m}f^{\varepsilon },f^{\varepsilon
})+\frac{1}{2}|f^{\varepsilon }|_{2,+}^{2}-\frac{1}{2}|P_{\gamma
}f^{\varepsilon }+r|_{2,-}^{2}=\int f^{\varepsilon }g,
\end{equation*}%
we deduce from the spectral gap of $L_m$
\begin{equation}
\varepsilon \|f^{\varepsilon }\|_{2}^{2}+\|(\mathbf{I}-\mathbf{P}%
)f^{\varepsilon }\|_{\nu_{m  } }^{2}+\frac{1}{2}|(1-P_{\gamma })f^{\varepsilon }|_{2,+}^{2}\leq \eta \lbrack \|f^{\varepsilon
}\|_{2}^{2}+|P_{\gamma }f^{\varepsilon }|_{2,-}^{2}]+C_{\eta
}[|r|_{2}^{2}+\|g\|_{2}^{2}].  \label{energye}
\end{equation}%
From the argument of (\ref{Pgamma}) and the trace theorem as well as the equation (\ref{linearf}),
\begin{equation*}
|P_{\gamma }f^{\varepsilon }|_{2}^{2}\lesssim \|v\cdot \nabla
_{x}(f^{\varepsilon })^2\|_{1} +\|f^{\varepsilon }\|_{2}^{2}\lesssim \|(%
\mathbf{I}-\mathbf{P})f^{\varepsilon }\|_{\nu_{m }
}^{2}+\|g\|_{2}^{2}+\|f^{\varepsilon }\|_{2}^{2}.
\end{equation*}%
Plugging this into (\ref{energye}) with $\eta $ small and adding a small constant $%
\times $(\ref{pfe}), collecting terms and using the fact
$$\| \mathbf{P} f^{\varepsilon} \|_2^2 \ \sim \ \| \mathbf{P}f^{\varepsilon}\|_{\nu_m}^2,\notag$$
we obtain the uniform in $\varepsilon $ estimate
\begin{equation}
\|f^{\varepsilon }\|_{\nu_{m } }^{2}+|f^{\varepsilon
}|_2  ^{2} \ \lesssim \ \|g\|_{2}^{2}+|r|_{2}^{2}.  \label{ebound}
\end{equation}%
We thus obtain a weak solution $f^{\varepsilon }\rightarrow f$ with
the same bound (\ref{ebound}). Moreover, we have%
\begin{equation*}
\varepsilon \lbrack f^{\varepsilon }-f]+v\cdot \nabla _{x}[f^{\varepsilon
}-f]+L_{m }[f^{\varepsilon }-f]=\varepsilon f,\text{ \ \ \ \ \ \ \ \ }%
[f^{\varepsilon }-f]_{ {-}}=P_{\gamma }[f^{\varepsilon }-f] .
\end{equation*}%
We conclude from (\ref{ebound}), written for the difference of $f^{\varepsilon}-f$, that
\begin{eqnarray*}
\|f^{\varepsilon }-f\|_{\nu_{m } }^{2}+|[f^{\varepsilon }-f] |_{2}^{2}&\lesssim& \varepsilon \| f \|_2^2 \\
&\lesssim& \varepsilon \big\{
\|g\|_{2}^{2}+\|r\|_{2}^{2}\big\}\rightarrow 0.  \label{ecauchy}
\end{eqnarray*}%
The proposition follows as $\varepsilon \rightarrow 0$.

Finally from (\ref{ebound}) again we can take the  limit $m\to \infty$ of $f^m$ solution to
\begin{equation*}
v\cdot \nabla _{x}f^{m}+L_{m}f^{m}=g,\text{ \ \ \ \ \ \ \ }f_{ {-}}^{m}=P_{\gamma }f ^{m}+r.
\end{equation*}%
 By a
diagonal process, there exists a weak solution $f$ such that $%
f^{m}\rightarrow f$ $\ $weakly in $\|\cdot \|_{\nu_{m_{0}}}$, for any fixed $%
m_{0}$.  It thus follows from weak semi-continuity of the norm $\|\,\cdot\,\|_{\nu_{m_0}}$ that
\begin{equation*}
\|f\|_{\nu _{m_{0}}}^{2}+|f|_{2}^{2} \ \lesssim \ \|g\|_{2}^{2}+|r|_{2}^{2}.
\end{equation*}%
Note that the limiting $f$ so obtained satisfies  $\iint_{\Omega\times\mathbf{R}^3} f(x,v)\sqrt{\mu(v)}dv dx=0. $
The proposition follows as $m_{0}\rightarrow \infty $ and uniqueness follows
from Green's identity. We remark that due to lack of moments control of $f$ we cannot show $f^m \rightarrow f$ strongly in $L^2$.
\end{proof}

\bigskip
\section{$L^{\infty}$ Estimate along the Stochastic Cycles}
We define a weight function scaled with parameter $\varrho,$
\begin{equation}
w_{\varrho}(v)= w_{\varrho, \beta, \zeta}(v) \equiv (1+\varrho^2|v|^2)^{\frac{\beta}{2}} e^{\zeta|v|^2}.\label{weight}
\end{equation}
The main purpose of this section is to prove the following:
\begin{proposition}\label{linfty}
Assume (\ref{constraint}). Then the solution $f$ to the linear Boltzmann equation (\ref{linearf})
satisfies
\begin{equation*}
\|w_{\varrho}f\|_{\infty }+|w_{\varrho}f |_{\infty }\lesssim \|w_{\varrho} \ g\|_{\infty
}+|w_{\varrho}\langle v\rangle r|_{\infty }.
\end{equation*}%
Moreover if $g$ and $r$ are continuous away from the grazing set $\gamma _{0}$, then $f$ is continuous away from $\mathfrak{D}$.  In particular, if $ \ \Omega   $ is convex then $\mathfrak{D}=\gamma_0$.
\end{proposition}

We define
\begin{equation}
\tilde{w}_{\varrho}(v)\equiv \frac{1}{ w_{\varrho, \beta, \zeta} (v)\sqrt{\mu (v)}}=\sqrt{2\pi} \frac{e^{(\frac{1}{4}- \zeta)
|v|^{2}}}{(1+\varrho ^{2}|v|^{2})^{\frac{\beta}{2} }},\label{tweight}
\end{equation}
and $\mathcal{V}(x)=\{v \in\mathbf{R}^3 : n(x)\cdot v >0\}$
with a probability measure $d\sigma=d\sigma(x)$ on $\mathcal{V}(x)$ which is given by
\begin{equation}
d\sigma \equiv\mu
(v)\{n(x)\cdot v\}dv.\label{smeasure}
\end{equation}

We use below Definition \ref{diffusecycles} of stochastic cycles and iterated integral and remind the dependence of $t_k$ on $(t,x,v,v_1,v_2,\dots,v_{k-1})$.
We first show that the set of points in the phase
space $\Pi _{j=1}^{k-1}\mathcal{V}_{j}$ not reaching $t=0$ after $k$ bounces
is small when $k$ is large.
\begin{lemma} \label{k}
For $T_{0}>0$ sufficiently large, there exist constants $%
C_{1},C_{2}>0$ independent of $T_{0}$, such that for $k=C_{1}T_{0}^{5/4}$, and all $(t,x,v)\in [0,T_0]\times\overline{\Omega}\times\mathbf{R}^3,$
\begin{equation}
\int_{\Pi _{j=1}^{k-1}\mathcal{V}_{j}}\mathbf{1}_{\{t_{k}(t,x,\bar{v},\bar{v}_{1},\bar{v}_{2},%
\cdots ,\bar{v}_{k-1})>0\}}\Pi _{j=1}^{k-1}d\sigma _{j}\leq \left\{ \frac{1}{2}%
\right\} ^{C_{2}T_{0}^{5/4}}.  \label{largek}
\end{equation}%
We also have, for $\beta>4$,%
\begin{eqnarray}
&&\int_{\Pi _{j=1}^{k-1}\mathcal{V}_{j}}\sum_{l=1}^{k-1}\mathbf{1}%
_{\{t_{l+1}\leq 0<t_{l}\}} \tilde{w}_{\varrho} (v_{l})\langle v_l \rangle \Pi _{j=1}^{k-1}d\sigma _{j}
\leq  \left\{1+\frac{C_{ \beta,\zeta }}{\varrho^4}\right\} ^{k-1},\label{ktildew}\\
&&\int_{\Pi _{j=1}^{k-1}\mathcal{V}_{j}} \mathbf{1}%
_{\{  0<t_{l+1}\}} \tilde{w}_{\varrho} (v_{l})\langle v_l\rangle \Pi _{j=1}^{k-1}d\sigma _{j}
\leq  \left\{1+\frac{C_{ \beta,\zeta }}{\varrho^4}\right\}, \notag
\end{eqnarray}
for all $l=1,2,\dots, k-1$.
\end{lemma}
\begin{proof}
Choosing $0<\z $ sufficiently small, we further define non-grazing
sets for $1\leq j\leq k-1$ as%
\begin{equation*}
\mathcal{V}_{j}^{\z }=\{v_{j}\in \mathcal{V}_{j}:\text{ }\bar{v}_{j}\cdot
n(x_{j})\geq \z \}\cap \{v_{j}\in \mathcal{V}_{j}:\text{ }|v_{j}|\leq
\frac{1}{\z }\}.
\end{equation*}%
Clearly, by the same argument used in (\ref{epsilonprime1}),
\begin{equation}
\int_{\mathcal{V}_{j}\setminus \mathcal{V}_{j}^{\z }}d\sigma _{j}\leq
\int_{\bar{v}_{j}\cdot n(x_{j})\leq \z }d\sigma _{j}+\int_{|v_{j}|\geq \frac{1%
}{\z }}d\sigma _{j}\leq C\z ,  \label{v-v}
\end{equation}%
where $C$ is independent of $j$.  On the other hand, if $v_{j}\in \mathcal{V}%
_{j}^{\z }$, then from the definition of diffusive back-time cycle (\ref{diffusecycle}),
we have $x_{j}-x_{j+1}=(t_{j}-t_{j+1})\bar{v}_{j}$.  By (\ref{tlower}), since $|v_{j}|\leq \frac{1}{\z }$, and $\bar{v}_{j}\cdot
n(x_{j})\geq \z ,$
\begin{equation*}
(t_{j}-t_{j+1})\geq \frac{\z ^{3}}{C_{\Omega }}.
\end{equation*}%
Therefore, if $t_{k}(t,x,\bar{v},\bar{v}_{1},\bar{v}_{2}...,\bar{v}_{k-1})>0$, then there can be at
most $\left[ \frac{C_{\xi }T_{0}}{\z ^{3}}\right] +1$ number of $v_{j}$ $%
\in \mathcal{V}_{j}^{\z }$ for $1\leq j\leq k-1$.  We therefore have
\begin{eqnarray*}
&&\int_{\mathcal{V}_{1}}...\left\{ \int_{\mathcal{V}_{k-1}}\mathbf{1}_{%
\mathcal{\{}t_{k}>0\}}d\sigma _{k-1}\right\} d\sigma _{k-2}...d\sigma _{1} \\
&\leq &\sum_{m=1}^{\left[ \frac{C_{\xi }T_{0}}{\z ^{3}}\right]
+1}\int_{\{\text{There are exactly }m\text{ of }v_{j_{i}}\in \mathcal{V}%
_{j_{i}}^{\z },\text{ and }k-1-m\text{ of }v_{j_{i}}\notin \mathcal{V}%
_{j_{i}}^{\z }\}}\Pi _{j=1}^{k-1}d\sigma _{j} \\
&\leq &\sum_{m=1}^{\left[ \frac{C_{\xi }T_{0}}{\z ^{3}}\right] +1}\binom{%
k-1}{m}|\sup_{j}\int_{\mathcal{V}_{j}^{\z }}d\sigma _{j}|^{m}\left\{
\sup_{j}\int_{\mathcal{V}_{j}\setminus \mathcal{V}_{j}^{\z }}d\sigma
_{j}\right\} ^{k-m-1}.
\end{eqnarray*}%
Since $d\sigma $ is a probability measure, $\int_{\mathcal{V}_{j}^{\z
}}d\sigma _{j}\leq 1$, and
\begin{equation*}
\left\{ \int_{\mathcal{V}_{j}\setminus \mathcal{V}_{j}^{\z }}d\sigma
_{j}\right\} ^{k-m-1}\leq \left\{ \int_{\mathcal{V}_{j}\setminus \mathcal{V}%
_{j}^{\z }}d\sigma _{j}\right\} ^{k-2-\left[ \frac{C_{\xi }T_{0}}{\z
^{3}}\right] }\leq \{C\z \}^{k-2-\left[ \frac{C_{\xi }T_{0}}{\z ^{3}}%
\right] }.
\end{equation*}%
But, from $\binom{k-1}{m}\leq \{k-1\}^{m}\leq \{k-1\}^{\left[ \frac{C_{\Omega }T_{0}}{%
\z ^{3}}\right] +1}$, we deduce that
\begin{equation}
\int \mathbf{1}_{\{t_{k}>0\}}\Pi _{l=1}^{k-1}d\sigma _{l}\leq {\left[
\frac{C_{\xi }T_{0}}{\z ^{3}}\right] }(k-1)^{\left[
\frac{C_{\xi }T_{0}}{\z ^{3}}\right] +1}C\z ^{k-2-\left[ \frac{%
C_{\xi }T_{0}}{\z ^{3}}\right] }.\label{kdelta}
\end{equation}%
Now let $k-2 = N \{[\frac{C_{\Omega}T_0}{\z^3}]+1\}$, so that if $\frac{C_{\Omega}T_0}{\z^3}\geq 1$, (\ref{kdelta}) can be further majorized by
\begin{eqnarray*}
&& \ \ \left\{ N \left( \frac{C_{\Omega}T_0}{\z^3}+1\right) (C\z)^N\right\}^{\left[\frac{C_{\Omega}T_0}{\z^3}\right]+1}\\
&& \leq \ \left\{ \frac{2NC_{\Omega}T_0}{\z^3} (C\z)^N\right\}^{\left[\frac{C_{\Omega}T_0}{\z^3}\right]+1}\\
&& \leq \ \left\{ C_{N,\Omega} T_0 \z^{N-3}\right\}^{\left[\frac{C_{\Omega}T_0}{\z^3}\right]+1}.
\end{eqnarray*}
We choose $C_{N,\Omega}T_0 \z^{N-3}=\frac{1}{2}$, so that $\z=\{\frac{1}{2C_{N,\Omega}T_0}\}^{\frac{1}{N-3}}$ is small for $T_0$ large and for $N>3$.  Moreover,
\begin{eqnarray*}
\left[\frac{C_{N,\Omega}T_0}{\z^3}\right]+ 1 \sim C_{N,\Omega} T_0^{1+\frac{3}{N-3}},
\end{eqnarray*}
and $\frac{C_{N,\Omega}T_0}{\z^3}\geq 2$, if $T_0$ is large so that we can close our estimate.

Finally we choose $N=15$.  For $T_0$ sufficiently large, $\left[\frac{C_{N,\Omega}T_0}{\z^3}\right] +1 \sim CT_0^{5/4}$ and $k=16\{\left[\frac{C_{\Omega}T_0}{\z^3}\right]+1\}+2 \sim CT_0^{5/4}$, and (\ref{largek}) follows.

Next, we  give the proof of the first estimate in (\ref{ktildew}). The left hand side of (\ref{ktildew}) is bounded by
\begin{eqnarray*}
&& \ \int_{\Pi_{j=1}^{k-1}\mathcal{V}_j} \sum_{l=1}^{k-1} \mathbf{1}_{\{t_{l+1}\leq 0 < t_l\}} \tilde{w}_{\varrho}(v_l) \langle v_l \rangle \Pi_{j=1}^{k-1} d\sigma_j\\
&& \leq \ \int_{\Pi_{j=1}^{k-1}\mathcal{V}_j} \mathbf{1}_{\{t_{k}\leq 0 < t_1\}} 
\tilde{w}_{\varrho}(v_l) \langle v_l \rangle\Pi_{j=1}^{k-1} d\sigma_j
 \\
&&\leq \ \prod_{j=1}^{k-1} \left\{ \int_{\mathcal{V}_j}[1+ \tilde{w}_{\varrho}(v_j)\langle v_j \rangle]  d\sigma_j \right\}  \leq  \left\{1+  \frac{1}{\sqrt{2\pi}} \int_{v_1>0} \frac{ v_1 e^{-(\frac{1}{4}+\zeta)|v|^2} }{(1+\varrho^2 |v|^2)^{\frac{\beta-1}{2}}} dv
 \right\}^{k-1} \\
&&\leq \ \left\{ 1+ \frac{1}{ \sqrt{2\pi} } \int_{\mathfrak{u}_1>0} \frac{ C_{\zeta} \mathfrak{u}_1}{(1+|\mathfrak{u}|^2)^{\frac{\beta-1}{2}}} \varrho^{-4}d\mathfrak{u}
\right\}^{k-1}  \\
&&\leq \ \left\{1+ \frac{C_{\beta,\zeta}}{\varrho^{4}}\right\}^{k-1},
\end{eqnarray*}
where we used the change of variables : $\varrho v=\mathfrak{u}$ and the fact $\beta>4$.

For the second estimate in (\ref{ktildew}) similarly we have
\begin{eqnarray*}
&& \ \int_{\Pi_{j=1}^{k-1}\mathcal{V}_j}   \mathbf{1}_{\{   0 < t_{l+1}\}} \tilde{w}_{\varrho}(v_l) \langle v_l \rangle \Pi_{j=1}^{k-1} d\sigma_j \ \leq \ \int_{\Pi_{j=1}^{k-1}\mathcal{V}_j}     \tilde{w}_{\varrho}(v_l) \langle v_l \rangle \Pi_{j=1}^{k-1} d\sigma_j\\
&&\leq \   \int_{\mathcal{V}_l}[1+ \tilde{w}_{\varrho}(v_l)\langle v_l \rangle]  d\sigma_l    \leq  \left\{1+  \frac{1}{\sqrt{2\pi}} \int_{v_1>0} \frac{ v_1 e^{-(\frac{1}{4}+\zeta)|v|^2} }{(1+\varrho^2 |v|^2)^{\frac{\beta-1}{2}}} dv
 \right\}  \\
&&\leq \ \left\{ 1+ \frac{1}{ \sqrt{2\pi} } \int_{\mathfrak{u}_1>0} \frac{ C_{\zeta} \mathfrak{u}_1}{(1+|\mathfrak{u}|^2)^{\frac{\beta-1}{2}}} \varrho^{-4}d\mathfrak{u}
\right\}   \\
&&\leq \ \left\{1+ \frac{C_{\beta,\zeta}}{\varrho^{4}}\right\} ,
\end{eqnarray*}
where we used the fact that $d\sigma_j$ is a probability measure on $\mathcal{V}_j$.
\end{proof}
\\

Denote $h=w_{\varrho} \ f$ and $K_{w_{\varrho}}( \ \cdot \ )=w_{\varrho} K( \frac {1} {w_{\varrho}} \ \cdot)$.  We present an {\it abstract iteration} scheme which gives a unified way to study both steady and dynamic problem with diffuse boundary condition. Recall $v=(\bar{v},\hat{v})$ and $v_l =(\bar{v}_l,\hat{v}_l)$ from (\ref{v}) and define:
\begin{eqnarray}
|h^{\ell+1}(t,x,v)| &\leq &\mathbf{1}_{t_{1}\leq 0}e^{-\nu
(v)t}|h^{\ell+1}(0,x-t{\bar{v}},v)|  \notag \\
&&+\mathbf{1}_{t_{1}\leq 0}\int_{0}^{t}e^{-\nu
(v)(t-s)}|[K_{w_{\varrho}}h^{\ell}+w_{\varrho}g](s,x-(t-s){\bar{v}},v)|ds  \notag \\
&&+\mathbf{1}_{t_{1}>0}\int_{t_{1}}^{t}e^{-\nu
(v)(t-s)}|[K_{w_{\varrho}}h^{\ell}+w_{\varrho}g](s,x-(t-s){\bar{v}},v)|ds \ \ \ \ \  \ \ \ \ \  \label{iteration} \\
&&+\mathbf{1}_{t_{1}>0}e^{-\nu (v)(t-t_{1})}|w_{\varrho}r(t_1,x_{1},v)|+\frac{e^{-\nu
(v)(t-t_{1})}}{\tilde{w_{\varrho}}(v)}\int_{\prod_{j=1}^{k-1}\mathcal{V}_{j}}|H|,  \notag
\end{eqnarray}%
where $|H|$ is bounded by
\begin{eqnarray}
&&\sum_{l=1}^{k-1}\mathbf{1}_{\{t_{l+1}\leq
0<t_{l}\}}|h^{\ell+1-l}(0,x_{l}-t_{l}{\bar{v}}_{l},v_{l})|d\Sigma _{l}(0)  \label{h1} \\
&&+\sum_{l=1}^{k-1}\int_{0}^{t_l}\mathbf{1}_{\{t_{l+1}\leq
0<t_{l}\}}|[K_{w_{\varrho}}h^{\ell-l}+w_{\varrho} \ g](s,x_{l}-(t_{l}-s){\bar{v}}_{l},v_{l})|d\Sigma
_{l}(s)ds  \ \ \ \  \ \ \ \ \ \ \label{h2} \\
&&+\sum_{l=1}^{k-1}\int_{t_{l+1}}^{t_{l}}\mathbf{1}_{\{0<t_{l}%
\}}|[K_{w_{\varrho}}h^{\ell-l}+w_{\varrho} \ g](s,x_{l}-(t_{l}-s){\bar{v}}_{l},v_{l})|d\Sigma _{l}(s)ds
\label{h3} \\
&&+\sum_{l=1}^{k-1} \mathbf{1}_{\{0<t_{l}\}} d\Sigma _{l}^r
\label{h4} \\
&&+\mathbf{1}_{\{0<t_{k}\}}|h^{\ell+1-k}(t_{k},x_{k},v_{k-1})|d\Sigma
_{k-1}(t_{k}),  \label{h5}
\end{eqnarray}%
and we define
\begin{eqnarray}
d\Sigma _{l} &=&\{\Pi _{j=l+1}^{k-1}d\sigma _{j}\}\times\{\tilde{w}(v_{l})d\sigma
_{l}\}\times \Pi _{j=1}^{l-1} d\sigma _{j}
\label{measure} \\
d\Sigma _{l}(s) &=&\{\Pi _{j=l+1}^{k-1}d\sigma _{j}\}\times\{e^{\nu
(v_{l})(s-t_{l})}\tilde{w}(v_{l})d\sigma _{l}\}\times \Pi
_{j=1}^{l-1}\{{{e^{\nu (v_{j})(t_{j+1}-t_{j})} d\sigma _{j}}}\}
\notag \\
d\Sigma _{l}^r &=&\{\Pi _{j=l+1}^{k-1}d\sigma _{j}\}\times\{e^{\nu
(v_{l})(t_{l+1}-t_{l})}\tilde{w}(v_{l})w(v_{l})r( t_{l+1},x_{l+1},v_{l} ) d\sigma_l\}\notag\\
&& \  \times \Pi
_{j=1}^{l-1}\{{{e^{\nu (v_{j})(t_{j+1}-t_{j})} d\sigma _{j}}}\}.
\notag
\end{eqnarray}
\begin{remark}
For the steady case you can regard the temporal variables $t,s$ as parameters and read $h(t,x,v)=h(x,v)$.  We used the notation in (\ref{v}), 
$v=(v_1,\cdots, v_d ; v_{d+1},\cdots, v_3) = ( \bar{v}; \hat{v})$ again and $\bar{v}\in\mathbf{R}^d$ and $\hat{v}\in\mathbf{R}^{3-d}$ where $d$ is the spatial dimension so that $x\in\overline{\Omega}\subset \mathbf{R}^d$, for $d=1,2,3$.
\end{remark}
\vskip .2cm

\begin{lemma}
\label{iterationlinfty}There exist $\varrho_0>0$ and $C >0$ such that for all $\varrho>\varrho_0, \ \beta>4$, and for $k= \varrho=Ct^{\frac{5}{4}}$,
\begin{eqnarray}
&& \ \ \sup_{0\leq s \leq t }e^{\frac{\nu _{0}}{2}t} \|h^{\ell+1}(s)\|_{\infty} \notag\\
 && \ \leq \  \frac{1}{8%
}\max_{0\leq l \leq 2k}\sup_{0\leq s\leq t}\{e^{\frac{\nu _{0}}{2}%
s}\|h^{\ell-l}(s)\|_{\infty }\}+\max_{0\leq l\leq 2k}\|h^{\ell+1-l}(0)\|_{\infty }\notag
\\
&& \ \ \  + \ \varrho \left[ \sup_{0\leq s\leq t}\left\{ e^{\frac{\nu _{0}}{2}%
s}|w_{\varrho} \ r(s)|_{\infty }\right\} + \sup_{0\leq s\leq t}\left\{ e^{\frac{\nu _{0}}{2%
}s}\left\Vert \frac{w_{\varrho} \ g(s)}{\langle v\rangle }\right\Vert _{\infty }\right\} %
\right] \notag\\
&& \ \ \  + \ C \max_{1\leq l\leq 2k}\int_{0}^{t}   \left\| \frac{h^{\ell-l}(s)}{w_{\varrho}}\right\|_{2}ds.\label{rho}
\end{eqnarray}
Furthermore for $k=\varrho =Ct^{\frac{5}{4}}$
\begin{eqnarray}
&&\ \ |h^{\ell +1}(t,x,v)|  \notag \\
&& \ \leq \ e^{-\nu (v)(t-t_{1})}w_{\varrho }(v)|r(t_{1},x_{1},v)|  \notag \\
&& \ \ \ +\left\{ \frac{C_{\beta ,\rho ,\zeta }}{N}+\varepsilon C_{\rho ,\beta
,N}\right\} \sup_{s}|\langle v\rangle ^{\beta +4}r(s)|_{\infty
}+C_{N,\varepsilon ,\rho }\left\vert \sup_{s,v}|r(s,\cdot ,v)|\right\vert
_{1}  \notag \\
&&\ \ \ +\ \frac{1}{N}\sup_{s,l}\Vert h^{l}(s)\Vert _{\infty }+e^{-\frac{\nu
_{0}}{2}t}\sup_{l}\Vert h^{l}(0)\Vert _{\infty }  \notag \\
&&\ \ \ +\ \varrho \sup_{s}\left\Vert \frac{w_{\varrho }\ g(s)}{\langle
v\rangle }\right\Vert _{\infty }+\ e^{-\frac{\nu _{0}}{2}t}\sup_{l}%
\int_{0}^{t}\left\Vert \frac{h^{l}(s)}{w_{\varrho }}\right\Vert _{2}ds.
\label{rho12}
\end{eqnarray}

\end{lemma}
\begin{remark}
The estimate (\ref{rho12}) is used only in the proof of the singularity formation in Theorem \ref{main1}.
\end{remark}

\begin{proof}
We first prove (\ref{rho}) and then sketch the proof of (\ref{rho12}).

We start with $r$-contribution in (\ref{iteration}) and (\ref{h4}). Clearly the contribution in (\ref{iteration}) is bounded by
$$
e^{-\nu(v)[t-t_1(t,x,v)]}|w_{\varrho}(v)  r(t_1,x_1,v)|.
$$
Since the exponent of $d\Sigma_l^r$ is bounded by $e^{-\nu_0(t-t_{l+1})}$, from (\ref{measure}) and Lemma \ref{k}
\begin{eqnarray}
(\ref{h4}) &\leq& k e^{-\frac{\nu _{0}}{2}t} \frac{1}{\tilde{w}_{\varrho}(v)}\Big\{
1+\frac{C_{\beta, \zeta }}{\varrho ^{4}}\Big\}  \left\vert e^{\frac{\nu _{0}}{2}s}%
\frac{w_{\varrho}r(s)}{\langle v \rangle}\right\vert _{\infty }\label{diff1}\\
&\leq& e^{-\frac{\nu _{0}}{2}t} \frac{ 2\varrho }{\tilde{w}_{\varrho}(v)}  \sup_{0\leq s\leq t}\Big\{e^{\frac{\nu _{0}}{2}s}%
\left\vert\frac{w_{\varrho}r(s)}{\langle v \rangle}\right\vert _{\infty }\Big\},\notag
\end{eqnarray}
for sufficiently large $\varrho>0$ and $k=\varrho=Ct^{\frac{5}{4}}$.

We now turn to the $g$- contribution in (\ref{iteration}), (\ref{h2}) and (\ref{h3}). Rewrite
\begin{equation*}
|w_{\varrho}   g(x_{l}-(t_{l}-s)\bar{v}_{l},v_{l})|=\frac{|w_{\varrho}g(x_{l}-(t_{l}-s){\bar{v}}_{l},v_{l})|}{%
\langle v_{l}\rangle}\times \langle v_{l}\rangle \leq \left\Vert \frac{w_{\varrho}g}{\langle v \rangle }%
\right\Vert _{\infty } \langle v_{l}\rangle .
\end{equation*}%
Since the exponent of $d\Sigma_l(s)$ is bounded by $e^{\nu_0(s-t_1)}$, the $g$-contributions in (\ref{iteration}), (\ref{h2}) and (\ref{h3}) are bounded by
\begin{eqnarray}
&&2\int_{0}^{t}e^{-\nu (v)(t-s)}|w_{\varrho}g (s,x-(t-s)\bar{v},v)|ds \label{g}\\
&&+\frac{e^{-\nu
_{0}(t-t_{1})}}{\tilde{w}_{\varrho}(v)}\int_{0 }^{t}e^{-\nu _{0}(t_{1}-s)}\left\Vert \frac{w_{\varrho}g(s)}{\langle v \rangle}\right\Vert _{\infty }ds\notag\\
&&  \ \ \ \ \ \ \ \ \ \ \ \ \times\sum_{l=1}^{k-1}\bigg\{\int
\mathbf{1}_{\{t_{l+1}\leq 0<t_{l}\}}\tilde{w}_{\varrho}%
 (v_{l})  \langle  v_l \rangle\Pi_{j=1}^{k-1}d\sigma_j\notag \\
 && \ \ \  \ \ \ \ \ \ \ \ \ \ \ \ \  \ \ \ \ \ \ \  +\max_{l}\int  \mathbf{1}_{\{0<t_{l+1}\}}\tilde{w}_{\varrho}%
 (v_{l})  \langle v_l \rangle \Pi_{j=1}^{k-1}d\sigma_j \bigg\}.\notag
\end{eqnarray}
From
\begin{equation}
\frac{1}{\tilde{w}_{\varrho}}\lesssim _{ \beta, \zeta }\varrho ^{\beta },\label{varrho}
\end{equation}
the second line of (\ref{g}) is bounded by
\begin{eqnarray*}
&\lesssim_{\beta,\zeta}& \ \varrho^{\beta} e^{-\nu_0(t-t_1)} \int_0^t e^{-\nu_0(t_1-s)} \left\| \frac{w_{\varrho}g(s)}{\langle v \rangle} \right\|_{\infty} ds\\
&\lesssim& \varrho^{\beta} e^{-\frac{\nu_0}{2}(t-t_1)} \int_0^t e^{-\frac{\nu_0}{2}(t_1-s)-\frac{\nu_0}{2}t_1} ds \sup_{0\leq s\leq t}\left\{ e^{\frac{\nu_0}{2}s} \left\| \frac{w_{\varrho}g(s)}{\langle v \rangle}\right\|_{\infty}\right\} \\
&\lesssim& \varrho^{\beta} e^{-\frac{\nu_0}{2}t} \left\{ \int_0^t e^{-\frac{\nu_0}{2}(t_1-s)}ds \right\}  \sup_{0\leq s\leq t}\left\{ e^{\frac{\nu_0}{2}s} \left\| \frac{w_{\varrho}g(s)}{\langle v \rangle}\right\|_{\infty}\right\}\\
&\lesssim& \varrho^{\beta} e^{-\frac{\nu_0}{2}t} \sup_{0\leq s\leq t}\left\{ e^{\frac{\nu_0}{2}s} \left\| \frac{w_{\varrho}g(s)}{\langle v \rangle}\right\|_{\infty}\right\}.
\end{eqnarray*}
With our choice $k=\varrho=Ct^{\frac{5}{4}}$ and from Lemma \ref{k}, the third and forth line of (\ref{g}) are bounded by
\begin{eqnarray*}
(1+k) \left\{ 1+\frac{C_{\beta,\zeta}}{\varrho^4}\right\}^{k-1} + k \left\{1+\frac{C_{\beta,\zeta}}{\varrho^4}\right\}
\leq 2(1+\varrho)\left\{ 1+ \frac{C_{\beta,\zeta}}{\varrho^4}\right\}^{\varrho}\lesssim_{\beta,\zeta}  (1+\varrho),
\end{eqnarray*}
where we have chosen sufficiently large $\varrho_0>0$ such that $\varrho > \varrho_0$,
\begin{equation}
\left\{1+ \frac{C_{\beta, \zeta }}{\varrho ^{ 4}}\right\}^{\varrho } < 2.\label{varrho2}
\end{equation}
Therefore, the total $g$-contribution is bounded by
\begin{eqnarray*}
&\lesssim_{\beta,\zeta} & (1+\varrho^{\beta})  e^{-\frac{\nu _{0}}{2}t}\sup_{0\leq s\leq t}\left\{ e^{ \frac{\nu _{0}}{2}%
s}\left\Vert \frac{w_{\varrho}g(s)}{\langle v\rangle }\right\Vert _{\infty }\right\} .
\end{eqnarray*}%

Notice that the exponent in $d\Sigma_l(s)$ is bounded by $e^{-\nu_0(t_1-s)}$ and from (\ref{varrho}) and (\ref{tweight}) and from Lemma \ref%
{k} we get
\begin{eqnarray}
&&e^{-\nu (v)t}|h^{\ell+1}(0,x-t\bar{v},v)| \ \mathbf{1}_{t_{1}\leq0}\notag \\
&&+\frac{e^{-\nu
(v)(t-t_{1})}}{\tilde{w}_{\varrho}(v)}
\int_{\prod_{j=1}^{k-1}\mathcal{V}_j} \Big\{ \ \sum_{l=1}^{k-1}\mathbf{1}_{\{t_{l+1}\leq
0<t_{l}\}}|h^{\ell+1-l}(0,x_{l}-t_{l}{\bar{v}}_{l},v_{l})|d\Sigma _{l}(0)\notag\\
&& \ \ \ \ \ \ \ \ \ \ \ \ \ \ \ \  \ \ \ \ \ \ \ \ \ \ \ \ \ \ \  \ \ \ \ \ \ \ \ \ +\mathbf{1}%
_{\{0<t_{k}\}}|h^{\ell+1-k}(t_{k},x_{k},v_{k-1})|d\Sigma _{k-1}(t_{k}) \ \Big\}\notag\\
&\lesssim_{\beta, \zeta }&e^{-\nu_{0}t}\|h^{\ell+1}(0)\|_{\infty }+\varrho ^{ \beta
}e^{-\nu _{0}t}\max_{1\leq l\leq k}\|h^{\ell+1-l}(0)\|_{\infty }\int
\sum_{l=1}^{k-1}\mathbf{1}_{\{t_{l+1}\leq 0<t_{l}\}}d\Sigma _{l}  \notag \\
&&+\varrho ^{ \beta }e^{- \nu _{0}  t}\sup_{0\leq s\leq t}\Big\{e^{ \nu
_{0}  s}\|h^{\ell+1-k}(s)\|_{\infty }\Big\}\int \mathbf{1}_{\{0<t_{k}\}}d\Sigma
_{k-1}  \notag \\
&\lesssim_{\beta, \zeta } &e^{- \nu _{0} t}\bigg\{\|h^{\ell+1}(0)\|_{\infty}+ \varrho^{ \beta }\left\{ 1+\frac{C_{\beta, \zeta }}{\varrho ^{4}}%
\right\} ^{\varrho}\max_{1\leq l\leq k}\|h^{\ell+1-l}(0)\|_{\infty }\nonumber\\
&& \ \ \ \ \ \ \ \ \ \ \ \ \ \ \ \ \ \ \ \ \ \ \ \ \ \ \ + \varrho
^{ \beta  }\left\{ \frac{1}{2}\right\} ^{C_2C^{-1}\varrho}\sup_{0\leq s\leq t}e^{%
 \nu _{0}  s}\|h^{\ell+1-k}(s)\|_{\infty }\bigg\}   \notag \\
&\lesssim_{\beta, \zeta } &e^{-\frac{\nu _{0}}{2}t}\left\{ \max_{0\leq l\leq
k}\|h^{\ell+1-l}(0)\|_{\infty }+\left\{ \frac{1}{2}\right\} ^{ \varrho
}\sup_{0\leq s\leq t}e^{\frac{\nu _{0}}{2}s}\|h^{\ell+1-k}(s)\|_{\infty
}\right\},\nonumber
\end{eqnarray}%
where we have chosen $t=\varrho=Ct^{\frac{5}{4}}$ for sufficiently large $\varrho>0$ but fixed and used (\ref{varrho}).

Now we obtain an upper bound for (\ref{iteration}), (\ref{h1}), (\ref{h2}), (\ref{h3}), (\ref{h4}), (\ref{h5}) as
\begin{eqnarray}
&&  |h^{\ell+1}(t,x,v)| \ \leq \ \mathbf{1}_{t_{1}\leq 0}\int_{0}^{t}e^{-\nu (v)(t-s)}|K_{w_{\varrho}}h^{\ell}(s,x-(t-s)\bar{v},v)|ds\label{hmain}\\
&& \ \ \ \ \ \ \ \ \ \ \ \ \ \ \ \ \ \ \ \ \    +\mathbf{1}_{t_{1}>0}\int_{t_{1}}^{t}e^{-\nu
(v)(t-s)}|K_{w_{\varrho}}h^{\ell}(s,x-(t-s)\bar{v},v)|ds  \notag\\
&& \ \ \ \ \ \ \ +\frac{e^{-\nu (v)(t-t_{1})}}{%
\tilde{w_{\varrho}}(v)}\times\int_{\prod_{j=1}^{k-1}\mathcal{V}_j}\sum_{l=1}^{k-1}\bigg\{   \int_{0}^{t_{l}}\mathbf{1}_{\{t_{l+1}\leq
0<t_{l}\}}|K_{w_{\varrho}}h^{\ell-l}(s,X_{\mathbf{cl}}(s),v_{l})|\notag\\
&& \ \ \ \ \ \ \ \ \ \ \ \ \ \ \ \ \ \ \ \ \ \ \ \ \ \ \ \ \ \ \ \ \ \    +
\int_{t_{l+1}}^{t_{l}}\mathbf{1}_{\{0<t_{l+1}\}}|K_{w_{\varrho}}h^{\ell-l}(s,X_{\mathbf{cl%
}}(s),v_{l})|\bigg\} d\Sigma _{l}(s)ds \notag\\
&& \ \ \ \ \ \ \ \ \ \ \ \ \ \ \ \ \ \ \ \ \   +e^{-\frac{\nu _{0}}{2}t}A_{\ell}(t,x,v),\notag
\end{eqnarray}%
where $A_{\ell}(t,x,v)$ denotes
\begin{eqnarray}
A_{\ell}(t,x,v)&=& e^{\frac{\nu_0}{2}t_1} w_{\varrho}(v) |r(t_1,x_1,v)| + \frac{2\varrho}{\tilde{w}_{\varrho}(v)} \sup_{0\leq s\leq t} \Big\{ e^{\frac{\nu_0}{2}s} \left| \frac{w_{\varrho} \ r(s)}{\langle v \rangle} \right|_{\infty}
\Big\}\notag\\
&&+(1+\varrho^{\beta}) \sup_{0\leq s\leq t} \Big\{ e^{\frac{\nu_0}{2}s} \left\| \frac{w_{\varrho} \ g(s)}{\langle v \rangle} \right\|_{\infty}
\Big\} +  \max_{0\leq l\leq k}\|h^{\ell+1-l}(0)\|_{\infty }\notag\\
&&+\left\{ \frac{1}{2}%
\right\} ^{\varrho }\sup_{0\leq s\leq t} \Big\{ e^{\frac{\nu _{0}}{2}%
s}\|h^{\ell+1-k}(s)\|_{\infty } \Big\}. \ \ \ \  \ \ \  \label{ant}
\end{eqnarray}%
Recall from (\ref{diffusecycle}) that the back-time cycle from $(s,X_{\mathbf{cl}}(s;t,x,v),v^{\prime})$ denotes
$$(t_{1}^{\prime },x_1^{\prime},v_1^{\prime}), \ (t_{2}^{\prime },x_2^{\prime},v_2^{\prime}), \cdots, \ (t_{l^{\prime}}^{\prime },x_{l^{\prime}} ^{\prime},v_ {l^{\prime}}^{\prime}),\cdots.$$
We now iterate (\ref{hmain}) for $\ell-l$ times to get the representation for $h^{\ell-l}$ and then plug in $K_{w_{\varrho}}h^{\ell-l}(s,X_{\mathbf{cl}}(s),v_l)$ to obtain
\begin{eqnarray}
&& \ K_{w_{\varrho}}h^{\ell-l}(s,X_{\mathbf{cl}}(s),v_{l})
 \leq \int_{\mathbf{R}^3} \mathbf{k}%
_{w_{\varrho}}(v_{l},v ^{\prime })|h^{\ell-l}(s,X_{\mathbf{cl}}(s),v ^{\prime
})|dv ^{\prime } \label{411+}\\
&&\leq  \ \ \iint \mathbf{1}_{t_{1}^{\prime }\leq 0}\int_{0}^{s}e^{-\nu (v^{\prime} )(s-s_{1})}\notag\\
&& \ \ \ \ \   \ \ \ \ \ \ \  \ \mathbf{k}_{w_{\varrho}}(v_{l},v ^{\prime })\mathbf{k}%
_{w_{\varrho}}(v ^{\prime },v ^{\prime \prime }) |h^{ \ell-1-l }(s_{1},X_{\mathbf{cl}%
}(s)-(s-s_{1})\bar{v} ^{\prime },v ^{\prime \prime })|ds_{1}dv ^{\prime
}dv ^{\prime \prime }\notag\\
&&  \ \  +  \iint \mathbf{1}_{t_{1}^{\prime }>0}\int_{t_{1}^{\prime }}^{s}e^{-\nu
(v ^{\prime })(s-s_{1})}\notag\\
&& \ \ \ \ \   \ \ \ \ \ \ \  \ \mathbf{k}_{w_{\varrho}}(v_{l},v ^{\prime })\mathbf{k}%
_{w_{\varrho}}(v ^{\prime },v ^{\prime \prime })|h^{ \ell-1-l }(s_{1},X_{\mathbf{cl}%
}(s)-(s-s_{1})\bar{v} ^{\prime },v ^{\prime \prime })|ds_{1}d v ^{\prime
}dv ^{\prime \prime }\notag \\
&& \ \ +  \iint dv^{\prime} dv^{\prime\prime} \int_{\prod_{j=1}^{k-1}\mathcal{V}_{j}^{\prime}}\frac{e^{-\nu (v^{\prime})(s-t_{1}^{\prime})}}{\tilde{w}_{\varrho}(v ^{\prime})} \sum_{l^{\prime
}=1}^{k-1}\int_{0}^{t_{l^{\prime}}^{\prime }}ds_1 \mathbf{1}_{\{t_{l^{\prime
}+1}^{\prime }\leq 0<t_{l^{\prime }}^{\prime }\}} \nonumber\\
&& \ \ \ \ \ \ \ \ \ \ \ \ \mathbf{k}_{w_{\varrho}}(v_{l},v ^{%
\prime })\mathbf{k}_{w_{\varrho}}(v ^{\prime }_{l^{\prime}},v ^{\prime \prime
}) |h^{\ell-1-l-l^{\prime }} (s_{1,}x_{l^{\prime}}^{\prime}+(s_1-t_{l^{\prime}}^{\prime})\bar{v}_{l^{\prime}}^{\prime},v^{\prime\prime})| d\Sigma _{l^{\prime }}(s_{1}) \notag \\
&& \ \   +
 \iint dv^{\prime} dv^{\prime\prime}\int_{\prod_{j=1}^{k-1}\mathcal{V}_{j}^{\prime}}\frac{e^{-\nu (v^{\prime})(s-t_{1}^{\prime})}}{\tilde{w}_{\varrho}(v ^{\prime})}
\sum_{l^{\prime }=1}^{k-1}
\int_{t_{l^{\prime }}^{\prime }}^{t_{l^{\prime }-1}^{\prime }} ds_1 \mathbf{1}%
_{\{t_{l^{\prime}}^{\prime }>0\}}\nonumber\\
&& \ \ \ \ \ \ \ \ \ \ \ \ \mathbf{k}_{w_{\varrho}}(v_{l},v ^{\prime })\mathbf{k}%
_{w_{\varrho}}(v ^{\prime }_{l^{\prime}},v ^{\prime \prime })
| h^{n-1-l-l^{\prime }}
(s_1, x^{\prime}_{l^{\prime}}+(s_1-t^{\prime}_{l^{\prime}})\bar{v}_{l^{\prime}}^{\prime},v ^{\prime \prime })|d\Sigma _{l^{\prime
}}(s_{1})\notag\\
&& \ \ +  \  e^{-\frac{\nu_0}{2}s}\int_{\mathbf{R}^3} \mathbf{k}_{w_{\varrho}}(v_{l},v ^{\prime })A_{\ell-1-l} (s,X_{\mathbf{cl}}(s),v^{\prime})dv ^{\prime }.\label{diff3}
\end{eqnarray}%

The total contributions of $A_{\ell-l-1}$'s in (\ref{hmain}) are obtained via plugging (\ref{411+}) with different $l$'s into (\ref{hmain}). Since $\int \mathbf{k}_{w_{\varrho}}(v_{l},v ^{\prime })dv ^{\prime }<\infty $, the summation of all contributions of $A_{\ell-l-1}$'s leads to the bound
\begin{eqnarray}
&\lesssim_{\beta,\zeta}& 2 A_{\ell-1}(t)\int_{0}^{t}e^{-\nu _{0}(t-s)}e^{-\frac{\nu _{0}}{2}%
s}ds +\varrho^{\beta} e^{-\frac{\nu_0}{2}t}\max_{1\leq l\leq k-1}A_{\ell-l-1}(t)\label{An}\\
&& \ \ \ \times \int_{\prod_{j=1}^{k-1}\mathcal{V}_j} \sum_{l=1}^{k-1}\left\{ \int_0^{t_l} \mathbf{1}%
_{\{t_{l+1}\leq 0<t_{l}\}}+\int_{t_{l+1}}^{t_l} \mathbf{1}_{\{0<t_{l}\}}\right\} e^{-\frac{\nu_0}{2}(t-s)}\tilde{w}_{\varrho}(v_l)d\Sigma
_{l} ds \notag\\
&&+e^{-\frac{\nu _{0}}{2}t}A_{\ell}(t) \notag\\
&\lesssim_{\beta,\zeta} & \varrho^{\beta} e^{-\frac{\nu _{0}}{2}t}\left\{ \left[ 1+\frac{C_{ \beta, \zeta }}{\varrho
^{4}}\right] ^{\varrho}+\left[ \frac{1}{2}\right] ^{ \varrho }\right\}
\max_{0\leq l\leq k}A_{\ell-l}(t) \notag\\
&\lesssim_{\beta,\zeta} &\varrho^{\beta}e^{-\frac{\nu _{0}}{2}t}\left[ \varrho  \sup_{0\leq s\leq
t}\left\{ e^{\frac{\nu _{0}}{2}s}|wr(s)|_{\infty }\right\} + \varrho \sup_{0\leq
s\leq t}\left\{ e^{\frac{\nu _{0}}{2}s}\left\Vert \frac{wg(s)}{\langle
v\rangle }\right\Vert _{\infty }\right\} \right]  \notag \\
&&+e^{-\frac{\nu _{0}}{2}t}\left\{ \max_{1\leq l\leq
2k}\|h^{\ell+1-l}(0)\|_{\infty }+\left\{ \frac{1}{2}\right\} ^{\varrho
}\max_{0\leq l \leq k}\sup_{0\leq s\leq t}e^{\frac{\nu _{0}}{2}s}\|h^{\ell-l+1-k}(s)\|_{\infty
}\right\}.\notag
\end{eqnarray}

To estimate the $h^{\ell-l-1}$ contribution, we first separate $s-s_{1}\leq
\varepsilon $ and $s-s_{1}\geq \varepsilon $.  In the first case, we use the
fact $\int \mathbf{k}_{w_{\varrho} }(p,v)dv<+\infty$ and by (\ref{411+}) to obtain the small contribution
\begin{equation}
\varepsilon e^{-\frac{\nu
_{0}}{2}s}\max_{1\leq l^{\prime }\leq k}\sup_{0\leq s_1\leq s}\{e^{\frac{\nu
_{0}}{2}s_1}\|h^{\ell-l-l^{\prime }}(s_1)\|_{\infty }\}.\label{epsilon1}
\end{equation}
Now we treat the case of $s-s_{1}\geq \varepsilon$.  For any large $N \gg 1$, we can choose a number $m(N)$
to define
\begin{equation}
\mathbf{k}_{m}(v ^{\prime }_{l^{\prime}},v ^{\prime \prime })\equiv \mathbf{1}%
_{|v^{\prime }_{l^{\prime}}-v^{\prime \prime }|\geq \frac{1}{m},|v^{\prime \prime }|\leq
m}\mathbf{k}_{w_{\varrho}  }(v ^{\prime }_{l^{\prime}},v ^{\prime \prime }),\label{km}
\end{equation}
such that
\begin{equation*}
\sup_{ v^{\prime }_{l^{\prime}} }\int_{\mathbf{R}^{3}}|\mathbf{k}_{m}(v ^{\prime
}_{l^{\prime}},v ^{\prime \prime })-\mathbf{k}_{w_{\varrho}  }(v ^{\prime }_{l^{\prime}},v ^{\prime
\prime })|dv ^{\prime \prime }\leq \frac{1}{N}.
\end{equation*}%
We split $\mathbf{k}_{w}=\{\mathbf{k}_{w_{\varrho} }(v ^{\prime }_{l^{\prime}},v ^{\prime
\prime })-\mathbf{k}_{m}(v ^{\prime }_{l^{\prime}},v ^{\prime \prime })\}+\mathbf{k}%
_{m}(v ^{\prime }_{l^{\prime}},v ^{\prime \prime })$, and the first difference
leads to a small contribution in (\ref{411+}) with $k=\varrho,$%
\begin{equation}
\frac{C_{\varrho}}{N}e^{-\frac{\nu _{0}}{2}s}\max_{1\leq l^{\prime }\leq
k}\sup_{0\leq s_1\leq s}\{e^{\frac{\nu _{0}}{2}s_1}\|h^{\ell-l-l^{\prime
}}(s_1)\|_{\infty }\}.\label{epsilon2}
\end{equation}
For the remaining main contribution of $\mathbf{k}_{m}(v ^{\prime
}_{l^{\prime}},v ^{\prime \prime })$, note that
\begin{equation*}
|\mathbf{k}_{m}(v ^{\prime }_{l^{\prime}},v ^{\prime \prime })|\lesssim C_{N}.
\end{equation*}%
We may make a change of variable $y=x^{\prime}_{l^{\prime}}+(s_1-t^{\prime}_{l^{\prime}})\bar{v}_{l^{\prime}}$ ($x_{l^{\prime}}^{\prime}$ does not depend on $ v _l^{\prime}$) so that $\left\vert \frac{%
dy}{d\bar{v}_{l}^{\prime }}\right\vert \geq \varepsilon ^{d}$, for $s-s_1 \geq \varepsilon$ to estimate (using the notation in (\ref{v}))
\begin{eqnarray*}
&&\hskip -.5cm\int_{|v^{\prime \prime }|\leq m}\int_{\mathcal{V}_{l^{\prime}}^{\prime}}|h^{\ell-l-l^{\prime }}(s_{1,}x^{\prime}_{l^{\prime}}+(s_1-t_{l^{\prime}}^{\prime})\bar{v}_{l}^{\prime },v ^{\prime \prime })|
  \frac{e^{-(\frac{1}{4}+ \zeta )|v_{l^{\prime}}^{\prime}|^2}}{(1+\varrho^2 |v^{\prime}_{l^{\prime}}|)^{\beta}} | n(x^{\prime}_{l^{\prime}})\cdot v^{\prime}_{l^{\prime}} |
dv_{l^{\prime}}^{\prime} dv^{\prime
\prime }\\
&&\leq  \int_{\mathbf{R}^{3-d}} e^{-(\frac{1}{4}+ \zeta )|\hat{v}_{l^{\prime}}|^2} d\hat{v}_{l^{\prime}} \int_{|v^{\prime \prime }|\leq m}
dv^{\prime\prime}\\
&& \ \ \ \ \ \ \ \ \ \times\int_{\mathbf{R}^d} h^{\ell-l-l^{\prime}}(s_1, x_{l^{\prime}}^{\prime}+(s_1-t_{l^{\prime}}^{\prime})\bar{v}_{l}^{\prime},v^{\prime\prime}) \mathbf{1}_{\{ x_{l^{\prime}}^{\prime}+(s_1-t_{l^{\prime}}^{\prime})\bar{v}_{l}^{\prime} \in \overline{\Omega} \}} d\bar{v}_{l^{\prime}}^{\prime}
\\
&&\leq \ \ \ \ \frac{1}{\varepsilon ^{d}}\int_{\Omega }\int_{|v^{\prime \prime
}|\leq m}|h^{\ell-l-l^{\prime }}(s_{1,}y,v ^{\prime \prime })|dy dv^{\prime
\prime }  \\
&&\lesssim_{ \varepsilon,m   }   \   \left\|\frac{h^{\ell-l-l^{\prime }}(s_{1})}{w_{\varrho}(v)}%
\right\|_{2}.
\end{eqnarray*}%
Hence the integrand of main contribution is bounded by $\lesssim _{\varepsilon , m
 }\|\frac{h^{\ell-l-l^{\prime }}(s_{1})}{ w_{\varrho}(v)}\|_{2}$.  Rearranging (\ref{411+}) and combining (\ref{ant}), (\ref{epsilon1}) and (\ref{epsilon2}) we have a bound for (\ref{411+}) as
\begin{eqnarray}
&& K_{w_{\varrho} }h^{\ell+1-l}(s,X_{\mathbf{cl}}(s),v_{l})\label{B} \\
&\leq& e^{-\frac{\nu_0}{2}s} [\varepsilon +\frac{C_{K}}{N}+\left\{ \frac{1}{2}\right\} ^{\varrho }] \max_{1\leq l\leq 2k}\sup_{0\leq s_1\leq s}\{e^{\frac{\nu
_{0}}{2}s_1}\|h^{\ell-l}(s_1)\|_{\infty }\} \notag\\
&&+e^{-\frac{\nu _{0}}{2}s}\bigg\{ \varrho
\sup_{0\leq s_1\leq
s}\left[ e^{\frac{\nu _{0}}{2}s_1}\left|\frac{w_{\varrho}r(s_1)}{\langle v \rangle}\right|_{\infty }\right] +\varrho \sup_{0\leq
s_1\leq s}\left[ e^{ \frac{\nu _{0}}{2}s_1}\left\Vert \frac{w_{\varrho}  g(s_1)}{\langle
v\rangle }\right\Vert _{\infty }\right]\notag
\\
&&  \ \ \ \ \ \ \ \ \ \ \ \ \ +\max_{0\leq l\leq
2k}\|h^{\ell+1-l}(0)\|_{\infty } \bigg\} \notag\\
&&+C_{\varepsilon ,m, N}\max_{1\leq l\leq
2k}\int_{0}^{s}\left\|\frac{h^{\ell-l}(s_1)}{w_{\varrho}(v)}\right\|_{2}ds_1 \notag\\
&\equiv &e^{-\frac{\nu _{0}}{2}s}B.\notag
\end{eqnarray}%
By plugging back (\ref{An}) and (\ref{B}) into (\ref{hmain}), we obtain%
\begin{eqnarray*}
&&|h^{\ell+1}(t,x,v)| \\
&\lesssim &B\left\{ \mathbf{1}_{t_{1}\leq 0}\int_{0}^{t}e^{-\nu (v)(t-s)}ds+%
\mathbf{1}_{t_{1}>0}\int_{t_{1}}^{t}e^{-\nu (v)(t-s)}ds\right\}  \\
&&+B\frac{\int e^{-\nu _{0}(t-s)}ds}{\tilde{w}_{\varrho}}\sum_{l=1}^{k-1}\left\{ \int
\mathbf{1}_{\{t_{l+1}\leq 0<t_{l}\}}+\int \mathbf{1}_{\{0<t_{l}\}}\right\}
d\Sigma _{l}+e^{-\frac{\nu _{0}}{2}t}A^{\ell} \\
&\lesssim&e^{-\frac{\nu_0}{2}t} [\varepsilon +\frac{C_{K}}{N}+\left\{ \frac{1}{2}\right\} ^{\varrho }] \max_{1\leq l\leq 2k}\sup_{0\leq s \leq t}\{e^{\frac{\nu
_{0}}{2}s }\|h^{\ell-l}(s )\|_{\infty }\} \notag\\
&&+e^{-\frac{\nu _{0}}{2}t}\bigg\{  {\varrho}
\sup_{0\leq s \leq
t}\left[ e^{\frac{\nu _{0}}{2}s }\left|\frac{w_{\varrho}r(s )}{\langle v \rangle}\right|_{\infty }\right] +\varrho \sup_{0\leq
s \leq t}\left[ e^{ \frac{\nu _{0}}{2}s }\left\Vert \frac{w_{\varrho} g(s )}{\langle
v\rangle }\right\Vert _{\infty }\right]  \notag
\\
&&  \ \ \ \ \ \ \ \ \ \ \ \ \  +\max_{0\leq l\leq
2k}\|h^{\ell+1-l}(0)\|_{\infty } \bigg\} \notag\\
&&+C_{\varepsilon ,m, N}\max_{1\leq l\leq
2k}\int_{0}^{t}\left\|\frac{h^{\ell-l}(s )}{w_{\varrho}(v)}\right\|_{2}ds_1 \notag\\
\end{eqnarray*}
We then deduce our lemma by choosing $t=\varrho=Ct^{5/4}$ and letting $\varrho $ large and fixed, (so are $k$ and
$t$), and then choosing $\varepsilon$ sufficiently small and $N$ sufficiently large. It is trivial to rescale to the case $\varrho=1$ with a different constant depending on $\varrho$.

 Now we sketch the proof of (\ref{rho12}). Since the proof
is similar to (\ref{rho}) we just highlight the differences. The key is to
estimate the $r$-contribution (\ref{h4}) differently, using its weaker $%
L_{x}^{1}(L_{s,v}^{\infty })$ norm. For large $N>0$ we choose $m>0$ such
that
\begin{equation}
\int_{\mathcal{V}}\mathbf{1}_{|v|\geq m}\langle v\rangle \tilde{w}_{\varrho
}(v)d\sigma \leq \frac{1}{N}.  \notag
\end{equation}%
From (\ref{measure}) and Lemma \ref{k}, we bound (\ref{h4}) by splitting $1=
\mathbf{1}_{|v|\geq m}+\mathbf{1}_{|v|<m}$:
\begin{eqnarray}
&&\frac{e^{-\frac{\nu _{0}}{2}t}}{\tilde{w}_{\varrho }(v)}\int
\sum_{l=1}^{k-1}\mathbf{1}_{0<t_{l}}d\sigma _{k-1}\cdots d\sigma _{l+1}
\notag \\
&&\ \ \ \ \ \ \ \ \ \ \times e^{\frac{\nu _{0}}{2}t_{l+1}}\left\vert \frac{%
w_{\varrho }(v_{l})r(t_{l+1})}{\langle v_{l}\rangle }\right\vert \mathbf{1}%
_{|v_{l}|\geq m}\langle v_{l}\rangle \tilde{w}_{\varrho }(v_{l})d\sigma
_{l}d\sigma _{l-1}\cdots d\sigma _{1}  \notag \\
&&+\frac{e^{-\frac{\nu _{0}}{2}t}}{\tilde{w}_{\varrho }(v)}\int
\sum_{l=1}^{k-1}\mathbf{1}_{0<t_{l}}d\sigma _{k-1}\cdots d\sigma _{l+1}e^{%
\frac{\nu _{0}}{2}t_{l+1}}w_{\varrho }(v_{l})|r(t_{l+1})|\mathbf{1}%
_{|v_{l}|<m}\tilde{w}_{\rho }(v_{l})d\sigma _{l}d\sigma _{l-1}\cdots d\sigma
_{1}  \notag \\
&\leq &\frac{k}{N\tilde{w}_{\varrho }(v)}\Big\{1+\frac{C_{\beta ,\zeta }}{%
\varrho ^{4}}\Big\}\sup_{0\leq s\leq t}\left\vert \frac{w_{\varrho }r(s)}{%
\langle v\rangle }\right\vert   \label{rsplit} \\
&&+\frac{kC_{m}}{\tilde{w}_{\varrho }(v)}\int \int_{|v_{l}|\leq
m}|r(t_{l+1},x_{l}-t_{\mathbf{b}}(x_{l},v_{l})\bar{v}_{l},\bar{v}_{l},%
\hat{v})\|n(x_{l})\cdot v_{l}|dv_{l}d\sigma _{l-1}\cdots d\sigma _{1}.
\notag
\end{eqnarray}%
We have used the $d\sigma _{j}$ is a probability measure for $j>l.$ Using
the notation (\ref{v}), $\bar{u}_{l}\equiv \frac{\bar{v}_{l}}{|\bar{v}_{l}|}$
and $t_{\mathbf{b}}(x_{l},v_{l})\bar{v}_{l}=t_{\mathbf{b}}(x_{l},\frac{\bar{v%
}_{l}}{|\bar{v}_{l}|})\frac{\bar{v}_{l}}{|\bar{v}_{l}|}$, and we have an
uppder bound:
\begin{eqnarray}
&&\int_{|v_{l}|\leq m}|r(t_{l+1},x_{l}-t_{\mathbf{b}}(x_{l},v_{l})\bar{v}%
_{l},v)| |n(x_{l})\cdot \bar{v}_{l}|dv_{l}  \notag \\
&\leq &\int_{|\hat{v}_{l}|\leq m}d\hat{v}_{l}\int_{|\bar{v}_{l}|\leq
m}\sup_{s,v}|r(s,x_{l}-t_{\mathbf{b}}(x_{l},v_{l})\bar{v}%
_{l},v)| |n(x_{l})\cdot \bar{u}_{l}\|\bar{v}_{l}|d\bar{v}_{l}  \label{change}
\\
&\lesssim _{m}&\int_{\mathbf{S}^{d-1}}\sup_{s,v}|r(s,x_{l}-t_{\mathbf{b}%
}(x_{l},v_{l})\bar{v}_{l},v)| |n(x_{l})\cdot \bar{u}_{l}|d\bar{u}_{l}.  \notag
\end{eqnarray}%
Since $x_{l}-t_{\mathbf{b}}(x_{l},v_{l})\bar{v}_{l}\in \partial \Omega $ and
$x_{l}\in \partial \Omega ,$ Now apply (\ref{change1}) in Lemma \ref%
{tbsmooth} to have
\begin{eqnarray}
&&\int_{|v_{l}|\leq m}\sup_{s,v}|r(s,x_{l}-t_{\mathbf{b}}(x_{l},v_{l})%
\bar{v}_{l},v)| |n(x_{l})\cdot \bar{u}_{l}|d\bar{u}_{l}\   \notag \\
&\lesssim _{m}&\ \int_{\partial \Omega }\sup_{s,v}|r(s,y,v)|dS(y).
\label{r1l2}
\end{eqnarray}
We now turn to (\ref{hmain}). Now instead of (\ref{ant}), combining (\ref%
{r1l2}) and (\ref{rsplit}) yields
\begin{eqnarray}
&&e^{-\frac{\nu _{0}}{2}t}A_{\ell }(t,x,v)  \notag \\
&=&e^{-\frac{\nu _{0}}{2}(t-t_{1})}w_{\varrho }(v)|r(t_{1},x_{1},v)|  \notag
\\
&&+\frac{C\varrho }{\tilde{w}_{\varrho }(v)}\Big\{1+\frac{C_{\beta ,\zeta }%
}{\varrho ^{4}}\Big\}
\left\{\frac{1}{N} \sup_{s}\left\vert \frac{w_{\varrho }r(s)}{\langle
v\rangle }\right\vert _{\infty }
+C_{N}\left\vert
\sup_{s,v}|r(s,\cdot ,v)|\right\vert _{1}  \right\}\notag \\
&&+C_{\beta ,\zeta }(1+\varrho ^{\beta })
\sup_{s}\left\Vert \frac{w_{\varrho
}\ g(s)}{\langle v\rangle }\right\Vert _{\infty }
+e^{-\frac{\nu _{0}}{2}%
t}\sup_{l}\Vert h^{l}(0)\Vert _{\infty }
+C_{\beta ,\zeta }\left\{ \frac{1}{2}%
\right\} ^{\varrho }\sup_{s,l}\Vert h^{l}(s)\Vert _{\infty }  \notag \\
&\equiv &e^{-\frac{\nu _{0}}{2}(t-t_{1})}w_{\varrho
}(v)|r(t_{1},x_{1},v)|+D.\ \ \ \ \ \ \   \label{diff2}
\end{eqnarray}%
We now plug (\ref{diff2}) back into (\ref{diff3}), the estimate for (\ref{411+}), for each $\ell$. Now we obtain the upper bound for the last term of (\ref{diff3}) using the estimate of $A_{\ell-1-l}(s,X_{\mathbf{cl}}(s),v^{\prime })$ from (\ref{diff2}) as
\begin{equation*}
e^{-\frac{\nu _{0}}{2}s}\int_{\mathbf{R}^{3}}\mathbf{k}_{w_{\varrho
}}^{\prime }e^{\frac{\nu _{0}}{2}t_{1}^{\prime }}w_{\varrho }(v^{\prime
})|r(t_{1}^{\prime },x_{1}^{\prime },v^{\prime })|dv^{\prime }+e^{-\frac{\nu
_{0}}{2}s}C_{K}D.
\end{equation*}%
where
\begin{equation}
t_{1}^{\prime }=t_{\mathbf{b}}(X_{\mathbf{cl}}(s),\bar{v}^{\prime })
\label{prime}
\end{equation}%
is the exit time of the point $(X_{\mathbf{cl}}(s),\bar{v}^{\prime })$ and 
\begin{equation}\label{primeprime}x_{1}^{\prime }=X_{\mathbf{cl}}(s)-t_{\mathbf{b}}(X_{\mathbf{cl}}(s),%
\bar{v}^{\prime })\bar{v}^{\prime },\text{ \ }\mathbf{k}_{w_{\varrho
}}^{\prime }=\mathbf{k}_{w_{\varrho }}(V_{\mathbf{cl}}(s),v^{\prime }).\end{equation} We
further plug this bound back into (\ref{411+}) and then back into (\ref%
{hmain}) to isolate \textit{only} the contribution of $A_{\ell -1-l}$.
Following exactly the same steps, it suffices to control
\begin{eqnarray}
&&\ \mathbf{1}_{t_{1}\leq 0}\int_{0}^{t}e^{-\nu (v)(t-s)}e^{-\frac{\nu _{0}}{%
2}s}\int \mathbf{k}_{w_{\varrho }}^{\prime }e^{\frac{\nu _{0}}{2}%
t_{1}^{\prime }}w_{\varrho }(v^{\prime })|r(t_{1}^{\prime },x_{1}^{\prime
},v^{\prime })|dv^{\prime }ds  \label{hrmain} \\
&&\ \ \ \ \ \ \ \ \ \ \ \ \ \ \ \ \ \ \ \ \ +\mathbf{1}_{t_{1}>0}%
\int_{t_{1}}^{t}e^{-\nu (v)(t-s)}\int \mathbf{k}_{w_{\varrho }}^{\prime }e^{%
\frac{\nu _{0}}{2}t_{1}^{\prime }}w_{\varrho }(v^{\prime })|r(t_{1}^{\prime
},x_{1}^{\prime },v^{\prime })|dv^{\prime }ds  \notag \\
&&\ \ \ \ \ \ \ +\frac{e^{-\nu (v)(t-t_{1})}}{\tilde{w_{\varrho }}(v)}\times
\int_{\prod_{j=1}^{k-1}\mathcal{V}_{j}}\sum_{l=1}^{k-1}\bigg\{%
\int_{0}^{t_{l}}\mathbf{1}_{\{t_{l+1}\leq 0<t_{l}\}}\int \mathbf{k}%
_{w_{\varrho }}^{\prime }e^{\frac{\nu _{0}}{2}t_{1}^{\prime }}w_{\varrho
}(v^{\prime })|r(t_{1}^{\prime },x_{1}^{\prime },v^{\prime })|  \notag \\
&&\ \ \ \ \ \ \ \ \ \ \ \ \ \ \ \ \ \ \ \ \ \ \ \ \ \ \ \ \ \ \ \ \ \
+\int_{t_{l+1}}^{t_{l}}\mathbf{1}_{\{0<t_{l+1}\}}\int \mathbf{k}_{w_{\varrho
}}^{\prime }e^{\frac{\nu _{0}}{2}t_{1}^{\prime }}w_{\varrho }(v^{\prime
})|r(t_{1}^{\prime },x_{1}^{\prime },v^{\prime })|\bigg\}d\Sigma _{l}(s)ds
\notag \\
&&\ \ \ \ \ \ \ \ \ \ \ \ \ \ \ \ \ \ \ \ \ +e^{-\frac{\nu _{0}}{2}%
(t-t_{1})}w_{\varrho }(v)|r(t_{1},x_{1},v)|+D+C_{\beta ,\zeta }\varrho ^{\beta
}\left\{ [1+\frac{C_{\beta ,\zeta }}{\varrho ^{4}}]^{\varrho }+\frac{1}{2^{\varrho}}%
\right\} D,  \notag
\end{eqnarray}%
by Lemma \ref{k}. Note $\varrho ^{\beta }\left\{ [1+\frac{C_{\beta ,\zeta }}{%
\varrho ^{4}}]^{\varrho }+\frac{1}{2^{\varrho }}\right\} \lesssim 1.$ We shall use a
sequence of approximations to estimate the above integrals. Let $\mathbf{k}%
_{w_{\varrho }}^{\prime }\mathbf{1}_{|v^{\prime }-V_{\mathbf{cl}}(s)|\geq
\frac{1}{m},|v^{\prime }-V_{\mathbf{cl}}(s)|\leq m}=\mathbf{k}_{w_{\varrho
}}^{m}.$ We again split for $m$ large
\begin{equation*}
\int |\mathbf{k}_{w_{\varrho }}^{\prime }-\mathbf{k}_{w_{\varrho }}^{m}|\leq
\frac{1}{N}.
\end{equation*}%
We therefore bound%
\begin{eqnarray}
&&\ \mathbf{1}_{t_{1}\leq 0}\int_{0}^{t}e^{-\nu (v)(t-s)}e^{-\frac{\nu _{0}}{%
2}s}\int \mathbf{k}_{w_{\varrho }}^{m}e^{\frac{\nu _{0}}{2}t_{1}^{\prime
}}w_{\varrho }(v^{\prime })|r(t_{1}^{\prime },x_{1}^{\prime },v^{\prime
})|dv^{\prime }ds  \notag \\
&&\ \ \ \ \ \ \ \ \ \ \ \ \ \ \ \ \ \ \ \ \ +\mathbf{1}_{t_{1}>0}%
\int_{t_{1}}^{t}e^{-\nu (v)(t-s)}\int \mathbf{k}_{w_{\varrho }}^{m}e^{\frac{%
\nu _{0}}{2}t_{1}^{\prime }}w_{\varrho }(v^{\prime })|r(t_{1}^{\prime
},x_{1}^{\prime },v^{\prime })|ds  \notag \\
&&\ \ \ \ \ \ \ +\frac{e^{-\nu (v)(t-t_{1})}}{\tilde{w_{\varrho }}(v)}\times
\int_{\prod_{j=1}^{k-1}\mathcal{V}_{j}}\sum_{l=1}^{k-1}\bigg\{%
\int_{0}^{t_{l}}\mathbf{1}_{\{t_{l+1}\leq 0<t_{l}\}}\int \mathbf{k}%
_{w_{\varrho }}^{m}e^{\frac{\nu _{0}}{2}t_{1}^{\prime }}w_{\varrho
}(v^{\prime })|r(t_{1}^{\prime },x_{1}^{\prime },v^{\prime })|  \notag \\
&&\ \ \ \ \ \ \ \ \ \ \ \ \ \ \ \ \ \ \ \ \ \ \ \ \ \ \ \ \ \ \ \ \ \
+\int_{t_{l+1}}^{t_{l}}\mathbf{1}_{\{0<t_{l+1}\}}\int \mathbf{k}_{w_{\varrho
}}^{m}e^{\frac{\nu _{0}}{2}t_{1}^{\prime }}w_{\varrho }(v^{\prime
})|r(t_{1}^{\prime },x_{1}^{\prime },v^{\prime })|\bigg\}d\Sigma _{l}(s)ds
\notag \\
&&+\frac{C_{\beta ,\varrho ,k,\zeta }}{N}\sup_{s}|w_{\varrho }r(s)|_{\infty }.
\label{ksplit}
\end{eqnarray}%
Since $\mathbf{k}_{w_{\varrho }}^{m}$ is bounded, and $w_{\varrho
}(v^{\prime })\lesssim _{\varrho }\langle v^{\prime }\rangle ^{\beta },$ we
further split $|v^{\prime }|\geq N$ and $|v^{\prime }|\leq N.$ Since $%
\int_{|v|\geq N}\langle v\rangle ^{-4}\lesssim \frac{1}{N},$ up to a
constant of $C_{N,\varrho ,k}$ the main integrals in (\ref{ksplit}) are bounded
by:%
\begin{eqnarray}
&&\ \mathbf{1}_{t_{1}\leq 0}\int_{0}^{t}\int_{\Pi _{m}}\langle v^{\prime
}\rangle ^{\beta }|r(t_{1}^{\prime },x_{1}^{\prime },v^{\prime })|dv^{\prime
}ds+\mathbf{1}_{t_{1}>0}\int_{t_{1}}^{t}\int_{\Pi _{m}}\langle v^{\prime
}\rangle ^{\beta }|r(t_{1}^{\prime },x_{1}^{\prime },v^{\prime })|dv^{\prime
}ds  \notag
\\
&&\ \ \ \ \ \ \ +\int_{\prod_{j=1}^{k-1}\mathcal{V}_{j}}\sum_{l=1}^{k-1}%
\bigg\{\int_{0}^{t_{l}}\mathbf{1}_{\{t_{l+1}\leq 0<t_{l}\}}\int_{\Pi
_{m}}\langle v^{\prime }\rangle ^{\beta }|r(t_{1}^{\prime },x_{1}^{\prime
},v^{\prime })|  \notag \\
&&\ \ \ \ \ \ \ \ \ \ \ \ \ \ \ \ \ \ \ \ \ \ \ \ \ \ \ \ \ \ \ \ \ \
+\int_{t_{l+1}}^{t_{l}}\mathbf{1}_{\{0<t_{l+1}\}}\int_{\Pi _{m}}\langle
v^{\prime }\rangle ^{\beta }|r(t_{1}^{\prime },x_{1}^{\prime },v^{\prime })|%
\bigg\}d\Sigma _{l}(s)ds  \notag \\
&&+\frac{1}{N}\sup_{s,v,x}\langle v\rangle ^{\beta +4}|r(s,x,v)|,
\label{nsplit}
\end{eqnarray}%
where%
\begin{equation}
v^{\prime}\in\Pi _{m} \ \ \ \text{iff} \ \ \ \ |v^{\prime }-V_{\mathbf{cl}}(s)|\geq \frac{1}{m},\text{ \ and \ }%
|v^{\prime }-V_{\mathbf{cl}}(s)|\leq m,\text{ \ and \ }|v^{\prime }|\leq N.
\label{pi}
\end{equation}%
Lastly, for any $\varepsilon >0,$ we further split the time intervals of the
main integrals in (\ref{nsplit}) to obtain:%
\begin{eqnarray}
&&\ \mathbf{1}_{t_{1}\leq 0}\int_{\varepsilon }^{t-\varepsilon }\int_{\Pi
_{N}}{}\langle v^{\prime }\rangle ^{\beta }|r(t_{1}^{\prime },x_{1}^{\prime
},v^{\prime })|dv^{\prime }ds+\mathbf{1}_{t_{1}>0}\int_{t_{1}+\varepsilon
}^{t-\varepsilon }\int_{\Pi _{N}}\langle v^{\prime }\rangle ^{\beta
}|r(t_{1}^{\prime },x_{1}^{\prime },v^{\prime })|dv^{\prime }ds  \notag \\
&&\ \ \ \ \ \ \ +\int_{\prod_{j=1}^{k-1}\mathcal{V}_{j}}\sum_{l=1}^{k-1}%
\bigg\{\int_{\varepsilon }^{t_{l}-\varepsilon }\mathbf{1}_{\{t_{l+1}\leq
0<t_{l}\}}\int_{\Pi _{N}}\langle v^{\prime }\rangle ^{\beta
}|r(t_{1}^{\prime },x_{1}^{\prime },v^{\prime })|  \notag \\
&&\ \ \ \ \ \ \ \ \ \ \ \ \ \ \ \ \ \ \ \ \ \ \ \ \ \ \ \ \ \ \ \ \ \
+\int_{t_{l+1}+\varepsilon }^{t_{l}-\varepsilon }\mathbf{1}%
_{\{0<t_{l+1}\}}\int_{\Pi _{N}}\langle v^{\prime }\rangle ^{\beta
}|r(t_{1}^{\prime },x_{1}^{\prime },v^{\prime })|\bigg\}d\Sigma
_{l}(s)dv^{\prime }ds  \notag \\
&&+\varepsilon C_{\varrho,k,\beta }\sup_{s,x,v}\langle v\rangle ^{\beta
}|r(s,x,v)|.  \label{esplit}
\end{eqnarray}%
We now are ready to use change of variables to estimate the main $v^{\prime
}$-integrals in (\ref{esplit}). Recall that, by the definition (\ref{primeprime}) of $x_{1}^{\prime }$  
and Definition (\ref{diffusecycles}), $X_{\mathbf{cl}}(s)$ reaches $\partial
\Omega $ if and only if $s$ at these $t,t_{1},t_{2},...t_{l+1}.$ From our
splitting of time intervals, there exists $c_{\varepsilon }>0$ such that
\begin{equation*}
\text{dist}(X_{\mathbf{cl}}(s),\partial \Omega )>c_{\varepsilon }>0,
\end{equation*}%
in the integrals in (\ref{esplit}), uniformly in $\Pi _{N}.$ We now repeat
the change of variables $x_{1}^{\prime }\rightarrow \frac{\bar{v}^{\prime }}{%
|\bar{v}^{\prime }|}$ as (\ref{change}) and (\ref{r1l2}), but using (\ref{change2}) instead of (\ref{change1}). We finally bound (\ref{esplit}) by%
\begin{equation}
C_{N,\varepsilon ,k,\varrho }\left\vert \sup_{s,v}|r(s,\cdot ,v)|\right\vert
_{1}+\varepsilon C_{\varrho ,k,\beta ,N}\sup_{s,x,v}\langle v\rangle ^{\beta
}|r(s,x,v)|.  \label{rchange}
\end{equation}%
Collecting and combining (\ref{diff2}), (\ref{hrmain}), (\ref{ksplit}), (\ref%
{esplit}) and (\ref{rchange}), we conclude the $r$ contribution in (\ref%
{hmain}) is bounded by%
\begin{eqnarray*}
&&e^{-\frac{\nu _{0}}{2}(t-t_{1})}w_{\varrho }(v)|r(t_{1},x_{1},v)|+C_{\varrho
}D+\left\{ \frac{C_{\beta ,\varrho ,k,\zeta }}{N}+\varepsilon C_{\varrho ,k,\beta
,N}\right\} \sup_{s}|\langle v\rangle ^{\beta +4}r(s)|_{\infty } \\
&&+C_{N,\varepsilon ,k,\varrho }\left\vert \sup_{s,v}|r(s,\cdot ,v)|\right\vert
_{1}.
\end{eqnarray*}%
We deduce (\ref{rho12}) by recalling $D$ defined in (\ref{diff2}).
\end{proof}

\bigskip

After preparing the above tools we return to the stationary problem to give the proof of Proposition \ref{linfty}.

\vskip .2cm
\begin{proof}[Proof of Proposition \ref{linfty}]
We use the exactly same approximation (\ref{approximate}) to establish the proposition and follow the same steps in the proof of Proposition \ref{linearl2}. We denote $h^{\ell+1}= w_{\varrho} \ f^{\ell+1} $, and rewrite (\ref%
{approximate}) as%
\begin{eqnarray}
\varepsilon h^{\ell+1}+v\cdot \nabla _{x}h^{\ell+1}+\nu h^{\ell+1} &=&K_{w_{\varrho}}h^{\ell}+w_{\varrho} g,
\label{happroximate} \\
h_{ {-}}^{\ell+1} &=&\frac{1-\frac{1}{j}}{\tilde{w}_{\varrho}(v)}%
\int_{n(x)\cdot v^{\prime }>0}h ^{\ell}(t,x,v^{\prime })\tilde{w}%
_{\varrho}(v^{\prime })d\sigma +w_{\varrho} \ r.  \notag
\end{eqnarray}
\noindent{\it Step 1}:  We take $\ell\rightarrow \infty $ in $L^{\infty }$.  Upon
integrating over the characteristic lines $\frac{dx}{dt}=v$, and $\frac{dv}{%
dt}=0$ and using the boundary condition repeatedly, we obtain (by replacing $1$ with $1-%
\frac{1}{j}$ and $\nu$ with $\nu+\varepsilon$) that the abstract iteration (\ref{iteration}) is valid
for stationary $h^{\ell+1}(s,x,v)=h^{\ell+1}(x,v), \
g(s,x,v)=g(x,v)$ and $r(s,x,v)=r(x,v)$.  Therefore, for $\ell\geq 2k$, by Lemma \ref{iterationlinfty}, we get (choosing $k=\varrho=C t^{5/4}$ large)
\begin{eqnarray*}
\|h^{\ell+1}\|_{\infty } &\leq &\frac{1}{8}\max_{1\leq l\leq
2k}\{\|h^{\ell-l}\|_{\infty }\}+e^{-\frac{\nu _{0}}{2}t}\max_{0\leq l\leq
2k}\|h^{\ell+1-l}\|_{\infty } \\
&&+\varrho \left[    |w_{\varrho} \ r|_{\infty }  +   \left\Vert \frac{w_{\varrho} \ g}{\langle v\rangle }\right\Vert _{\infty
}  \right] +C_k \max_{1\leq l\leq 2k}\left\| f^{\ell-l} \right\|_{2}.
\end{eqnarray*}
Then, absorbing $e^{-\frac{\nu_0}{2}t}\|h^{\ell+1}\|_{\infty}$ in the left hand side, we have
\begin{eqnarray*}
\|h^{\ell+1}\|_{
\infty}&\leq &\frac{1}{4}\max_{1\leq l\leq 2k}\{\|h^{\ell+1-l}\|_{\infty }\}+C_{k}%
\left[   |w_{\varrho} \ r|_{\infty }  +  \left\Vert \frac{w_{\varrho} \ g}{\langle v\rangle }\right\Vert _{\infty
}  \right]  \\
&&+C_{k}\max_{1\leq l\leq 2k}\left\| f^{\ell-l} \right\|_{2}.
\end{eqnarray*}%
Now this is valid for all $\ell\geq 2k$.  By  induction on $\ell$, we
can iterate such bound for $\ell+2,....\ell+2k$ to obtain%
\begin{eqnarray*}
&&\hskip -.8cm\|h^{\ell+i}\|_{\infty } \leq \frac{1}{4}\max_{1\leq l\leq
2k}\{\|h^{\ell+i-l}\|_{\infty }\}+C_{k}\left[ |w_{\varrho} \ r|_{\infty }+\left\Vert \frac{%
w_{\varrho} \ g}{\langle v\rangle }\right\Vert _{\infty }+\max_{-2k\leq l\leq
2k}\left\| f^{\ell+l} \right\|_{2}\right]  \\
&\leq &\frac{1}{4}\max_{1\leq l\leq 2k}\{\|h^{\ell-1+i-l}\|_{\infty }\}+2C_{k}%
\left[ |w_{\varrho} \ r|_{\infty }+\left\Vert \frac{w_{\varrho} \ g}{\langle v\rangle }\right\Vert
_{\infty }+\max_{-2k\leq l\leq 2k}\left\| f^{\ell+l} \right\|_{2}\right]  \\
&\vdots& \\
&\leq &\frac{1}{4}\max_{1\leq l\leq 2k}\{\|h^{\ell-l}\|_{\infty }\}+(2i+1)C_{k}%
\left[ |w_{\varrho} \ r|_{\infty }+\left\Vert \frac{w_{\varrho} \ g}{\langle v\rangle }\right\Vert
_{\infty }+ \max_{-2k\leq l\leq 2k}\left\| f^{\ell+l} \right\|_{2} \right].
\end{eqnarray*}%
We now take a maximum over $1\leq i \leq 2k$ to get
\begin{eqnarray}
\max_{1\leq l\leq 2k}\|h^{\ell+1-l}\|_{\infty }
 &\leq&  \frac{1}{4}\max_{1\leq l\leq 2k}\|h^{\ell-l}\|_{\infty }\nonumber\\
 &+&(4k+1)C_{k}%
\left[ |w_{\varrho} \ r|_{\infty }+\left\Vert \frac{w_{\varrho} \ g}{\langle v\rangle }\right\Vert
_{\infty }+\max_{-2k\leq l\leq 2k} \| {f^{\ell+l}}  \|_{2} \right].\notag\\
\label{hn1l}
\end{eqnarray}%
This implies that, from Lemma \ref{iterationlinfty} for all $\ell\geq 2k,$
\begin{equation}
\max_{1\leq l\leq 2k}\|h^{\ell+1-l}\|_{\infty }\lesssim _{k}\left[ \max_{1\leq
l\leq 2k}\|h^{l}\|_{\infty }+|w_{\varrho} \ r|_{\infty }+\sup_{0\leq s\leq t}\left\Vert
\frac{w_{\varrho} \ g}{\langle v\rangle }\right\Vert _{\infty }+\max_{1\leq l\leq
\ell}\|f^{l}\|_{2}\right] .\label{n1-l}
\end{equation}%
Now in order to control $\max_{1\leq l\leq 2k}\|h^{l}\|_{\infty }$, we can use (\ref%
{iteration}) repeatedly for $h^{2k}\rightarrow h^{2k-1\text{ }%
}\cdots\rightarrow h^{0}$ to obtain, by Lemma \ref{k},
\begin{equation}
\max_{1\leq l\leq 2k}\|h^{l}\|_{\infty }\lesssim _{k}|w_{\varrho} \ r|_{\infty
}+ \left\Vert \frac{w_{\varrho} \ g}{\langle v\rangle }\right\Vert
_{\infty }.\label{l2k}
\end{equation}%
We therefore conclude that, from (\ref{n1-l}) and (\ref{l2k}),
\begin{eqnarray*}
\max_{1\leq l\leq 2k}\|h^{\ell+1-l}\|_{\infty } & \lesssim_{k}  &\left[
\|h^{0}\|_{\infty }+|w_{\varrho} \ r|_{\infty }+\sup_{0\leq s\leq t}\left\Vert \frac{w_{\varrho} \ g}{%
\langle v\rangle }\right\Vert _{\infty }+\max_{1\leq l\leq \ell}\|f^{l}\|_{2}%
\right].
\end{eqnarray*}%
But $\max_{1\leq l\leq \infty }\|f^{l}\|_{2}$ is bounded by step 1 in the
proof of Lemma \ref{lineareb}. Furthermore for $\beta>4$, $\|\cdot \|_{2}$ is bounded by $%
\|w_{\varrho}\cdot \|_{\infty }$ and $|r|_{2}$ is bounded by $%
|w_{\varrho}\langle v\rangle r|_{\infty }$.  Hence we have
\begin{eqnarray*}
\max_{1\leq l\leq 2k}\|h^{\ell+1-l}\|_{\infty }& \lesssim_{k}  &\left[ \|h^{0}\|_{\infty }+|w_{\varrho} \ r|_{\infty }+\sup_{0\leq s\leq
t}\left\Vert \frac{w_{\varrho} \ g}{\langle v\rangle }\right\Vert _{\infty }\right].
\end{eqnarray*}%
Therefore there exists a limit (unique) solution $%
h^{\ell}\rightarrow h=w_{\varrho} \ f\in L^{\infty }$.  Furthermore, $h$ satisfies (\ref%
{iteration}) with $h^{\ell+1}\equiv h$.

By subtracting $h^{\ell+1}-h$ in (\ref%
{iteration}) with $r=g=0$, we obtain from Lemma \ref{iterationlinfty}, for $h^{\ell+1}-h$, (choosing $k=\varrho=Ct^{5/4}$ large)
\begin{eqnarray*}
\|h^{\ell+1}-h\|_{\infty } &\leq & \frac{1}{8} \max_{0\leq l \leq 2k} \| h^{\ell-l}-h \|_{\infty} + e^{-\frac{\nu_0}{2}k} \max_{0\leq l \leq 2k} \| h^{\ell+1-l}-h \|_{\infty}\\
&&+ C_k \max_{1\leq l \leq 2k} \| f^{\ell-l}-f\|_2\\
&\leq& \frac{1}{4} \max_{1\leq l \leq 2k} \| h^{\ell-l}-h \|_{\infty} + C_k \max_{1\leq l \leq 2k} \| f^{\ell-l}-f\|_2,
\end{eqnarray*}
where $e^{-\frac{\nu_0}{2}k} \|h^{\ell+1}-h\|_{\infty}$ is absorbed in the left hand side.
\begin{eqnarray*}
&\leq &\frac{1}{4^{\ell/2k}}\max_{1\leq l\leq 2k}\{\|h^{l}-h\|_{\infty }\}+C_{k}%
\left[ \max_{-2k\leq l\leq 2k}\|f^{\ell+l}-f\|_{2}\right]  \\
&\leq &\frac{C_{k}(|w_{\varrho} r|_{\infty }+\sup_{0\leq s\leq t}\left\Vert \frac{w_{\varrho} g}{%
\langle v\rangle }\right\Vert _{\infty })}{4^{\ell/2k}}+C_{k}\left[
\max_{-2k\leq l\leq 2k}\|f^{\ell+l}-f\|_{2}\right] .
\end{eqnarray*}%
From step 1 of Lemma \ref{lineareb}, we deduce $\|h^{\ell+1-l}-h\|_{\infty
}\rightarrow 0$ for $\ell$ large.
\vskip .3cm
\noindent{\it Step 2:} Now we let $j\rightarrow \infty$.  We take $f^{j}$ to be the solution
to (\ref{jequation}) and integrate along $\frac{dx}{dt}=v,\frac{dv}{dt}=0$
repeatedly. We establish (\ref{iteration}) for $h^{\ell}\equiv h^{j}$
(we may replace $(1-\frac{1}{j})$ by $1$ and $\varepsilon =0$ to preserve inequality) and $l=0$.
Lemma \ref{iterationlinfty} implies
\begin{equation*}
\|h^{j}\|_{\infty }\leq \frac{1}{8} \|h^{j}\|_{\infty } +e^{-\frac{\nu _{0}%
}{2}k}\|h^{j}\|_{\infty }+\varrho^{1+4\beta}\left[ |w_{\varrho} r|_{\infty }+\left\Vert \frac{w_{\varrho} g}{%
\langle v\rangle }\right\Vert _{\infty }\right] +C(k)\|f^{j}\|_{2} ,
\end{equation*}%
so that
\begin{equation*}
\|h^{j}\|_{\infty }\leq C_{k}\left[ |w_{\varrho} r|_{\infty }+\left\Vert \frac{w_{\varrho} g}{%
\langle v\rangle }\right\Vert _{\infty }\right] +C_{k}\|f^{j}\|_{2}.
\end{equation*}%
Since $\|f^{j}\|_{2}$ is bounded, this implies that $\|h^{j}\|_{\infty }$ is
uniformly bounded and we obtain a (unique) solution $h=wf\in L^{\infty }$.
Taking the difference, we have
\begin{eqnarray*}
&&\varepsilon \lbrack h^{j}-h]+v\cdot \nabla _{x}[h^{j}-h]+\nu \lbrack
h^{j}-h] =K_{w}[h^{j}-h], \\
&& \ \ \ \ \ \ \ \ h_{ {-}}^{j}-h_{ {-}} =\frac{1}{\tilde{w}_{\varrho}(v)}%
\int_{n(x)\cdot v^{\prime }>0}[h^{j}-h](t,x,v^{\prime }) \tilde{w}_{\varrho}(v^{\prime}) d\sigma(v^{\prime}) \\
&& \ \ \ \ \ \ \ \ \ \ \ \ \ \ \ \ \ \ \ \ \ \ -\frac{1}{j}\frac{1}{\tilde{w}_{\varrho}(v)}\int_{n(x)\cdot v^{\prime
}>0}h ^{j}(t,x,v^{\prime })\tilde{w}_{\varrho}(v^{\prime}) d\sigma(v^{\prime})   .
\end{eqnarray*}%
We regard $-\frac{1}{j}\frac{1}{\tilde{w}_{\varrho}(v)}\int_{n(x)\cdot v^{\prime
}>0}h (t,x,v^{\prime })\tilde{w}_{\varrho}(v^{\prime })(n(x)\cdot v^{\prime})dv^{\prime}  =r$.
So Lemma \ref{iterationlinfty} implies that
\begin{equation*}
\|h^{j}-h\|_{\infty }\leq \frac{1}{4}\{\|h^{j}-h\|_{\infty }\}+\frac{1}{j}%
 |h^{j} |_{\infty }+C_{k}\left[ \|f^{j}-f\|_{2}\right],
\end{equation*}%
which goes to zero as $j$ to $\infty $.

We obtained a $L^{\infty }$ solution $h^{\varepsilon }=w_{\varrho}  f^{\varrho}$ to (\ref%
{elinear}). By integrating over the trajectory, (\ref{iteration}) is valid for $%
h^{\ell}$ replaced by  $h^{\varepsilon }$ so that from Lemma \ref{iterationlinfty}
\begin{equation*}
\|h^{\varepsilon }\|_{\infty }\leq \frac{1}{8}\{\|h^{\varepsilon }\|_{\infty
}\}+e^{-\frac{\nu _{0}}{2}t}\|h^{\varepsilon }\|_{\infty }+k\left[
|wr|_{\infty }+\left\Vert \frac{wg}{\langle v\rangle }\right\Vert _{\infty }%
\right] +C(k)\|f^{\varepsilon }\|_{2} ,
\end{equation*}%
and hence
\begin{equation*}
\|h^{\varepsilon }\|_{\infty }\lesssim k\left[ |wr|_{\infty }+\left\Vert \frac{wg%
}{\langle v\rangle }\right\Vert _{\infty }\right] +C(k)\|f^{\varepsilon
}\|_{2} ,
\end{equation*}%
which implies that, from Proposition \ref{linearl2}, that $\|h^{\varepsilon
}\|_{\infty }$ is uniformly bounded and we obtain $h=wf$ solution to the linear
equation.

Now we have
\begin{eqnarray*}
\varepsilon h^{\varepsilon }+v\cdot \nabla _{x}[h^{\varepsilon }-h]+\nu
L_m[h^{\varepsilon }-h] &=&0, \\
h_{ {-}}^{\varepsilon }-h_{ {-}} &=&\frac{1}{\tilde{w}_{\varrho}(v)}%
\int_{n(x)\cdot v^{\prime }>0}[h^{\varepsilon }-h](t,x,v^{\prime })\tilde{w}_{\varrho}(v^{\prime })d\sigma ,
\end{eqnarray*}%
so that from Lemma \ref{iterationlinfty}%
\begin{equation*}
\|h^{\varepsilon }-h\|_{\infty }\lesssim _{k}\varepsilon \|h^{\varepsilon
}\|_{\infty }+\|f^{\varepsilon }-f\|_{2},
\end{equation*}%
which goes to zero. For hard potential kernel, we recall that in previous section we have constructed an approximating sequence $f^m$ to the equation
\begin{equation*}
v\cdot \nabla f^{m}+L_{m}f^{m}=g,\text{ \ \ \ \ }f^{m}_{ -}=P_{\gamma }f^{m} +r,
\end{equation*}%
with uniform bound in $L^{2}$ and a limit $f^{m}\rightarrow f$ weakly in $%
\|\cdot \|_{\nu }$, see Step 2 in the proof of Proposition \ref{linearl2}. Moreover, we obtain from %
(\ref{iterationlinfty}) that $\|w_{\varrho}f^{m}\|_{\infty }$ is uniformly bounded and
so is $w_{\varrho}f$.  Note that
\begin{equation*}
v\cdot \nabla_x \lbrack f^{m}-f]+L_{m}[f^{m}-f]=[-L+L_{m}]f,\text{ \ \ \ \ }%
[f^{m}-f]_{ -}=P_{\gamma }[f^{m}-f].
\end{equation*}%
It follows from Proposition \ref{linearl2} and the boundedness of $w_{\varrho}f$ that
\begin{equation*}
\|f^{m}-f\|_{\nu }\leq \|[-L+L_{m}]f\|_{2}\rightarrow 0.
\end{equation*}%
Now, to show $\| f^{m}-f \|_{\infty }\to 0$, we apply (\ref{iteration}) with $g=(-L+L_m)f$
\begin{equation*}
\| \{f^{m}-f\}\|_{\infty }\lesssim \|f^{m}-f\|_{2}+\left\|\frac{%
 \{L_{m}-L\}f}{\langle v\rangle }\right\|_{\infty }\rightarrow 0.
\end{equation*}%
Since $w_{\varrho}f\in L^{\infty }$, the second term goes to zero.

In the iteration scheme $f^\ell$ is continuous away from $\mathfrak{D}$ and hence $f$ is continuous away from $\mathfrak{D}$.
\end{proof}
\begin{remark}
Our construction fails to imply the $F_{s}=\mu +\sqrt{\mu }f_{s}\geq 0$.
This can only be shown by the dynamical asymptotic stability discussed in Section 7.
\end{remark}

\bigskip
\section{Well-posedness, Continuity and Fourier Law}

\begin{proof}[Proof of the Theorem \ref{main1}]

\noindent{\it Wellposedness}.  We consider the following iterative sequence%
\begin{eqnarray}
 v\cdot \nabla _{x}f^{\ell+1}+Lf^{\ell+1}  &=&\Gamma (f^{\ell},f^{\ell}), \label{nlinearfn}\\
f_{ {-}}^{\ell+1} &=&P_{\gamma }f^{\ell+1}+\frac{\mu _{\delta
}-\mu }{\sqrt{\mu }}\int_{n(x)\cdot v >0}f ^{\ell}\sqrt{\mu }(n(x)\cdot
v)dv+\frac{\mu _{\delta }-\mu }{\sqrt{\mu }},\notag
\end{eqnarray}% 
{with $f^0=0$.}

Note $\int \Gamma (f^{\ell},f^{\ell})\sqrt{\mu }=0$ and from $\int_{n(x)\cdot v<0}\mu_{\delta}(n(x)\cdot v)dv=\int_{n(x)\cdot v<0}\mu (n(x)\cdot v)dv=1,$
\begin{equation*}
\int_{\gamma _{-}}\sqrt{\mu }\left\{ \frac{\mu _{\delta }-\mu }{\sqrt{\mu }}%
\int_{n(x)\cdot v>0}f ^{\ell}\sqrt{\mu }(n(x)\cdot v)dv+\frac{\mu
_{\delta }-\mu }{\sqrt{\mu }}\right\} d\gamma =0.
\end{equation*}%
Since $|w_{\varrho}\frac{\mu_{\delta}-\mu}{\sqrt{\mu}}|_\infty \lesssim \delta$, we apply Proposition \ref{linfty} to get
\begin{equation*}
\|w_{\varrho}f^{\ell+1}\|_{\infty }+|w_{\varrho}f ^{\ell+1}|_{\infty }\lesssim \left\Vert
\frac{w_{\varrho}\Gamma (f^{\ell},f^{\ell})}{\langle v\rangle }\right\Vert _{\infty }+\delta
|w_{\varrho}f^{\ell}|_{  \infty,{+} }+\delta .
\end{equation*}%
Since $\left\Vert \frac{w_{\varrho}\Gamma (f^{\ell},f^{\ell})}{\langle v\rangle }\right\Vert
_{\infty }\lesssim \|w_{\varrho}f^{\ell}\|_{\infty }^{2}$, we deduce
\begin{equation*}
\|w_{\varrho}f^{\ell+1}\|_{\infty }+|w_{\varrho}f ^{\ell+1}|_{\infty }\lesssim
\|w_{\varrho}f^{\ell}\|_{\infty }^{2}+\delta |w_{\varrho}f ^{\ell}|_{\infty,+ }+\delta,
\end{equation*}%
so that for $\delta $ small,
\begin{equation*}
\|w_{\varrho}f^{\ell+1}\|_{\infty }+|w_{\varrho}f ^{\ell+1}|_{\infty }\lesssim \delta.
\end{equation*}%
Upon taking differences, we have
\begin{eqnarray*}
&&\lbrack f^{\ell+1}-f^{\ell}]+v\cdot \nabla _{x}[f^{\ell+1}-f^{\ell}]+L[f^{\ell+1}-f^{\ell}]
\\
&& \ \ \ \ \ \ \ \ \ \ \ \ \ \ \ \  \ \ \ \ \ \ \ \ \ \ \ \ \ \ \ \ \ \ \ \ \ \ \ \ \ \ \ \ \ \ \ \ =\Gamma (f^{\ell}-f^{\ell-1},f^{\ell})+\Gamma (f^{\ell-1},f^{\ell}-f^{\ell-1}), \\
&&f_{ {-}}^{\ell+1}-f_{ {-}}^{\ell} = P_{\gamma }[f^{\ell+1}-f^{\ell}]+\frac{\mu _{\delta }-\mu }{\sqrt{\mu }}%
\int_{n(x)\cdot v>0}[f ^{\ell}-f ^{\ell-1}](n(x)\cdot v) dv+ \sqrt{\mu },
\end{eqnarray*}%
and by Proposition \ref{linfty} again for $f^{\ell+1}-f^\ell,$
\begin{equation*}
\|w_{\varrho}[f^{\ell+1}-f^{\ell}]\|_{\infty }+|w_{\varrho}[f ^{\ell+1}-f ^{\ell}]|_{\infty }\lesssim \delta \big\{ \|w_{\varrho}[f^{\ell}-f^{\ell-1}]\|_{\infty
}+|w_{\varrho}[f ^{n}-f ^{n-1}]|_{\infty }\big\}.
\end{equation*}%
Hence $f^{\ell}$ is Cauchy in $L^{\infty }$ and  we construct our solution by taking the limit $%
f^{\ell}\rightarrow f_{s}$.  Uniqueness follows in the standard way. Moreover due to Theorem 2 and Theorem 3 in \cite{Kim} $f^{\ell}$ is continuous away from $\mathfrak{D}$ and so is $f_{s}$.  Moreover, if $\Omega $
is convex, then $\mathfrak{D}=\gamma_0$.

\vskip .2cm
\noindent{\it Formation of singularities}.
Now we prove the formation of singularity. Recall that $|\vartheta|_{\infty}\leq 1$ and the wall
Maxwellian is defined as
\begin{equation*}
\mu_{\delta}(v)=\frac{1}{2\pi [1+\delta \vartheta(x)]^2}\exp\left[{-\frac{|v|^{2}}{2[1+\delta \vartheta(x)] }}\right]\ ,
\end{equation*}%
while  the global Maxwellian is $\mu (v)=\frac{1}{2\pi }\ e^{-\frac{|v|^{2}}{2}}$.  Then
\begin{eqnarray}
\Big| \ \mu_{\delta}(x,v)-\mu(v)-2[\frac{|v|^2}{4}-1]\mu(v) \delta \vartheta(x) \ \Big| \notag \\ \lesssim \ \delta^2 |\vartheta|_{\infty}^2[1+|v|^4] \exp \Big[-\frac{|v|^2}{2(1+ \delta|\vartheta|_{\infty})}\Big].\label{mudeltaEXP}
\end{eqnarray}

For any non-convex domain $\Omega\subset \mathbf{R}^d$ for $d=2,3$, there exists at least one $(x_0,v)\in\gamma_0^{\mathbf{S}}$ with $\bar v\neq 0$.  Suppose $h=w_{\varrho}f$ satisfies (\ref{nlinearfn}) with $f^{\ell+1}=f^\ell=f$.  We prove that $h$ is discontinuous at $(x_0,v)$ by contradiction argument. Assume $h$ is continuous at $(x_0,v)$ where the magnitude of $v$ will be chosen later. By definition
\[
x_1 \equiv x_0 -t_{\mathbf{b}}(x_0,\frac{v}{|v|})\frac{\bar{v}}{|\bar{v}|}.
\]

We choose $\vartheta $ to be a continuous function such
that
\begin{equation}
\vartheta (x_{0})=|\vartheta |_{\infty }=1,\ \text{\ \ \ }\vartheta (x_{1})=%
\frac{1}{2}|\vartheta |_{\infty }=\frac{1}{2}  \label{theta},
\end{equation}%
and mainly zero elsewhere. Then by the continuity assumption we know that
the quantities $\mathbf{I}$ and $\mathbf{II}$ defined below must coincide: $%
\mathbf{I}\ =\ \mathbf{II}$. $\mathbf{I}$ is obtained by evaluating $%
h(x_{0},v)$ through the boundary condition
\begin{eqnarray}
\mathbf{I}=h(x_{0},v) &=&\frac{1}{\tilde{w}_{s}(v)}\int_{n(x_{0})\cdot
v^{\prime }>0}h(x_{0},v^{\prime })\tilde{w}_{s}(v^{\prime })d\sigma
\label{rk_1} \\
&+&w_{s}(v)\frac{\mu _{\delta }-\mu }{\sqrt{\mu }}|_{(x_{0},v)}%
\int_{n(x_{0})\cdot v^{\prime }>0}h(x_{0},v^{\prime })\tilde{w}(v^{\prime
})d\sigma   \label{sk_1} \\
&+&w_{s}(v)\frac{\mu _{\delta }-\mu }{\sqrt{\mu }}|_{(x_{0},v)}.
\label{ek_1}
\end{eqnarray}%
On the other hand, the existence of $x_{1}$ allow us to evaluate $h(x_{0},v)=%
\mathbf{II}$ along the trajectory and has the expression (\ref{iteration}%
) with $H$ as in (\ref{h1}) -- (\ref{h5}) with
\begin{equation*}
r=\frac{\mu _{\delta }-\mu }{\sqrt{\mu }},\text{ \ }h^{\ell}\equiv h,\text{ }%
g=\Gamma (\frac{h}{w},\frac{h}{w}),
\end{equation*}%
which are independent of time. Since from (\ref{mudeltaEXP}),
\begin{eqnarray*}
|r|_{2} &\lesssim &\delta |\vartheta |_{2}, \\
\left\vert \sup_{v}|r(\cdot ,v)|\right\vert _{1} &\lesssim &\delta
|\vartheta |_{1}\lesssim _{\Omega }\delta |\vartheta |_{2}, \\
|\langle v\rangle ^{\beta +4}r|_{\infty } &\lesssim &\delta |\vartheta
|_{\infty }\lesssim \delta , \\
\left\Vert \frac{w_{\varrho }\ g}{\langle v\rangle }\right\Vert _{\infty }
&\lesssim &\left\Vert \frac{w_{\varrho }\ \Gamma (\frac{h}{w},\frac{h}{w})}{%
\langle v\rangle }\right\Vert _{\infty }\lesssim \left\Vert w_{\rho
}h\right\Vert _{\infty }^{2}\lesssim \delta ^{2}|\vartheta |_{\infty
}^{2}\lesssim \delta ^{2}, \\
\Vert f\Vert _{2} &\lesssim &\delta ^{2}|\vartheta |_{\infty }^{2}+\delta
|\vartheta |_{2},
\end{eqnarray*}%
where we have employed Propositions \ref{linearl2} and \ref{linfty} for the
last two estimates. Due to (\ref{mudeltaEXP}) we can use the estimate (\ref%
{rho12}) for $\mathbf{II}$ :
\begin{eqnarray}
&&\ \ \ |h(x_{0},v)|\   \notag \\
&\leq &e^{-\nu (v)(t-t_{1})}w_{\varrho }(v)|r(x_{1},v)|  \notag \\
&&+\left\{ \frac{C_{\beta ,\rho ,\zeta }}{N}+\varepsilon C_{\rho ,\beta
,N}\right\} |\langle v\rangle ^{\beta +4}r|_{\infty }+C_{N,\varepsilon ,\rho
}\left\vert \sup_{v}|r(\cdot ,v)\right\vert _{1}  \notag \\
&&+\ \left\{ \frac{1}{N}+e^{-\frac{\nu _{0}}{2}\rho }\right\} \Vert h\Vert
_{\infty }+\varrho \left\{ \left\Vert \frac{w_{\varrho }\ g}{\langle
v\rangle }\right\Vert _{\infty }\right\} +\ C\int_{0}^{t}\left\Vert
f\right\Vert _{2}ds\   \notag \\
&\leq &e^{-\nu (v)(t-t_{1})}w_{\varrho }(v)|r(x_{1},v)|  \label{hb} \\
&&+\left\{ \frac{C_{\beta ,\rho ,\zeta }}{N}+\varepsilon C_{\rho ,\beta ,N}+%
\frac{C}{N}+Ce^{-\frac{\nu _{0}}{2}\rho }+C\delta \right\} \delta +C_{\rho
,\varepsilon ,N}\delta |\vartheta |_{2}  \notag \\
&\equiv &e^{-\nu (v)(t-t_{1})}w_{\varrho }(v)|r(x_{1},v)|+B.  \notag
\end{eqnarray}

Now we will show that, for a suitable choice of $\vartheta$, $\mathbf{I}-%
\mathbf{II}\gneqq0$, which is a contradiction so that $h=w_{\varrho}f$ has a
discontinuity at $(x_0,v)$.

Rewrite
\begin{eqnarray}
\mathbf{I}-\mathbf{II} &\geq &\{(\ref{ek_1})-e^{-\nu (v)(t-t_{1})}w_{\varrho
}(v)\frac{\mu _{\delta }(x_{1},v)-\mu (v)}{\sqrt{\mu }}\}  \notag \\
&&-\{|(\ref{rk_1})|+|(\ref{sk_1})|+B\}.  \label{I-II}
\end{eqnarray}
For large $v$, by (\ref{mudeltaEXP}), $t\geq t_{1},e^{-\nu (v)(t-t_{1})}\leq
1,$ we have
\begin{eqnarray}
&&(\ref{ek_1})-e^{-\nu (v)(t-t_{1})}w_{\varrho }(v)\frac{\mu _{\delta
}(x_{1},v)-\mu (v)}{\sqrt{\mu }}  \notag \\
&\geq &2w_{\varrho }(v)[\frac{|v|^{2}}{4}-1]\sqrt{\mu }\delta |\vartheta
|_{\infty }-w_{\varrho }(v)[\frac{|v|^{2}}{4}-1]\sqrt{\mu }\delta |\vartheta
|_{\infty }-C\delta ^{2}  \notag \\
&\geq &w_{\varrho }(v)[\frac{|v|^{2}}{4}-1]\sqrt{\mu }\delta |\vartheta
|_{\infty }-C\delta ^{2}.  \label{2>1}
\end{eqnarray}%
Recall that $\Vert h\Vert _{\infty }\lesssim \delta |\vartheta |_{\infty
}\lesssim \delta $ to conclude that
\begin{eqnarray*}
&&|(\ref{rk_1})|,|(\ref{sk_1})|\  \\
&\leq &\ \Big\{\varrho ^{-4}w_{\varrho }(v)\sqrt{\mu (v)}+\Big(1+[\frac{%
|v|^{2}}{4}-1]w_{\varrho }(v)\sqrt{\mu (v)}\Big)\delta |\vartheta |_{\infty }%
\Big\}\delta |\vartheta |_{\infty }.
\end{eqnarray*}%
By (\ref{2>1}), we can find $|v_{0}|>0$ so that $c_{v_{0}}>0$ sufficiently
large so that
\begin{equation*}
w_{\varrho }(v_{0})[\frac{|v_{0}|^{2}}{4}-1]\sqrt{\mu (v_{0})}-\varrho
^{-4}w_{\varrho }(v_{0})\sqrt{\mu (v_{0})}\geq c_{v_{0}}>0.
\end{equation*}%
On the other hand in (\ref{hb}) and (\ref{2>1}), we choose $\delta $
sufficiently small, and $\rho $ sufficiently large, then $N$ sufficiently
large to get, then $\varepsilon $ sufficiently small, such that for $v=v_{0},
$
\begin{equation*}
\mathbf{I}-\mathbf{II\geq }\frac{c_{v_{0}}}{2}\delta -C\delta |\vartheta
|_{2}.
\end{equation*}%
Since $\vartheta $ is almost zero except for $x_{0}$ and $x_{1},$ we can can
make $|\vartheta |_{2}$ arbitrarily small. In particular, there is a
continuous function $\vartheta $ such that $\frac{c_{v_{0}}}{2}\delta
-C\delta |\vartheta |_{2}>\frac{c_{v_{0}}}{4}\delta >0.$ Hence $\mathbf{I}-%
\mathbf{II}>0$ and this is a contradiction.
\end{proof}

\vskip .3cm

\begin{proof}[Proof of Theorem \ref{main2}]

We now prove the $\delta-$expansion. Since $\mu_{\delta}(x)$ is analytic with respect to $\delta$ we have
\begin{eqnarray}
\mu_{\delta }(x,v) &=&\mu +\delta \mu _{1}+\delta ^{2}\mu _{2}+\cdots.\label{mudelta1}
\end{eqnarray}
Note that $|\mu_i(v)| \lesssim p_i(v)e^{-\frac{|v|^2}{2}}$,  where $p_i(v)$ is some polynomial. Further we seek a formal expansion
\begin{eqnarray*}
f_{s} &\sim& \delta f_{1}+\delta ^{2}f_{2}+\cdots  .
\end{eqnarray*}%
Plugging this into the equation,
\begin{eqnarray*}
&&v\cdot \nabla _{x}[\delta f_{1}+\delta ^{2}f_{2}+  \dots  ]+L[\delta f_{1}+\delta ^{2}f_{2}+  \dots ] \\
&& \ \ \ \ \ \ \ \ \ \ \ \ \ \ \ \ \  \  \ \  \ \ \ \ \ \ \ \ \ \ \ \ =\Gamma (\delta f_{1}+\delta ^{2}f_{2}+ \dots ,\delta
f_{1}+\delta ^{2}f_{2}+  \dots ), \\
&& \lbrack \delta f_{1}+\delta ^{2}f_{2}+  \dots ]_{{-}} =P_{\gamma }[\delta f_{1}+\delta ^{2}f_{2}+  \dots ]  \\
&& \ \ \ \ \ \ \ \ \    + \frac{[\delta \mu _{1}+\delta ^{2}\mu _{2}+ \dots]}{\sqrt{\mu }}\int_{n(x)\cdot
v>0}[ \sqrt{\mu}  +\delta f_{1}+\delta ^{2}f_{2}+ \dots ]\sqrt{%
\mu }(n(x)\cdot v)dv.
\end{eqnarray*}%
We compare the coefficients of power of $\delta$ to get an equation for $f_i$ for $i=1, \dots, m-1$, (assuming $f_0\equiv 0$)
\begin{eqnarray}
v\cdot\nabla_x f_i + L f_i &=& \sum_{j=1}^{i-1} \Gamma(f_j,f_{i-j}),\label{eqnfi}\\
{f_i}|_{\gamma_-} &=& P_{\gamma}f_i + \sum_{j=1}^{i-1}\frac{\mu_j}{\sqrt{\mu}} \int_{n(x)\cdot v>0} f_{i-j} \sqrt{\mu} ( n(x)\cdot v)dv + \frac{\mu_i}{\sqrt{\mu}}.\notag
\end{eqnarray}
Note that, from the $\delta-$expansion,
\[
\int_{\gamma_{\pm}} (\mu^{\delta}(x,v)-\mu(v)) |n(x)\cdot v| dv =0, \ \ \ \ \sum_{i=1}^{\infty} \delta^i \int_{\gamma_{\pm}} \mu_i |n(x)\cdot v| dv =0.
\]
Since
\begin{eqnarray*}
\int_{\mathbf{R}^3} \sum_{j=1}^{i-1} \Gamma(f_j,f_{i-j}) \sqrt{\mu} dv=0, \ \ \ \int_{n(x)\cdot v < 0} \Big[\sum_{j=1}^{i-1}\frac{\mu_j}{\sqrt{\mu}} \int_{n(x)\cdot v>0} f_{i-j} \sqrt{\mu}   + \frac{\mu_i}{\sqrt{\mu}}\Big] =0,
\end{eqnarray*}
by applying Proposition \ref{linfty} repeatedly, we can construct $%
f_{1},f_{2},...,f_{m-1}$ inductively so that for $0 \leqq \zeta< \frac{1}{4}$,
\begin{equation}
\|w_{\varrho}f_{i}\|_{\infty }+|w_{\varrho}f_{i}|_{\infty }\lesssim _{m}1.\label{estfi}
\end{equation}%
In particular $f_{1}$ satisfies (\ref{f1}). Now we define the remainder $f_m^{\delta}$ satisfying
\[
f_s = \delta f_1 + \cdots + \delta^{m-1}f_{m-1} + \delta^m f_m^{\delta},
\]
and we obtain the equation for the remainder $f_m^{\delta}$ \begin{eqnarray*}
v\cdot \nabla _{x}f_{m}^{\delta }+Lf_{m}^{\delta } &=&\delta \lbrack \Gamma
(f_{1}+\delta f_{2}+ \dots+\delta ^{m-2}f_{m-1},f_{m}^{\delta }) \\
&+&\Gamma
(f_{m}^{\delta },f_{1}+\delta f_{2}+ \dots +\delta ^{m-2}f_{m-1})+\delta ^{m}\Gamma (f_{m}^{\delta },f_{m}^{\delta })+g^{\delta }, \\
f_{m}^{\delta }|_{\gamma _{-}} &=&P_{\gamma }f_{m}^{\delta }  +\frac{\mu
^{\delta }-\mu }{\sqrt{\mu }}\int_{n(x)\cdot v>0}f_{m}^{\delta }\sqrt{\mu }%
(n(x)\cdot v)dv+r^{\delta }.
\end{eqnarray*}%
where
\begin{eqnarray*}
g^{\delta}= \sum_{i=m}^{2m-1}\sum_{j=1}^{m-1} \delta^{i-m}\Gamma(f_j,f_{i-j}), \ \  r^{\delta}=\sum_{i=1}^{m}\delta^{i-1}\frac{\mu_i}{\sqrt{\mu}}\int_{n(x)\cdot v>0} f_{m-i}\sqrt{\mu}(n(x)\cdot v)dv.
\end{eqnarray*}
 Since by (\ref{estfi})$$\left\|\frac{w_{\varrho} \ g^{\delta }}{\langle v\rangle }\right\|_{\infty }+|w_{\varrho} \ r^{\delta
}|_{\infty }\lesssim _{m}1,$$ we can apply Proposition \ref{linfty} to
deduce that%
\begin{eqnarray*}
&&\|w_{\varrho} \ f_{m}^{\delta }\|_{\infty }+|w_{\varrho} \ f_{m}^{\delta }|_{\infty }\\
&& \ \ \ \ \lesssim
_{m}\delta \|w_{\varrho} \ f_{m}^{\delta }\|_{\infty }+\delta ^{m}\|w_{\varrho} \ f_{m}^{\delta
}\|_{\infty }^{2}+\delta |w_{\varrho} \ f_{m}^{\delta }|_{\infty }+\|w_{\varrho} \ g^{\delta
}\|_{\infty }+|w_{\varrho} \ r^{\delta }|_{\infty }.
\end{eqnarray*}%
We therefore conclude that $\|w_{\varrho} \ f_{m}^{\delta }\|_{\infty }+|w_{\varrho} \ f_{m}^{\delta
}|_{\infty }\lesssim 1$ and the expansion is valid. \end{proof}

\vskip.3cm
\begin{proof}[Proof of Theorem \ref{fourierlaw}]
We now consider a slab $-\frac{1}{2}\leq x\leq +\frac{1}{2}$.
We consider the stationary solution $f_{s}\in L^{\infty }$ satisfying
\begin{eqnarray}
v_{1}\partial _{x}f_{s}+Lf_{s}&=&\Gamma (f_{s},f_{s}), \ \ \ \ \  \ \ (x,v)\in (-\frac{1}{2},+\frac{1}{2})\times\mathbf{R}^3, \label{fs}\\
f_{s}(x,v)&=&P_{\gamma }f_{s}+\frac{\mu ^{\delta }-\mu }{\sqrt{\mu }}%
\int_{n(x)\cdot v> 0}f_{s}\sqrt{\mu }(n(x)\cdot v)dv,\notag
\end{eqnarray}%
for $x=-\frac{1}{2}$ or $x=+\frac{1}{2}$.  From Theorem \ref{main1}, $%
\ \|w_{\varrho} \ f_{s}\|_{\infty }\lesssim \delta $. We claim that $%
\partial_x f_{s}(x,v)\in L^{2}([-\frac{1}{2}+\varepsilon ,+\frac{1}{2}-\varepsilon ]\times \mathbf{R}^{3})$
for any small $\varepsilon >0$.  In fact, we multiply (\ref{fs}) by a 1-dimensional spatial
smooth cutoff function $\chi (x)
$, $(\chi\equiv 0 \ \text{near} \ x=-\frac{1}{2} \ \text{and} \ x=+\frac{1}{2})$ so that%
\begin{equation*}
v_{1}\partial _{x}[\chi f_{s}]+L[\chi f_{s}]=\Gamma (\chi
f_{s},f_{s})-v_{1}f_{s}\chi ^{\prime }.
\end{equation*}%
We take one spatial derivative (\ref{fs}) to get
\begin{eqnarray*}
v_{1}\partial _{xx}[\chi f_{s}]&+&L\partial _{x}[\chi f_{s}] =\Gamma
(\partial _{x}[\chi f_{s}],f_{s})+\Gamma (\chi f_{s},\partial
_{x}f_{s})-v_{1}\chi ^{\prime }\partial _{x}f_{s}-v_{1}\chi ^{\prime \prime
}f_{s} \\
&=&\Gamma (\partial _{x}[\chi f_{s}],f_{s})+\Gamma (\chi f_{s},\partial
_{x}f_{s})+\chi ^{\prime }(Lf_{s}-\Gamma (f_{s},f_{s}))-v_{1}\chi ^{\prime
\prime }f_{s} \\
&=&\Gamma (\partial _{x}[\chi f_{s}],f_{s})+\Gamma (f_{s},\partial _{x}[\chi
f_{s}]) \\
&&-\Gamma (f_{s},\chi ^{\prime }f_{s})+\chi ^{\prime }(Lf_{s}-\Gamma
(f_{s},f_{s}))-v_{1}\chi ^{\prime \prime }f_{s},
\end{eqnarray*}%
where we have used (\ref{fs}) to replace $-v_1 \partial_x f_s = Lf_s -\Gamma(f_s,f_s)$.  Letting $Z=\chi \partial_x f_s $, then we have
\begin{eqnarray*}
v_{1}Z_{x}+LZ &=&\Gamma (Z,f_{s})+\Gamma (f_{s},Z) \\
&&-\Gamma (f_{s},\chi ^{\prime }f_{s})+\chi ^{\prime }(Lf_{s}-\Gamma
(f_{s},f_{s}))-v_{1}\chi ^{\prime \prime }f_{s}.
\end{eqnarray*}%
Note that $-\Gamma (f_{s},\chi ^{\prime }f_{s})+\chi ^{\prime }(Lf_{s}-\Gamma
(f_{s},f_{s}))-v_{1}\chi ^{\prime \prime }f_{s} \in L^{\infty}$.  Multiply by $Z$ and use Green's identity to have (no boundary contribution)
\begin{equation*}
\|(\mathbf{I}-\mathbf{P})Z\|_{\nu }^{2} \ \lesssim \ \delta \big\{ \ \|Z\|_{\nu }^{2} \ + \ 1 \ \big\}.
\end{equation*}%
Then we repeat the proof of Lemma \ref{steadyabc} but replacing $\phi_a$ in (\ref{phia}) with the solution of $-\Delta \phi _{a}=a$ with $\phi _{a}=0$ on $%
\partial \Omega$.  Since the boundary condition of $Z$ is $Z_{\gamma}\equiv 0$, Lemma \ref{steadyabc} is valid so that
\[
\|\chi \partial_x f_s \|_2^2 \lesssim \|Z\|_{2}^2 < +\infty,
\]
and $\partial_x f_s$ is locally in $L_{loc}^2(0,1)$.

Recall that
\[
u_s(x) = \frac{1}{\int_{\mathbf{R}^3}F_s(x,v)dv } \int_{\mathbf{R}^3} v F_s(x,v)dv.
\]
Now we are going to show
\begin{equation*}
u_s(x)=0, \ \ \ \ \text{for any\ }x\in [-\frac{1}{2},+\frac{1}{2}].
\end{equation*}
For any smooth test function $\psi\in C^{\infty}([-\frac{1}{2},+\frac{1}{2}])$, due to the weak formulation for $F_s$, we have for the first component of $\rho_s u$ that
\begin{eqnarray*}
&& \int_{-\frac{1}{2}}^{+\frac{1}{2}} \psi^{\prime}(x) (\rho_s u_s)_1 dx\\
&& \  = \int_{-\frac{1}{2}}^{+\frac{1}{2}}  \psi^{\prime}(x) \int_{\mathbf{R}^3} v_1 F_s(x,v) dvdx\\
&& \   = \psi(\frac{1}{2})\int_{\mathbf{R}^3} v_1 F_s(\frac 12)  - \psi(-\frac{1}{2})\int_{\mathbf{R}^3} v_1 F_s(-\frac 1 2) - \int_{-\frac{1}{2}}^{+\frac{1}{2}} \psi(x)\int_{\mathbf{R}^3} v_1 \partial_x F_s(x,v)\\
&& \   = -\int_{-\frac{1}{2}}^{+\frac{1}{2}} \psi(x) \int_{\mathbf{R}^3} Q(F_s,F_s)(x,v) dv dx\\
 && \   =0,
\end{eqnarray*}
where we used the boundary condition as well as the orthogonality of $Q$ to the collision invariants. Hence $\partial_x [(\rho_s u_s)_1(x)]=0$ in the distribution sense and hence $\rho_s u_s \in W^{1,\infty}$.  Therefore $\rho_s u_s $ is continuous up to the boundary and $\lim_{x\downarrow -\frac{1}{2}}  \rho_s u_s (x)= \rho_s u_s (-\frac{1}{2})$ and $\lim_{x\uparrow +\frac{1}{2}}  \rho_s u_s (x)= \rho_s u_s (+\frac{1}{2})$.  Moreover for all $\psi \in C^{\infty}([-\frac{1}{2},+\frac{1}{2}])$
\begin{eqnarray*}
0=\int_{-\frac{1}{2}}^{+\frac{1}{2}} \psi^{\prime}(x) (\rho_s u_s)_1
= (\rho_s u_s)_1(\frac{1}{2}) \psi(\frac{1}{2})- (\rho_s u_s)_1(-\frac{1}{2})\psi(-\frac{1}{2}),
\end{eqnarray*}
to conclude
\begin{equation*}
(\rho_s u_s )_1(x) \equiv 0, \ \text{ for all } \ x\in [-\frac{1}{2},+\frac{1}{2}].
\end{equation*}
 On the other hand $$\rho_s(x)=\int_{\mathbf{R}^3} \mu(v)+\sqrt{\mu}f_s(x,v) dv \geq \int_{\mathbf{R}^3} \mu(v) dv - \int_{\mathbf{R}^3} \sqrt{\mu}(v)w^{-1}(v) dv \times \| wf_s\|_{\infty} >0, $$
for small $\| wf_s\|_{\infty}$ so that $\rho_s$ is positive. Thus $(u_{s })_1(x)\equiv0$.  The components $(u_{s})_i$ for $i=2,3$ vanish by symmetry:
\begin{eqnarray}
(u_s)_i(x) &=& \frac{1}{\rho_s(x)} \int_{\mathbf{R}^3} v_i F_s(x,v)dv\notag\\
&=&\frac{1}{\rho_s(x)} \int_{v_i>0} v_i F_s(x,v)dv - \frac{1}{\rho_s(x)} \int_{\tilde{v}_i>0}  \tilde{v}_i F_s(x,{v})d\tilde{v}\notag\\
&=&\frac{1}{\rho_s(x)} \int_{v_i>0} v_i F_s(x,v)dv - \frac{1}{\rho_s(x)} \int_{ \tilde{v}_i>0} \tilde{v}_i \tilde{F}_s(x, \tilde{v})d \tilde{v},\label{i}
\end{eqnarray}
where $(\tilde{v})_{i}=-v_i$ and $\tilde{v}_j = v_j$ for $j\neq i$, and we have defined $\tilde{F}_s(x,v)=F_s(x,\tilde{v})$.  Then $\tilde{F}$ solves
\begin{eqnarray*}
v_1 \partial_{x_1} \tilde{F}_s = Q(\tilde{F}_s,\tilde{F}_s) \ , \ \ \ \tilde{F}_s(x,v) = \mu^{\theta}(x,v) \int_{n(x)\cdot v >0} \tilde{F}_s(x,v) (n(x)\cdot v)dv,
\end{eqnarray*}
where $x\in\partial\Omega$ and $n(x)\cdot v<0$.
By the uniqueness, we conclude $\tilde{F}_s(x,v)=F_s(x,v)$ and hence $\tilde{F}_s(x,\tilde{v})=F_s(x,\tilde{v})$ almost everywhere. Therefore it follows that (\ref{i}) vanishes and $(u_s)_i \equiv 0$ almost everywhere for $i=2,3$.

Recall the definition of temperature $\theta_s$ in (\ref{thetas}). Since $\partial _{x}f_{s}\in L^{2}([-\frac{1}{2}+\varepsilon
,+\frac{1}{2}-\varepsilon ]\times \mathbf{R}^{3})$, we have $\partial _{x}\theta_s \in
L^{2}([-\frac{1}{2}+\varepsilon ,+\frac{1}{2}-\varepsilon ]\times \mathbf{R}^{3})$ so that in the
sense of distribution,
\begin{equation*}
\kappa (\theta_s ) \partial_{x}\theta_s=\partial_x \{K(\theta_s ) \},
\end{equation*}%
with $K^{\prime }=\kappa >0$. Since $u_s \equiv 0$, by Green's identity,
\begin{equation}
\partial_x \left\{\frac 1 2\int_{\mathbf{R}^3} v|v|^2 f_s\right\} \equiv \partial_x q_s =0.\notag
\end{equation}
By the Fourier law (\ref{fourierlaw1}) we have then, in the sense of distribution
$$
\partial_{xx} \{K(\theta_s )\} =0.
$$
Therefore, in the sense of distributions
\begin{equation*}
K(\theta_s )=Ax+B.
\end{equation*}%
But if $\theta_s =\frac{1}{3\rho_s}\int |v|^{2}f_{s}\sqrt{\mu }%
\in L^{\infty }$, then $K(\theta_s )\in L^{\infty }$, and this implies that $K(\theta_s )$
is in $W^{2,\infty}$ so that $K(\theta_s )$ is continuous up to the boundary. But
if $K^{\prime }=\kappa >0$ we thus deduce that $\theta_s $ itself is
continuous up to the boundary $x=-\frac{1}{2},+\frac{1}{2}$.  We then can rewrite
\begin{equation*}
A={K(\theta ( +\frac{1}{2} ))-K(\theta (  -\frac{1}{2}))},\text{ \ \ \ }B=\frac{K(\theta(+\frac{1}{2}))+K(\theta(-\frac{1}{2}))}{2}.
\end{equation*}%
By the main theorem, $f_{s}= \delta f_{1}+O(\delta ^{2})$, so that
\begin{equation*}
\theta_s (x)=\theta _{0}+\delta \theta _{1}(x)+O(\delta ^{2}),
\end{equation*}%
and we know that $\theta _{0}=\frac 1 3\int |v|^{2}\mu $ is constant in our
construction. Then we have for all $ -\frac{1}{2}\leq x\leq + \frac{1}{2},$
\begin{equation}
K(\theta _{0}+\delta \theta _{1}+O(\delta ^{2}))=K(\theta _{0})+\delta
K^{\prime }(\theta _{0})\theta _{1}+O(\delta ^{2}).
\end{equation}%
We therefore have
\begin{eqnarray*}
A&=&\delta K^{\prime}(\theta_0)[{\theta_1(+\frac{1}{2})-\theta_1(-\frac{1}{2})}]+ O(\delta^2),\\
B&=&K(\theta_0) + \delta K^{\prime}(\theta_0) \frac{\theta_1(+\frac{1}{2})+\theta_1(-\frac{1}{2}) }{2} +O(\delta^2).
\end{eqnarray*}%
Hence the first order expansion $\theta _{1}(x)$ must to be a linear
function:
\begin{equation*}
\theta _{1}(x)=\frac{\theta_1(+\frac{1}{2})+\theta_1(-\frac{1}{2})}{2}
+ x [{\theta_1(+\frac{1}{2})-\theta_1(-\frac{1}{2})}]
.
\end{equation*}%
Here $\theta _{1}=\frac{1}{3\int dv f_1}\int |v|^{2}f_{1}$ and from (\ref{slab}) $f_{1}$ satisfies
(\ref{slablinear}).
\end{proof}

\bigskip
\section{$L^2$ Decay}
The main purpose of this section is to prove the following:
\begin{proposition}
\label{dlinearl2}Suppose that for all $t> 0$
\begin{equation}\iint_{\Omega\times\mathbf{R}^3} g(t,x,v)\sqrt{\mu }dv=\int_{\gamma_-}r\sqrt{\mu}d\gamma=0. \label{dlinearcondition}
\end{equation} Then there exists a unique solution to the problem
\begin{equation}
\partial _{t}f+v\cdot \nabla _{x}f+Lf=g,\text{ \ \ }f(0)=f_{0},\quad \text{ in }\Omega\times\mathbf{R}^3\times \mathbf{R}_+, \label{dlinear}
\end{equation}%
with
\begin{equation}
f_{-}=P_{\gamma }f+r,\quad \text{ on } \gamma_-\times \mathbf{R}_+,\label{dlinearl2BC}
\end{equation}
such that for all $t\> 0$,
\begin{equation}
  \iint_{\Omega\times\mathbf{R}^3} f(t,x,v)\sqrt{\mu }dxdv= 0.\label{dlinearcondition1}
\end{equation}
Moreover, there is $\lambda>0$ such that
\begin{equation*}
\|f(t)\|_{2}^{2}+|f (t)|_{2}^{2}\lesssim e^{-\lambda
t}\Big\{\|f_{0}\|_{2}^2+\int_{0}^{t}e^{\lambda
s}\|g(s)\|^{2}_2ds+\int_{0}^{t} e^{\lambda s}|r(s)|_2^{2}ds\Big\}.
\end{equation*}
\end{proposition}

\begin{lemma}
\label{dabc}Assume that $g$ and $r$ satisfy (\ref{dlinearcondition}) and $f$ satisfies (\ref{dlinear}) , (\ref{dlinearl2BC}) and (\ref{dlinearcondition1}). Then there exists a function $G(t)$ such that, for all $0\le s\leq t$,
\begin{eqnarray*}
\int_{s}^{t}\|\mathbf{P}f(\tau)\|_{\nu }^{2} &\lesssim& G(t)-G(s)+\int_{s}^{t}\|g(\tau)\|_{2}^{2} +|r(\tau)|_{2}^{2}\\
&+&\int_{s}^{t}\|(%
\mathbf{I}-\mathbf{P})f(\tau)\|_{\nu
}^{2} +\int_{s}^{t}|(1-P_{\gamma
})f(\tau)|_{2,+}^{2}
.
\end{eqnarray*}%
Moreover, $G(t)\lesssim \|f(t)\|_{2}^{2}$.
\end{lemma}

\begin{proof}
The key of the proof is to use the same choices of test functions (allowing extra dependence on time) of (%
\ref{phic}), (\ref{phibj}), (\ref{phibij}) and (\ref{phia}) and estimate the new contribution $\int_{s}^{t}\iint_{\Omega\times\mathbf{R}^3} \partial _{t}\psi f$ in the time dependent weak formulation
\begin{eqnarray}
&& \ \int_{s}^{t}\int_{\gamma_+ }\psi f d\gamma - \int_{s}^{t}\int_{\gamma_- }\psi f d\gamma -\int_{s}^{t}\iint_{\Omega\times\mathbf{R}^3} v\cdot \nabla _{x}\psi f\notag\\
&&=-\iint_{\Omega\times\mathbf{R}^3} \psi f(t)+\iint_{\Omega\times\mathbf{R}^3} \psi f(s)
+\int_{s}^{t}\iint_{\Omega\times\mathbf{R}^3} -\psi L(\mathbf{I}-\mathbf{%
P})f\notag +  \psi g\nonumber\\
&& \ \ \ \ +\int_{s}^{t}\iint_{\Omega\times\mathbf{R}^3} \partial _{t}\psi f .  \label{timeweak}
\end{eqnarray}%
We note that, with such choices \newline
\begin{equation*}
G(t) = -\iint_{\Omega\times\mathbf{R}^3}\psi f(t), \quad |G(t)|\lesssim \|f(t)\|_{2}^{2}.
\end{equation*}%
Without loss of generality we give the proof for $s=0$.

\noindent {\bf Remark:} We note that (\ref{dlinearcondition}), (\ref{dlinear}), (\ref{dlinearl2BC}) and (\ref{dlinearcondition1}) are all invariant under a standard $t$-mollification for all $t>0$. The estimates in step 1 to step 3  are obtained via a $t$-mollification  so that all the functions are smooth in $t$. For the notational simplicity we do not write explicitly the parameter of the regularization.
\newline

\noindent {\it Step 1.} Estimate of  ${ \nabla_x  \Delta_N^{-1} \partial_{t}a   =   \nabla_x \partial_t \phi_a}$ in (\ref{timeweak})

In the weak formulation(with time integration over $[t,t+\varepsilon]$), if we choose the test function $\psi=\phi \sqrt{\mu }$ with $\phi(x)$ dependent only of $x$, then we get
(note that $L f$ and $g$ against $\phi(x) \sqrt{\mu }$
are zero)
\begin{eqnarray*}
\sqrt{2\pi}\int_{\Omega }[a(t+\varepsilon )-a(t)]\phi (x)
 = 2\pi\sqrt{2\pi}\int_{t}^{t+\varepsilon }\int_{\Omega }(b\cdot \nabla _{x})\phi (x)
 +\int_{t}^{t+\varepsilon }\int_{\gamma_-}  r \phi\sqrt{\mu},
\end{eqnarray*}%
where $\int_{\mathbf{R}^3}\mu(v)dv=\sqrt{2\pi}, \   \int_{\mathbf{R}^{3}}(v_{1})^{2}\mu (v)dv=2\pi\sqrt{2\pi}$ and we have used the splitting (\ref{bsplit}) and (\ref{insidesplit}). Taking difference
quotient, we obtain, for almost all $t$,
\begin{equation*}
\int_{\Omega }\phi \partial _{t}a=\sqrt{2\pi}\int_{\Omega }(b\cdot \nabla_x )\phi
+\frac{1}{{2\pi}}\int_{\gamma_-}  r  \phi \sqrt{\mu}.
\end{equation*}
Notice that, for $\phi =1$, from (\ref{dlinearcondition}) the right hand side of the above equation is zero so that, for all $t> 0 $,%
\begin{equation*}
\int_{\Omega }\partial _{t}a(t )dx =0.
\end{equation*}
On the other hand, for all $\phi(x)
\in H^{1}\equiv H^{1}(\Omega )$, we have, by the trace theorem $|\phi|_2 \lesssim \|\phi \|_{H^1}$,
\begin{eqnarray*}
\left|\int_{\Omega }\phi (x)\partial _{t}adx\right| &\lesssim & |r|_{2} |\phi  |_{ 2 }+\|b\|_{2 }\|\phi
\|_{H^{1} } \\
&\lesssim & \{ \ \|b(t)\|_{2}  +|r|_{2} \} \ \|\phi \|_{H^{1}}.
\end{eqnarray*}%
Therefore we conclude that, for all $t>0$,
\begin{equation*}
\|\partial _{t}a(t)\|_{(H^{1})^{\ast }}\ \lesssim \
\|b(t)\|_{2} +|r|_{2},
\end{equation*}%
where $(H^{1})^{\ast }\equiv(H^{1}(\Omega ))^{\ast }$ is the dual space of $H^{1}(\Omega )$
with respect to the dual pair $$\langle A,B \rangle=\int_{\Omega }A(x)B(x)dx,$$ for $A\in H^{1}$
and $B\in (H^1)^{\ast}$.

By the standard elliptic theory, we can solve the Poisson equation with the Neumann
boundary condition
\begin{equation}
-\Delta \Phi_a =\partial _{t}a(t)\ ,\ \ \ \ \ \frac{\partial \Phi_a }{\partial n}%
\Big{|}_{\partial \Omega }=0, \notag
\end{equation}%
with the crucial condition $\int_{\Omega }\partial _{t}a(t,x)dx=0$ for all $t>0$.   Notice that $\Phi_a =-\Delta _{N}^{-1}\partial _{t}a=\partial_t \phi_a$ where $\phi_a$ is defined in (\ref{phia}).
Moreover we have
 \begin{eqnarray*}
 \|\nabla_x \partial_t \phi_a \|_{2} &=& \|\Delta _{N}^{-1}\partial
 _{t}a(t)\|_{H^{1}}=\|\Phi_a \|_{H^1}\\
 &\lesssim &\|\partial _{t}a(t)\|_{(H^{1})^{\ast }}
  \lesssim  \|b(t)\|_{2}+ |r|_{2}.
 \end{eqnarray*}%
 Therefore we conclude, for almost all $t> 0,$
 \begin{equation}
 \|\nabla_x \partial _{t}\phi _{a}(t)\|_{2 }\lesssim
 \|b(t)\|_{2 } +|r|_{2}.  \label{-1a}
 \end{equation}%

\vskip .3cm
\noindent{\it Step 2.} {Estimate of} $\nabla_x  \Delta^{-1}\partial _{t}{b}^{j} = \nabla_x \partial_t \phi_b^i$ in (\ref{timeweak})

In (\ref{timeweak}) we choose a test function $\psi =\phi (x)v_{i}%
\sqrt{\mu }$.  Since $\mu(v)=\frac{1}{2\pi}e^{-\frac{|v|^2}{2}}$,  $\int v_{i}v_{j}\mu(v)dv=\int
v_{i}v_{j}(\frac{|v|^{2}}{2}-\frac{3}{2})\mu(v)dv=2\pi\sqrt{2\pi}\delta _{i,j}$ and we get
\begin{eqnarray*}
&&2\pi\sqrt{2\pi}\int_{\Omega }[b_i(t+\varepsilon )-b_i(t)]\phi\\
&=&-\int_{t}^{t+\varepsilon }\int_{\gamma }f\phi v_{i}\sqrt{\mu }  +2\pi\sqrt{2\pi} \int_{t}^{t+\varepsilon
}\int_{\Omega}\partial _{i}\phi [a+ c] \\
&&+\int_{t}^{t+\varepsilon }\iint_{\Omega \times \mathbf{R}%
^{3}}\sum_{j=1}^{d}v_{j}v_{i}\sqrt{\mu }\partial _{j}\phi (\mathbf{I}-\mathbf{P})f + \int_{t}^{t+\varepsilon} \iint_{\Omega\times\mathbf{R}^3} \phi v_i g \sqrt{\mu}.
\end{eqnarray*}%
Taking difference quotient, we obtain
\begin{eqnarray*}
&&\int_{\Omega }\partial _{t}b_{i}(t)\phi\notag\\
&=& \frac{-1}{2\pi\sqrt{2\pi}}\int_{\gamma}f(t)v_{i}\phi \sqrt{\mu } + \int_{\Omega }\partial _{i}\phi [a(t)+ c(t)]\\
&&+\frac{1}{2\pi\sqrt{2\pi}}\Big\{\iint_{\Omega
\times \mathbf{R}^{3}}\sum_{j=1}^{d}v_{j}v_{i}\sqrt{\mu }\partial _{j}\phi(\mathbf{I}-\mathbf{P})f(t)+\iint_{\Omega\times\mathbf{R}^3} \phi v_i g(t) \sqrt{\mu}\Big\}.
\end{eqnarray*}%
For fixed $t> 0$, we choose $\phi=\Phi_b^i $ solving
\begin{eqnarray*}
-\Delta \Phi^{i}_{b}   = \partial _{t}b_{i}(t) , \ \ \ \ \
 \Phi^{i}_{b}|_{\partial \Omega }  = 0.
\end{eqnarray*}%
Notice that $\Phi_b^i = -\Delta^{-1}\partial_t b_i = \partial_t \phi_b^i$ where $\phi_b^i$ is defined in (\ref{jb}). The boundary terms vanish because of the Dirichlet
boundary condition on $\Phi^{i}_{b}$.  Then we have, for $t\geq 0$,
\begin{eqnarray*}
\int_{\Omega }|\nabla_x \Delta ^{-1}\partial _{t}b_{i}(t)|^{2}dx
&=&\int_{\Omega }|\nabla_x \Phi^{i}_{b} |^{2}dx=-\int_{\Omega }\Delta  \Phi^{i}_{b} \Phi^{i}_{b} dx \\
&\lesssim &\varepsilon \{\|\nabla_x \Phi
^{i}_{b}\|_{2}^{2} + \|  \Phi
^{i}_{b}\|_{2}^{2} \}+\|a(t)\|_{2}^{2}+\|c(t)\|_{2}^{2}\\
&& \ \   +\|(\mathbf{I}-\mathbf{P}%
)f(t)\|_{2}^{2} +\|g(t)\|_2^2 \\
&\lesssim &\varepsilon  \|\nabla_x \Phi
^{i}_{b}\|_{2}^{2}  +\|a(t)\|_{2}^{2}+\|c(t)\|_{2}^{2}\\
&& \ \   +\|(\mathbf{I}-\mathbf{P}%
)f(t)\|_{2}^{2} +\|g(t)\|_2^2,
\end{eqnarray*}%
where we have used the Poincar$\acute{e}$ inequality.
Hence, for all $t>0$
\begin{eqnarray}
&&  \ \  \|\nabla_x \partial_t \phi_b^i (t)\|_{2}=\|\nabla_x \Delta ^{-1}\partial _{t}b_{i}(t)\|_{2}\notag\\
&&  \lesssim \ \
\|a(t)\|_{2}+\|c(t)\|_{2}+\|(\mathbf{I}-%
\mathbf{P})f(t)\|_{2}+\|g(t)\|_2.\label{-1b}
\end{eqnarray}

\vskip .3cm
\noindent{\it Step 3.} {Estimate of  }$\nabla_x \Delta^{-1} \partial_t c= \nabla_x \partial _{t}\phi _{c} $ in (\ref{timeweak})

In the weak formulation, we choose a test function $\phi (x)(\frac{|v|^{2}}{2}-%
\frac{3}{2})\sqrt{\mu }$.  Since $\int_{\mathbf{R}^3}\mu (v)(\frac{|v|^{2}}{2}-\frac{3}{2})dv=0, \ \int_{\mathbf{R}^3} \mu(v) v_i v_j (\frac{|v|^2-3}{2})=2\pi\sqrt{2\pi}\delta_{i,j}$ and $%
\int_{\mathbf{R}^3}\mu (v)(\frac{|v|^{2}}{2}-\frac{3}{2})^{2}dv=3\pi\sqrt{2\pi}$ we get
\begin{eqnarray*}
&&3\pi\sqrt{2\pi} \int_{\Omega} \phi(x)  c(t+\varepsilon,x) dx -3\pi\sqrt{2\pi} \int_{\Omega} \phi(x)  c(t,x) dx  \\
 & &=  2\pi\sqrt{2\pi}\int_{t}^{t+\varepsilon}\int_{\Omega }b\cdot \nabla_x \phi
-\int_{t}^{t+\varepsilon}\int_{\gamma }(\frac{|v|^{2}}{2}-\frac{3}{2})\sqrt{\mu }%
\phi f \\
&& \ \ +\int_{t}^{t+\varepsilon}\iint_{\Omega \times \mathbf{R}^{3}}(\mathbf{I}-\mathbf{P})f(%
\frac{|v|^{2}}{2}-\frac{3}{2})\sqrt{\mu }(v\cdot \nabla_x )\phi \\
&& \ \ +\int_{t}^{t+\varepsilon}\iint_{\Omega \times \mathbf{R}^{3}} \phi g (%
\frac{|v|^{2}}{2}-\frac{3}{2})\sqrt{\mu }
,
\end{eqnarray*}%
and taking difference quotient, we obtain
\begin{eqnarray*}
&&\int_{\Omega} \phi(x)\partial_t c(t,x) dx\\
&& =\frac{2}{3} \int_{\Omega} b(t)\cdot \nabla_x \phi -\frac{1}{3\pi\sqrt{2\pi}}\int_{\gamma} (\frac{|v|^2}{2}-\frac{3}{2})\sqrt{\mu} \phi f(t)\\
&& \ \ + \frac{1}{3\pi\sqrt{2\pi}}\iint_{\Omega \times \mathbf{R}^{3}}(\mathbf{I}-\mathbf{P})f(t)(%
\frac{|v|^{2}}{2}-\frac{3}{2})\sqrt{\mu }(v\cdot \nabla_x )\phi \\
&& \ \ + \frac{1}{3\pi\sqrt{2\pi}}\iint_{\Omega \times \mathbf{R}^{3}} \phi g(t) (%
\frac{|v|^{2}}{2}-\frac{3}{2})\sqrt{\mu }.
\end{eqnarray*}
For fixed $t> 0$, we define a test function $\phi=\Phi_c$ where $\phi_c$ is defined in (\ref{phic}). The boundary terms vanish because of the Dirichlet
boundary condition on $\Phi_{c}$.  Then we have, for $t>0$,
\begin{eqnarray*}
 -\Delta \Phi _{c}=\partial _{t}c(t) , \ \ \ \ \
 \Phi_{c}|_{\partial \Omega }=0.
\end{eqnarray*}%
Notice that $\Phi_c = -\Delta^{-1}\partial_t c(t) = \partial_t \phi_c(t)$ in (\ref{phic}).
We follow the same procedure of estimates $\nabla_x \Delta ^{-1}\partial _{t}a$
and $\nabla_x \Delta ^{-1}\partial _{t}b$ to have
\begin{eqnarray*}
&&\|\nabla_x \Delta ^{-1}\partial _{t}c(t)\|_{2}^{2}  =
\int_{\Omega} |\nabla_x \Phi_c(x)|^2 dx\\
&&= \int_{\Omega} \Phi_c(x) \partial_t c(t,x) dx\\
&&\lesssim \
\varepsilon \{ \|\nabla_x \Phi_c \|_2^2 + \|  \Phi_c \|_2^2 \}+ \|b(t)\|_{2}^2
 + \|(\mathbf{I}-\mathbf{P})f(t)\|_2^2 + \|g(t)\|_2^2\\
&&\lesssim \ \varepsilon   \|\nabla_x \Phi_c \|_2^2  + \|b(t)\|_{2}^2
 + \|(\mathbf{I}-\mathbf{P})f(t)\|_2^2 + \|g(t)\|_2^2,
\end{eqnarray*}%
where we have used the Poincar$\acute{e}$ inequality. Finally we have, for all $t>0$,
\begin{eqnarray}
&&\|\nabla_x \partial _{t}\phi _{c}\|_{2}\lesssim \| \nabla_x \Delta^{-1} \partial_t c(t) \|_2\notag\\
&& \lesssim \ \|b(t)\|_{2}+\|(%
\mathbf{I}-\mathbf{P})f(t)\|_{2} + \|g(t)\|_2.  \label{-1c}
\end{eqnarray}

\vskip.3cm
\noindent{\it Step 4.} {Estimate of } ${a,b,c}$ contributions in (\ref{timeweak})

To estimate $c$ contribution in (\ref{timeweak}), we plug (\ref{phic}) into (\ref{timeweak}) to have from (\ref{insidesplit})
\begin{eqnarray*}
&& \int_{0}^{t}\iint_{\Omega \times \mathbf{R}^{3}}(|v|^{2}-\beta_c )v_{i}\sqrt{%
\mu }\partial _{t}\partial _{i}\phi _{c}f  \\
&=& \sum_{j=1}^{d}\int_{0}^{t}\iint_{\Omega \times \mathbf{R}^{3}}(|v|^{2}-\beta_c
)v_{i}v_{j}\mu(v) \partial _{t}\partial _{i}\phi _{c}b_{j}\\
 & &+\int_{0}^{t}\iint_{\Omega \times \mathbf{R}^{3}}(|v|^{2}-\beta_c )v_{i} %
\sqrt{\mu }\partial _{t}\partial _{i}\phi _{c} (\mathbf{I}-\mathbf{P})f .
\end{eqnarray*}%
The second line has non-zero contribution only for $j=i$ which leads to zero by the
definition of $\beta_c $ in (\ref{beta}). We thus have from (\ref{-1c}), {for $\varepsilon$ small,}
\begin{eqnarray}
&&\left|\int_{0}^{t}\iint_{\Omega \times \mathbf{R}^{3}}(|v|^{2}-\beta_c )v_{i}\sqrt{%
\mu }\partial _{t}\partial _{i}\phi _{c}f \right| \notag\\
&\lesssim&
\int_{0}^{t}\Big\{\|b\|_{2}+\|(\mathbf{I}-\mathbf{P})f\|_{2} +\|g\|_2\Big\}\|(\mathbf{I}-%
\mathbf{P})f\|_{2} \notag \\
&\lesssim &\varepsilon \int_{0}^{t}\|b \|_{2}^2+ \int_{0}^{t}\Big\{ \|(\mathbf{I}%
-\mathbf{P})f \|_{2}^{2} +\|g\|_2^2\Big\}.\label{ct}\\ &&\notag
\end{eqnarray}
{Combining  with  (\ref{cestimate}) in Lemma \ref{steadyabc} we conclude, for $\varepsilon$ small},
\begin{eqnarray}
\int_{{0}}^t \|c(s)\|_2^2 ds &\lesssim& G(t)-G(0)\label{timec}\\
&+&\int_{{0}}^t \Big\{ \|(\mathbf{I}-\mathbf{P})f(s)\|_{\nu}^2 + \|g(s)\|_2^2 + |(1-P_{\gamma})f(s)|_{2,+}^2 + |r(s)|_2^2+\varepsilon \|b(s)\|_2^2
\Big\}ds .\notag
\end{eqnarray}
To estimate $b$ contribution in (\ref{timeweak}), we plug (\ref{phibj}) into (\ref{timeweak}) to have from (\ref{insidesplit}) %
\begin{eqnarray}
&& \int_{0}^{t}\iint_{\Omega \times \mathbf{R}^{3}}(v_{i}^{2}-\beta_b )\sqrt{%
\mu }\partial _{t}\partial _{j}\phi _{b}^j f   \notag \\
&=& {{\int_{0}^{t}\iint_{\Omega \times \mathbf{R}^{3}}(v_{i}^{2}- \beta_b ){\mu
}\partial _{t}\partial _{j}\phi _{b}^{j} \{\frac{|v|^{2}}{2}-\frac{3}{2}\}c }}
\label{i_7} \\
&&+ \int_{0}^{t}\iint_{\Omega \times \mathbf{R}^{3}}(v_{i}^{2}- \beta_b ) %
\sqrt{\mu }\partial_t\partial _{j}\phi _{b}^{j} (\mathbf{I}-\mathbf{P})f, \notag
\end{eqnarray}%
where we used (\ref{alpha}) to remove the $a$ contribution. We thus have from (\ref{-1b}), \begin{eqnarray}
&&\left|\int_{0}^{t}\iint_{\Omega \times \mathbf{R}^{3}}(v_{i}^{2}-\beta_b )\sqrt{\mu
}\partial _{t}\partial _{j}\phi _{b}^jf \right|\notag\\
 &\lesssim
&\int_{0}^{t}\Big\{\|a\|_{2}+\|c\|_{2}+\|(\mathbf{I}-\mathbf{P}%
)f\|_{2} +\|g\|_2\Big\}\Big\{\|c\|_{2}+\|(\mathbf{I}-\mathbf{P})f\|_{2}\Big\}  \notag\\
&\lesssim &
\int_{0}^{t}\Big\{ \|(\mathbf{I}-\mathbf{P})f\|_{2}^{2}
+\|c\|_{2}^{2}+\|g\|_2^2 +{\varepsilon\|a\|_2^2\Big\}}.\label{bt1}
\end{eqnarray}%
Next we plug (\ref{phibij}) into (\ref{timeweak}) and we have from (\ref{-1b}), 
\begin{eqnarray}
&&  \int_{0}^{t}\int_{\Omega \times \mathbf{R}^{3}} |v|^{2} v_{i}v_{j}%
\sqrt{\mu }\partial _{t}\partial _{j}\phi _{b}^i f  \notag \\
&=&  \int_{0}^{t}\int_{\Omega \times \mathbf{R}^{3}} |v|^{2} v_{i}v_{j}\sqrt{\mu }\partial_{t}\partial_{j}\phi _{b}^i(\mathbf{I}-\mathbf{%
P})f  \notag \\
&\lesssim &\int_{0}^{t}\{\|a\|_{2}+\|c\|_{2}+\|(\mathbf{I}-\mathbf{P}%
)f\|_{2} +\|g\|_2 \}\|(\mathbf{I}-\mathbf{P})f\|_{2}  \notag \\
&\lesssim &\int_0^t \Big\{ \|(\mathbf{I}-\mathbf{P})f\|_{\nu}^2 + \|g\|_2^2 +{\varepsilon\big[\|a\|_2^2 +\|c\|_2^2}\big]
\Big\}
.  \label{bt2}
\end{eqnarray}
{Combining this with (\ref{b_i}) in Lemma \ref{steadyabc} we conclude from (\ref{timeweak}) that}
\begin{eqnarray}
&&\int_0^t \|b(s)\|_2^2 ds \lesssim G(t)-G(0)\label{timeb}\\
&&+\int_0^t \Big\{ \|(\mathbf{I}-\mathbf{P})f(s)\|_{\nu}^2 + \|g(s)\|_2^2 + |(1-P_{\gamma})f(s)|_{2,+}^2 + |r(s)|_2^2+\varepsilon\|a\|_2^2 +\|c\|_2^2(1+\varepsilon)
\Big\}ds.\notag
\end{eqnarray}

Finally in order to estimate $a$ contribution in (\ref{timeweak}) we plug (\ref{phia}) for into (\ref{timeweak}). We estimate%
\begin{eqnarray}
&& \int_{0}^{t}\int_{\Omega\times\mathbf{R}^3} (|v|^{2}-\beta_a )v_{i}\mu \partial _{t}\partial _{i}\phi
_{a}f  \label{at} \\
&=& \int_{0}^{t}\int_{\Omega\times\mathbf{R}^3}  (|v|^{2}- \beta_a )(v_{i})^{2}\mu \partial _{t}\partial
_{i}\phi _{a} b_{i}  \notag \\
&& +\int_{0}^{t}\int_{\Omega\times\mathbf{R}^3}  (|v|^{2}- \beta_a )v_{i}\mu \partial _{t}\partial
_{i}\phi _{a} (\mathbf{I}-\mathbf{P})f  \notag \\
&\lesssim &\int_{0}^{t}\{\|b(t)\|_{2,\Omega}+ |r|_{2}\}\{\|b\|_{2,\Omega}+\|(\mathbf{I}-\mathbf{P})f\|_{2}\}\notag
\end{eqnarray}%
Combining this with Lemma \ref{steadyabc} we conclude
\begin{eqnarray}
\int_0^t \|a(s)\|_2^2 ds &\lesssim& G(t)-G(0)\label{timea}\\
&+&\int_0^t \Big\{ \|(\mathbf{I}-\mathbf{P})f(s)\|_{\nu}^2 + \|g(s)\|_2^2 + |(1-P_{\gamma})f(s)|_{2,+}^2 + |r(s)|_2^2+{\|b\|_2^2}
\Big\}ds.\notag
\end{eqnarray}
From (\ref{timec}), (\ref{timeb}) and (\ref{timea}), we prove the lemma {by choosing $\varepsilon$ sufficiently small.}
\end{proof}

\vskip.3cm

\begin{proof}[Proof of Proposition \ref{dlinearl2}]
The proof is parallel to the steady case,  proof of Proposition \ref{linearl2}. We start with the approximating sequence. with $f^0 \equiv f_0$,
\begin{eqnarray}
\partial _{t}f^{\ell+1}+v\cdot \nabla _{x}f^{\ell+1}+\nu f^{\ell+1}-Kf^{\ell}&=&g,\text{ \
\ }f^{n+1}(0)=f_{0},\label{dapproximate}\\
f_{-}^{\ell+1}&=&(1-\frac{1}{j}%
)P_{\gamma }f ^{\ell}+r.\notag
\end{eqnarray}

\vskip .3cm
\noindent{\it Step 1.}  We first show that for fixed $j$, $f^{\ell}\rightarrow f^{j}$ as $%
\ell\rightarrow \infty $. Notice that by the compactness of $K$,
\begin{equation*}
(Kf^{\ell},f^{\ell+1})\leq \varepsilon \lbrack \|f^{\ell+1}\|_{\nu
}^{2}+\|f^{\ell}\|_{\nu }^{2}]+C_{\varepsilon
}[\|f^{\ell+1}\|_{2}^{2}+\|f^{\ell}\|_{2}^{2}].
\end{equation*}%
We first take an inner product(using Green's identity) with $f^{\ell+1}$ in (\ref{dapproximate}). With the same choice of $C_j$ in (\ref{Cj}) and using the boundary condition
\begin{eqnarray}
&&\|f^{\ell+1}(t)\|_{2}^{2}+(1-\varepsilon)\int_{0}^{t}\|f^{\ell+1}(s)\|_{\nu
}^{2}+\int_{0}^{t}|f^{\ell+1}(s)|_{2,+}^{2}ds  \notag \\
&\leq &\varepsilon\int_{0}^{t}\|f^{\ell}(s)\|_{\nu }^{2}ds+\left[ (1-\frac{1}{j})^{2}+\frac{1}{j^{2}}\right] \int_{0}^{t}|P_{%
\gamma }f^{\ell}|_{2,-}^{2}\notag\\
&&+C_{j}\int_{0}^{t}|r|_{2
}^{2}+C_{\varepsilon} \int_0^t \max_{1\leq i\leq \ell+1}\|f^i(s)\|_2^2 ds
+\int_{0}^{t}\|g(s)\|_{\nu
}^{2}ds+\|f_{0}\|_{2}^{2}.  \notag  \\
&\leq& \varepsilon\int_{0}^{t}\|f^{\ell}(s)\|_{\nu }^{2}ds+
\left[ (1-\frac{1}{j})^{2}+\frac{1}{j^{2}}\right]\int_{0}^{t}|f^{\ell}(s)|_{2,+}^{2}ds
\notag\\
&&+C_{j}\int_{0}^{t}|r|_{2
}^{2}+C_{\varepsilon} \int_0^t \max_{1\leq i\leq \ell+1}\|f^i(s)\|_2^2 ds
+\int_{0}^{t}\|g(s)\|_{\nu
}^{2}ds+\|f_{0}\|_{2}^{2}  \notag \\
&\leq& \max\left\{\frac{\varepsilon}{1-\varepsilon},\left[ (1-\frac{1}{j})^{2}+\frac{1}{j^{2}}\right]\right\} \left\{ (1-\varepsilon)\int_{0}^{t}\|f^{\ell}(s)\|_{\nu }^{2}ds+ \int_{0}^{t}|f^{\ell}(s)|_{2,+}^{2}ds\right\}\notag\\
&&+C_{j}\int_{0}^{t}|r|_{2
}^{2}+C_{\varepsilon} \int_0^t \max_{1\leq i\leq \ell+1}\|f^i(s)\|_2^2 ds
+\int_{0}^{t}\|g(s)\|_{\nu
}^{2}ds+\|f_{0}\|_{2}^{2}.\notag
\end{eqnarray}
We choose $\varepsilon>0$ sufficiently small, so that $\max \left\{\frac{\varepsilon}{1-\varepsilon},\left[ (1-\frac{1}{j})^{2}+\frac{1}{j^{2}}\right]\right\} \leq 1-\varepsilon$. Now we iterate again to have
\begin{eqnarray}
&\leq&(1-\varepsilon)^2 \left\{(1- \varepsilon)\int_{0}^{t}\|f^{\ell-1}(s)\|_{\nu }^{2}ds+\int_{0}^{t}|f^{\ell}(s)|^{2}_{2,+}ds\right\}  \notag \\
&&+[(1-\varepsilon )+1]\left\{ C_{\varepsilon}\int_{0}^{t}|r|_{2
}^{2}+C_{\varepsilon}\int_{0}^{t}\max_{1\leq i\leq
\ell+1}\|f^{i}(s)\|_{2}^{2}ds+\int_{0}^{t}\|g(s)\|_{\nu
}^{2}ds+\|f_{0}\|_{2}^{2}\right\}   \notag \\
& \vdots&  \notag \\
&\leq &(1-\varepsilon )^{\ell+1}\left\{(1- \varepsilon)\int_{0}^{t}\|f^{0}(s)\|_{\nu }^{2}ds+\int_{0}^{t}|f^{0}|^{2}_{2,+}\right\}\notag\\ &&
+ \frac{1-(1-\varepsilon)^{\ell+1}}{%
\varepsilon }\left\{ C_{\varepsilon}\int_{0}^{t}|r|_{2
}^{2}+C_{\varepsilon}\int_{0}^{t}\max_{1\leq i\leq
\ell+1}\|f^{i}(s)\|_{2}^{2}ds+\int_{0}^{t}\|g(s)\|_{\nu
}^{2}ds+\|f_{0}\|_{2}^{2}\right\} .  \notag\\
\label{lastest}
\end{eqnarray}%
We therefore have, from $f^0\equiv f_0$,
\begin{eqnarray*}
\max_{1\leq i\leq \ell+1}\|f^{i}(t)\|_{2}^{2}&\lesssim _{\varepsilon ,j}&\Big\{
\int_{0}^{t}|r|_{2}^{2}+\int_{0}^{t}\max_{1\leq i\leq
\ell+1}\|f^{i}(s)\|_{2}^{2}ds+\int_{0}^{t}\|g(s)\|_{\nu
}^{2}ds+\|f_{0}\|_{2}^{2}\\
&&\ \ \ \ \ \ \ \ \ \ \ \    + t \|f_0\|_{\nu}^2 + t |f_0|_{2,+}^2
\Big\}.
\end{eqnarray*}%
By Gronwall's lemma we have
\begin{equation*}
\max_{1\leq i\leq \ell+1}\|f^{i}(t)\|_{2}^{2}\lesssim _{\varepsilon ,j,t}\left\{
\int_{0}^{t}|r|_{2}^{2}+\int_{0}^{t}\|g(s)\|_{\nu
}^{2}ds+\|f_{0}\|_{2}^{2}+ t \|f_0\|_{\nu}^2 + t |f_0|_{2,+}^2\right\},
\end{equation*}%
where $t> 0$ is fixed.
This in turns leads to
\begin{eqnarray*}
&&\max_{1\leq i\leq \ell+1} \left\{\|f^{i}(t)\|_{2}^{2}+\int_{0}^{t}\|f^{i}(s)\|_{\nu
}^{2}+\int_{0}^{t}|f^{i}(s)|_{{+}}^{2}ds \right\}\\
&& \ \ \  \ \ \ \lesssim_{\varepsilon ,j,t}\left\{ \int_{0}^{t}|r|_{2
}^{2}+\int_{0}^{t}\|g(s)\|_{\nu }^{2}ds+\|f_{0}\|_{2}^{2}+ t \|f_0\|_{\nu}^2 + t |f_0|_{2,+}^2\right\}.
\end{eqnarray*}%
Upon taking the difference, we have
\begin{equation}
\partial _{t}[f^{\ell+1}-f^{\ell}]+v\cdot \nabla _{x}[f^{\ell+1}-f^{\ell}]+\nu \lbrack
f^{\ell+1}-f^{\ell}]=K[f^{\ell}-f^{\ell-1}],\text{ \ \ \ }  \label{dndifference}
\end{equation}%
with \ $[f^{\ell+1}-f^\ell](0)\equiv 0$ and  $f_{-}^{\ell+1}-f_{-}^{\ell}=(1-%
\frac{1}{j})P_{\gamma }[f ^{\ell}-f ^{\ell-1}]$.   Applying (\ref{lastest}) to $f^{\ell+1}-f^\ell$ yields%
\begin{eqnarray*}
&&\|f^{\ell+1}(t)-f^{\ell}(t)\|_{2}+\int_{0}^{t}\|f^{\ell+1}(s)-f^{v}(s)\|_{\nu
}^{2}+\int_{0}^{t}|f^{\ell+1}(s)-f^{\ell}(s)|_{2,+}^{2}ds \\
&\leq &(1-\varepsilon )^{\ell}\int_{0}^{t}|f^{1}-f^{0}|_{2,+}^{2}+\frac{1}{ \varepsilon }\left\{ \int_{0}^{t}\max_{1\leq i\leq
n}\|f^{i}(s)-f^{i-1}(s)\|_{\nu }^{2}ds\right\}.
\end{eqnarray*}%
From (\ref{lastest}), this implies, from the Gronwall's lemma, that, by taking maximum over $1\leq i\leq \ell+1,$
\begin{eqnarray*}
\|f^{\ell+1}(t)-f^{\ell}(t)\|_{2} &\lesssim_{\varepsilon ,j,t} &(1-\varepsilon
)^{\ell}\int_{0}^{t}|f^{1}-f^{0}|_{2,+}^{2},
\\
\int_{0}^{t}\|f^{\ell+1}(s)-f^{\ell}(s)\|_{\nu
}^{2}+\int_{0}^{t}|f^{\ell+1}(s)-f^{\ell}(s)|_{2,+}^{2}ds &\lesssim_{\varepsilon ,j,t}
&(1-\varepsilon )^{\ell}\int_{0}^{t}|f^{1}-f^{0}|_{2,+}^{2}.
\end{eqnarray*}%
Therefore, $f^{\ell}$ is a Cauchy sequence and $f^{\ell}\rightarrow f^{j}$ so that $f^{j}$
is the solution of
\begin{equation}
\partial _{t}f+v\cdot \nabla _{x}f+Lf=g,\text{ \ \ }f(0)=f_{0},\text{ \ \ \
\ \ }f_{-}=(1-\frac{1}{j})P_{\gamma }f +r.
\label{dlinearj}
\end{equation}

\vskip.3cm
\noindent {\it Step 2.} We let $j\rightarrow \infty $. Upon using Green's identity and the boundary condition and (\ref{eta}), for any $\eta>0$,
\begin{eqnarray}
&&\|f^{j}(t)\|_{2}^{2}+\int_{0}^{t}\|(\mathbf{I}-\mathbf{P})f^{j}(s)\|_{\nu
}^{2}ds+\int_{0}^{t}|(1-P_{\gamma })f^{j}(s)|^{2}_{2,+}ds  \notag \\
&\leq &\eta \int_{0}^{t}|P_{\gamma }f^{j}(s)|_{2,+}^{2}ds+C_{\eta
}\int_{0}^{t}|r|_{2}^{2}+\int_{0}^{t}\|g(s)\|_{\nu
}^{2}ds+\|f_{0}\|_{2}^{2}.  \label{jgreen}
\end{eqnarray}%
Note that, from Ukai's trace theorem (Lemma \ref{ukai}) and (\ref{dlinearj})
\begin{eqnarray}
\int_{0}^{t}|P_{\gamma }f^{j}(s)|_{2,+}^{2}ds &\lesssim
&\int_{0}^{t}|f^{j}(s)\mathbf{1}_{\gamma  ^{\varepsilon }}|_2^{2}ds  \label{dukai} \\
&\lesssim &\int_{0}^{t}\|f^{j}(s)\|_{2}^{2}ds+\int_{0}^{t}\|\partial
_{t}[f^{j}]^{2}+v\cdot \nabla _{x}[f^{j}]^{2}\|_{1}  \notag \\
&\lesssim &\int_{0}^{t}\|f^{j}(s)\|_{2}^{2}ds+\int_{0}^{t}
|(Lf^{j},f^{j})|+\int_{0}^{t}\|g(s)\|_{\nu }^{2}ds  \notag \\
&\lesssim &\int_{0}^{t}\|f^{j}(s)\|_{2}^{2}ds+\int_{0}^{t}  \|(\mathbf{I}-%
\mathbf{P})f^{j}(s)\|_{\nu }^{2}+\int_{0}^{t}\|g(s)\|_{\nu }^{2}ds.  \notag
\end{eqnarray}%
But from (\ref{dukai}), choosing $\eta $ small, we obtain
from (\ref{jgreen}) and (\ref{dukai}) that%
\begin{eqnarray*}
&&\|f^{j}(t)\|_{2}+\int_{0}^{t}\|(\mathbf{I}-\mathbf{P})f^{j}(s)\|_{\nu
}^{2}ds+\int_{0}^{t}|(1-P_{\gamma })f^{j}(s)|_{2,+}^{2}ds \\
&\lesssim &\int_{0}^{t}\|f^{j}(s)\|_{2}^{2}ds+\int_{0}^{t}|r|_{2
}^{2}+\int_{0}^{t}\|g(s)\|_{\nu }^{2}ds+\|f_{0}\|_{2}^{2}.
\end{eqnarray*}%
It follows from the Gronwall's lemma that
\begin{eqnarray}
\|f^{j}(t)\|^2_2 &\lesssim_{t} &\int_{0}^{t}|r|_{2
}^{2}+\int_{0}^{t}\|g(s)\|_{\nu }^{2}ds+\|f_{0}\|_{2}^{2},\label{tlinearj}
\end{eqnarray}
and hence
\begin{eqnarray}
&& \ \ \   \|f^{j}(t)\|_{2}^{2}+\int_{0}^{t}\|(\mathbf{I}-\mathbf{P})f^{j}(s)\|_{\nu
}^{2}ds+\int_{0}^{t}|(1-P_{\gamma })f^{j}(s)|_{2,+}^{2}ds \notag\\
&&\lesssim_{t}\int_{0}^{t}|r|_{2 }^{2}+\int_{0}^{t}\|g(s)\|_{\nu
}^{2}ds+\|f_{0}\|_{2}^{2}.\label{j1}
\end{eqnarray}%
Since $\|\mathbf{P}f^{j}(s)\|_{\nu }^{2}\lesssim $ $%
\|f^{j}(s)\|_{2}^{2}$, integrating (\ref{tlinearj}) from $0$ to $t$ and combining with (\ref{dukai}) we have
\begin{eqnarray}
&&\int_0^t \|\mathbf{P} f^j(s)\|_{\nu}^2 ds + \int_0^t |P_{\gamma}f^j(s)|_{2,+}ds\label{j2}\\
&&\lesssim_t \int_0^t \|(\mathbf{I}-\mathbf{P})f^j (s)\|_{\nu} ds + \int_0^t \big\{ |r(s)|_2^2 + \|g(s)\|_2^2\big\} ds + \|f_0\|_2^2.\notag
\end{eqnarray}
Add (\ref{j1}) and (\ref{j2}) together to have
\begin{equation}
\|f^{j}(t)\|_{2}^{2}+\int_{0}^{t}\|f^{j}(s)\|_{\nu
}^{2}ds+\int_{0}^{t}|f^{j}(s)|_{2}^{2}ds\lesssim
_{t}\int_{0}^{t}|r|_{2 }^{2}+\int_{0}^{t}\|g(s)\|_{\nu
}^{2}ds+\|f_{0}\|_{2}^{2}.  \label{jbound}
\end{equation}%
By taking a weak limit, we obtain a weak solution $f$ to (\ref{dlinear}) with
the same bound (\ref{jbound}). Taking difference, we have%
\begin{equation}
\partial _{t}[f^{j}-f]+v\cdot \nabla _{x}[f^{j}-f]+L[f^{j}-f]=0,\text{ \
\ \ \ }[f^{j}-f]_- =P_{\gamma }[f^{j}-f]+\frac{1}{j}P_{\gamma }f^{j},
\end{equation}%
with $[f^j-f](0)=0$. Applying the same estimate (\ref{jbound}) with $r=\frac{1}{j}P_{\gamma
}f^{j}$ we obtain
\begin{eqnarray*}
&&\|f^{j}(t)-f(t)\|_{2}^{2}+\int_{0}^{t}\|f^{j}(s)-f(s)\|_{\nu
}^{2}ds+\int_{0}^{t}|f^{j}(s)-f(s)|_{2}^{2}ds\\
&&\lesssim _{t}\frac{1}{j}%
\int_{0}^{t}|P_{\gamma }f^{j}|^{2}\lesssim _{t}\frac{1}{j}%
\rightarrow 0.
\end{eqnarray*}%
We thus construct $f$ as a $L^2$ solution to (\ref{dlinear}).

\vskip.2cm
\begin{remark} In both Step 1 and Step 2, we do not need the zero mass constraint which is instead essential for next step.
\end{remark}

\noindent {\it Step 3.} Decay estimate.

To conclude our proposition, let
\begin{equation*}
y(t)\equiv e^{\lambda t}f(t).
\end{equation*}%
We multiply  (\ref{dlinear}) by $e^{\lambda t}$, so that $y$ satisfies
\begin{equation}
\partial _{t}y+v\cdot \nabla _{x}y+Ly=\lambda y+e^{\lambda t}g,\text{ \ \ \
\ }y|_{\gamma _{-}}=P_{\gamma }y_{+}+e^{\lambda t}r.  \label{lineary}
\end{equation}%
We apply Green's identity and (\ref{dukai}), (\ref{jgreen}) together with $\eta \ll\lambda $ to obtain%
\begin{eqnarray}
&&\|y(t)\|_{2}^{2}+\int_{0}^{t}\|(\mathbf{I}-\mathbf{P}%
)y(s)\|_{\nu }^{2}+\int_{0}^{t}|(1-P_{\gamma })y(s)|^{2,+}  \label{ygreen} \\
&\leq &\lambda \int_{0}^{t}\|y(s)\|_{2}^{2}+\|y(0)\|_{2}^{2}+C_{\lambda
}\int_{0}^{t}e^{\lambda s}|r|_{2}^{2}+\int_{0}^{t}e^{\lambda
s}\|g(s)\|_{2}^{2}ds.  \notag
\end{eqnarray}%
From (\ref{dlinear}) we know that
\begin{equation*}
\iint_{\Omega\times\mathbf{R}^3} y \sqrt{\mu} = \iint_{\Omega\times\mathbf{R}^3} (\lambda y + e^{\lambda t}g) \sqrt{\mu} =0, \ \ \int_{\gamma_-} e^{\lambda t} r \sqrt{\mu} d\gamma=0.
\end{equation*}
Applying Lemma \ref{dabc} to (\ref{lineary}), we deduce
\begin{eqnarray}
&&\int_{0}^{t}\|\mathbf{P}y(s)\|_{\nu }^{2}ds\lesssim G(t)-G(0)  \label{py}
\\
&&+\int_{0}^{t}\|(\mathbf{I}-\mathbf{P})y(s)\|_{\nu
}^{2}ds+\int_{0}^{t}e^{\lambda s}\|g\|_{2}^{2}ds\\&&+\lambda
\int_{0}^{t}\|y\|_{2}^{2}ds+\int_{0}^{t}\{|(1-P_{\gamma })y(s)|_{2,+}^{2}+e^{\lambda s}|r|_{2 }^{2}\}ds , \notag
\end{eqnarray}%
where $G(t)\lesssim \|y(t)\|_{2}^{2}$. Multiplying with a small constant$%
\times $(\ref{py})$+$(\ref{ygreen}) to obtain, for some $\varepsilon \ll 1$
\begin{eqnarray*}
&&[\|y(t)\|_{2}^{2}-\varepsilon G(t)]+\varepsilon \left\{ \int_{0}^{t}\|(%
\mathbf{I}-\mathbf{P})y(s)\|_{\nu }^{2}+\|\mathbf{P}y(s)\|_{\nu
}^{2}\right\} +\varepsilon \int_{0}^{t}|(I-P_{\gamma })y(s)\|^{2}_{2,+} \\
&\leq &C\lambda \int_{0}^{t}\|y(s)\|_{2}^{2}+[\|y(0)\|_{2}^{2}-\varepsilon
G(0)]+C_{\varepsilon ,\lambda }\int_{0}^{t}e^{\lambda s}|r|_{2
}^{2}+\int_{0}^{t}e^{\lambda s}\|g(s)\|_{2}^{2}ds.
\end{eqnarray*}%
By further choosing $\lambda \ll \varepsilon $, and $\|(\mathbf{I}-\mathbf{P}%
)y(s)\|_{\nu }^{2}+\|\mathbf{P}y(s)\|_{\nu }^{2}\gtrsim \|y(s)\|_{2}^{2}$
for hard potentials, we conclude our proposition.
\end{proof}

\bigskip
\section{$L^{\infty}$ Stability and Non-Negativity}

To conclude the proof of Theorem \ref{main3} we need $L^\infty$ estimates. They are provided, in the linear case, by next
\begin{proposition}
\label{dlinearlinfty} Let $\|w_\rho f_{0}\|_{\infty }+|\langle v\rangle w_\rho r|_{\infty
}+\|w_\rho g\|_{\infty }<+\infty $ and $\iint \sqrt{\mu }g=\int_{\gamma }r\sqrt{\mu } = \iint f_0 \sqrt{\mu}
=0$. Then the solution $f$ to (\ref{dlinear}) satisfies
\begin{equation*}
\|w_\rho f(t)\|_{\infty }+|w_\rho f(t)|_{\infty }\leq e^{-\lambda
t}\big\{\|w_\rho f_{0}\|_{\infty }+\sup e^{\lambda s}\|w_\rho g\|_{\infty
}+\int_{0}^{t}e^{\lambda s}|\langle v\rangle w_\rho r(s)|_{\infty }ds\big\}.
\end{equation*}%
Furthermore, if $f_{0}|_{\gamma_{-}}=P_{\gamma }f_{0}+r_{0}$, $f_{0},r$ and $g$ are
continuous, then $f(t,x,v)$ is continuous away from $\mathfrak{D}$.  In particular, it $\Omega$ is convex then $\mathfrak{D}=\gamma_0$.
\end{proposition}

\begin{proof}
We use exactly  the same approximating sequence and repeat step 1 to step 3 of Proposition \ref{dlinearl2} to show that  $f^{j},f$ are bounded, via Lemma \ref{iterationlinfty}.
We denote $h^{\ell }=w_{\varrho}f^{\ell}$, where $w_{\varrho}$ is scaled weight in (\ref{weight}). Rewrite (\ref{dapproximate}) as%
\begin{eqnarray}
\partial _{t}h^{\ell+1}+v\cdot \nabla _{x}h^{\ell+1}+\nu h^{\ell+1} &=&K_{w_{\varrho}}h^{\ell}+w_{\varrho}g,
\label{dhapproximate} \\
h_{ {-}}^{\ell+1} &=&\frac{1-\frac{1}{j}}{\tilde{w}_{\varrho}(v)}%
\int_{n(x)\cdot v^{\prime }>0}h^{\ell}(t,x,v^{\prime })\tilde{w}%
_{\varrho}(v^{\prime })d\sigma +w_{\varrho}r.  \notag
\end{eqnarray}

\noindent{\it Step 1.}We take $\ell\rightarrow \infty $ in $L^{\infty }$.   Upon
integrating over the characteristic lines $\frac{dx}{dt}=v$, and $\frac{dv}{%
dt}=0$ and the boundary condition repeatedly, along the stochastic cycles, we obtain (by replacing $1-%
\frac{1}{j}$ with $1$ and $\varepsilon =0$) that (\ref{iteration}) is valid for $h^{\ell+1}$.
Therefore, for $\ell\geq 2\varrho$, by Lemma \ref{iterationlinfty} that
\begin{eqnarray*}
&&\sup_{0\leq s\leq T_0}e^{\frac{\nu _{0}}{2}s}|h^{\ell+1}(s,x,v)| \\
&\leq &\frac{1}{8}\max_{0\leq l\leq 2k}\sup_{0\leq s\leq T_0}\{e^{\frac{\nu
_{0}}{2}s}\|h^{\ell-l}(s)\|_{\infty }\}+C(k)\max_{1\leq l\leq
2k}\int_{0}^{T_0}\|f^{\ell-l}(s)\|_{2}ds \\
&&+\|h_{0}\|_{\infty }+k\left[ \sup_{0\leq s\leq T_0}\left\{ e^{\frac{\nu _{0}%
}{2}s}|w_{\varrho}r(s)|_{\infty }\right\} +\sup_{0\leq s\leq T_0}\left\{ e^{\frac{\nu
_{0}}{2}s}\left\Vert \frac{w_{\varrho}g(s)}{\langle v\rangle }\right\Vert _{\infty
}\right\} \right]  \\
&\equiv &\frac{1}{8}\max_{0\leq l\leq 2k}\sup_{0\leq s\leq T_0}\{e^{\frac{\nu
_{0}}{2}s}\|h^{\ell-l}(s)\|_{\infty }\}+C(k)\max_{1\leq l\leq
2k}\int_{0}^{T_0}\|f^{\ell-l}(s)\|_{2}ds \\
&&+D,
\end{eqnarray*}%
where $T_0 =\varrho^{4/5}, \ \varrho =k$ are chosen to be sufficiently large but fixed.
Now this is valid for all $\ell\geq 2k$. By induction on $\ell$, we can iterate such bound for $\ell+2,....\ell+2k$ to obtain%
\begin{eqnarray}
&&\sup_{0\leq s\leq T_0}e^{\frac{\nu _{0}}{2}s}\|h^{\ell+i}(s)\|_{\infty }
\label{dnbound} \\
&\leq &\frac{1}{8}\max_{1\leq l\leq 2k}\sup_{0\leq s\leq T_0}\{e^{\frac{\nu
_{0}}{2}s}\|h^{\ell+i-l}(s)\|_{\infty }\}+C(k)\max_{-2k\leq l\leq
2k}\int_{0}^{T_0}\|f^{\ell-l}(s)\|_{2}ds+D  \notag \\
&\leq &\frac{1}{4}\max_{1\leq l\leq 2k} \sup_{0\leq s \leq T_0}\{e^{\frac{\nu _{0}}{2}%
s}\|h^{\ell+i-l-1}(s)\|_{\infty }\}+2C(k)\max_{-2k\leq l\leq
2k}\int_{0}^{T_0}\|f^{\ell-l}(s)\|_{2}ds+2D  \notag \\
&\vdots&\notag \\
&\leq &\frac{1}{4}\max_{1\leq l\leq 2k}\sup_{0\leq s\leq T_0}\{e^{\frac{\nu_0}{2}s}\|h^{\ell-l}(s)\|_{\infty
}\}+(i+1)C(k)\left\{ \max_{-2k\leq l\leq
2k}\int_{0}^{T_0}\|f^{\ell-l}(s)\|_{2}ds+D\right\} .  \notag
\end{eqnarray}%
Taking maximum over $i=1,...2k$, and by induction on $\ell$,
\begin{eqnarray*}
&&\max_{1\leq l\leq 2k}\sup_{0\leq s\leq T_0}e^{\frac{\nu _{0}}{2}%
s}\|h^{\ell+1-l}(s)\|_{\infty } \\
&\leq &\frac{1}{4}\max_{1\leq l\leq 2k}\sup_{0\leq s\leq T_0}\{e^{\frac{\nu
_{0}}{2}s}\|h^{\ell-l}(s)\|_{\infty }\}+C(k)\Big[\max_{-2k\leq l\leq
2k}\int_{0}^{T_0}\|f^{\ell-l}(s)\|_{2}ds\Big]+D \\
&\leq &\frac{1}{4^{2}}\max_{1\leq l\leq 2k}\sup_{0\leq s\leq T_0}\{e^{\frac{%
\nu _{0}}{2}s}\|h^{\ell-1-l}(s)\|_{\infty }\}+(1+\frac{1}{4})\Big\{
C(k)\Big[\max_{-2k\leq l\leq 2k}\int_{0}^{T_0}\|f^{\ell-l}(s)\|_{2}ds\Big]+D\Big\}  \\
&\vdots & \\
&\leq &\frac{1}{4^{\ell/2k}}\max_{1\leq l\leq 2k}\sup_{0\leq s\leq T_0}\{e^{\frac{%
\nu _{0}}{2}s}\|h^{l}(s)\|_{\infty }\}+\Big\{ C(k)\Big[\max_{-2k\leq l\leq
2k}\int_{0}^{T_0}\|f^{\ell-l}(s)\|_{2}ds\Big]+D\Big\},
\end{eqnarray*}%
where we may assume that $\ell$ is a multiple of $k$.  Now for $\max_{1\leq
l\leq 2k}\|h^{l}\|_{\infty }$, we can use (\ref{iteration}) for $k=1$
repeatedly for $h^{2k}\rightarrow h^{2k-1\text{ }}...\rightarrow h_{0}$ to
obtain
\begin{equation*}
\max_{1\leq l\leq 2k}\sup_{0\leq s\leq T_0}\{e^{\frac{\nu _{0}}{2}%
s}\|h^{l}(s)\|_{\infty }\}\lesssim _{k}\|h_{0}\|_{\infty }+\sup_{0\leq s\leq
T_0}\left\{ e^{\frac{\nu _{0}}{2}s}|w_{\varrho}r(s)|_{\infty }\right\} +\sup_{0\leq
s\leq T_0}\left\{ e^{\frac{\nu _{0}}{2}s}\left\Vert \frac{w_{\varrho}g(s)}{\langle
v\rangle }\right\Vert _{\infty }\right\} .
\end{equation*}%
We therefore conclude that
\begin{eqnarray*}
&&\max_{1\leq l\leq 2k}\sup_{0\leq s\leq T_0}e^{\frac{\nu _{0}}{2}%
s}\|h^{\ell+1-l}(s)\|_{\infty } \\
&&\lesssim_{k}  C(k)\Big\{\max_{-2k\leq l\leq
2k}\int_{0}^{T_0}\|f^{\ell-l}(s)\|_{2}ds+D\Big\}  \\
&& \ \ \ +\|h_{0}\|_{\infty }+\sup_{0\leq s\leq T_0}\left\{ e^{\frac{\nu _{0}}{2}%
s}|w_{\varrho}r(s)|_{\infty }\right\} +\sup_{0\leq s\leq T_0}\left\{ e^{\frac{\nu _{0}}{2%
}s}\left\Vert \frac{w_{\varrho}g(s)}{\langle v\rangle }\right\Vert _{\infty }\right\} ,
\end{eqnarray*}%
Now $\max_{1\leq l\leq \infty }\|f^{l}\|_{2}$ is bounded by step 1 in the
proof of Proposition \ref{dlinearl2},  $\|\cdot \|_{2}$ is bounded by $%
\|w_{\varrho}\cdot \|_{\infty }$ for $\beta > 3$ and $|r|_{2}$ is bounded by $%
|w_{\varrho}\langle v\rangle r(s)|_{\infty }$.  Hence, there is a limit (unique) solution $%
h_{\ell}\rightarrow h=w_{\varrho}f\in L^{\infty }$.  Furthermore, $h$ satisfies (\ref%
{iteration}) with $h^{\ell+1}\equiv h$. By subtracting $h^{\ell+1}-h$ in (\ref%
{iteration}) with $D=0$, we obtain from Lemma \ref{iterationlinfty} for $h^{\ell+1}-h$ and (\ref%
{dnbound}):%
\begin{eqnarray*}
&&\max_{1\leq l\leq 2k}\sup_{0\leq s\leq T_0}e^{\frac{\nu _{0}}{2}%
s}\|h^{\ell+1-l}(s)-h(s)\|_{\infty }\\
&& \leq \frac{1}{4^{\ell/2k}}\max_{1\leq l\leq
2k}\sup_{0\leq s\leq T_0}\{e^{\frac{\nu _{0}}{2}s}\|h^{l}(s)-h(s)\|_{\infty }\} \\
&& \ \ +\Big\{ C(k)\max_{-2k\leq l\leq
2k}\int_{0}^{T_0}\|f^{\ell-l}(s)-f(s)\|_{2}ds\Big\}.
\end{eqnarray*}%
From step 1 of Lemma \ref{dlinearl2} we deduce $\sup_{0\leq s \leq T_0}\|h^{\ell+1-l}(s)-h(s)\|_{\infty
}\rightarrow 0$ for $\ell$ large.

\vskip.3cm
\noindent{\it
Step 2.} We take  $j\rightarrow \infty$.  Let $f^{j}$ to be the solution
to (\ref{dlinearj}) and integrate along $\frac{dx}{dt}=v,\frac{dv}{dt}=0$
repeatedly. Again (\ref{iteration}) is valid for $h^{\ell}\equiv h^{j}$
(replacing $(1-\frac{1}{j})$ by $1$ and $\varepsilon =0$) and $l=0$.
Lemma \ref{iterationlinfty} implies%
\begin{eqnarray*}
&&\sup_{0\leq s\leq T_0}e^{\frac{\nu _{0}}{2}s}\|h^{j}(s)\|_{\infty }\\
 &&\leq
\frac{1}{8}\sup_{0\leq s\leq T_0}\{e^{\frac{\nu _{0}}{2}s}\|h^{j}(s)\|_{\infty
}\}+C(k)\int_{0}^{T_0}\|f^{j}(s)\|_{2}ds \\
&& \ \ \ +\|h_{0}\|_{\infty }+k\left[ \sup_{0\leq s\leq T_0}\left\{ e^{\frac{\nu _{0}%
}{2}s}|w_{\varrho}r(s)|_{\infty }\right\} +\sup_{0\leq s\leq T_0}\left\{ e^{\frac{\nu
_{0}}{2}s}\left\Vert \frac{w_{\varrho}g(s)}{\langle v\rangle }\right\Vert _{\infty
}\right\} \right],
\end{eqnarray*}%
with $T_0 =\varrho^{4/5}, \ \varrho=k$, and $\varrho$ sufficiently large but fixed. Therefore, by an induction over $j$,
\begin{eqnarray*}
&&\sup_{0\leq s\leq T_0}e^{\frac{\nu _{0}}{2}s}\|h^{j}(s)\|_{\infty } \leq
C(k)\int_{0}^{T_0}\|f^{j}(s)\|_{2}ds+\|h_{0}\|_{\infty } \\
&&+k\left[ \sup_{0\leq s\leq T_0}\left\{ e^{\frac{\nu _{0}}{2}%
s}|w_{\varrho}r(s)|_{\infty }\right\} +\sup_{0\leq s\leq T_0}\left\{ e^{\frac{\nu _{0}}{2%
}s}\left\Vert \frac{w_{\varrho}g(s)}{\langle v\rangle }\right\Vert _{\infty }\right\} %
\right].
\end{eqnarray*}%
Since $\int_{0}^{T_0}\|f^{j}(s)\|_{2}ds$ is bounded from step 2 of Proposition \ref{dlinearl2} (with no mass constraint), this implies that $%
\|h^{j}\|_{\infty }$ is uniformly bounded and we obtain a (unique) solution $%
h=w_{\varrho}f\in L^{\infty }$. Taking the difference, we have
\begin{eqnarray*}
\partial _{t}[h^{j}-h]&+&v\cdot \nabla _{x}[h^{j}-h]+\nu [ h^{j}-h]
=K_{w_{\varrho}}[h^{j}-h], \\
{[h^{j}-h]}_- &=&\frac{1}{\tilde{w}_{\varrho}(v)}%
\int_{n(x)\cdot v^{\prime }>0}[h ^{j}-h ](t,x,v^{\prime })\tilde{w}_{\varrho}(v^{\prime })d\sigma  \\
&&-\frac{1}{j}\frac{1}{\tilde{w}_{\varrho}(v)}\int_{n(x)\cdot v^{\prime
}>0}h ^{j}(t,x,v^{\prime })\tilde{w}_{\varrho}(v^{\prime })d\sigma .
\end{eqnarray*}%
We regard $\frac{1}{j}\frac{1}{\tilde{w}_{\rho}(v)}\int_{n(x)\cdot v^{\prime
}>0}h_{\gamma _{-}}(t,x,v^{\prime })\tilde{w}_{\rho}(v^{\prime })d\sigma$ as $r$ in Lemma \ref{iterationlinfty}.
So Lemma \ref{iterationlinfty} implies that%
\begin{eqnarray*}
&&\sup_{0\leq s\leq T_0}e^{\frac{\nu _{0}}{2}s}\|h^{j}(s)-h(s)\|_{\infty }\\
&&\lesssim
C(k)\int_{0}^{T_0}\|f^{j}(s)-f(s)\|_{2}ds
+k \sup_{0\leq s\leq T_0}|w_{\varrho} r(s)|_{\infty} \lesssim \frac{1}{j}
,
\end{eqnarray*}%
which goes to zero as $j$ to $\infty $.

We now obtain a $L^{\infty }$ solution $h=w_\rho f$ to (\ref{dlinear}). Since we
have $L^{\infty }$ convergence at each step, we deduce that $h$ is
continuous away from $\mathfrak{D}$.

To obtain decay estimate, we integrate (\ref{dlinear}) over the trajectory. The estimate  (\ref{iteration}%
) is valid for $h^{\ell}\equiv h$ so that from Lemma \ref%
{iterationlinfty} for $\varrho>0$ sufficiently large we have:%
\begin{eqnarray}
&&\sup_{0\leq s\leq T_0}e^{\frac{\nu _{0}}{2}s}\|h(s)\|_{\infty }+\sup_{0\leq
s\leq T_0}e^{\frac{\nu _{0}}{2}s}|h(s)|_{\infty }  \label{hinfty} \\
&\leq &\|h_{0}\|_{\infty }+C(k)\bigg[ \sup_{0\leq s\leq T_0}\left\{ e^{\frac{%
\nu _{0}}{2}s}|w_{\varrho}r(s)|_{\infty }\right\} +\sup_{0\leq s\leq T_0}\left\{ e^{%
\frac{\nu _{0}}{2}s}\left\Vert \frac{w_{\varrho}g(s)}{\langle v\rangle }\right\Vert
_{\infty }\right\}\notag\\
&& \ \ \ \ \ \ \ \ \ \ \ \ \ \ \ \ \ \ \ \ \ \ \ \ \  \ \ \ \ \ \  \ \ \ \ \ \ \ \ \ \ \ \ \ \ \ \ \ \ \ \ \ \ \ +\int_{0}^{T_0}\|f(s)\|_{2}ds\bigg],  \notag
\end{eqnarray}%
where $T_{0}=\varrho^{4/5}, \ \varrho=k$ and $\varrho$ sufficiently large. Letting $s=T_0$ in (\ref{hinfty}), we obtain
\begin{eqnarray}
&&\|h(T_{0})\|_{\infty }+|h(T_{0})|_{\infty } \leq e^{-\frac{\nu _{0}}{2}T_{0}}\|h_{0}\|_{\infty }\label{hT0}\\
&& \ \ \ \ +C(k)\left[
\sup_{0\leq s\leq T_0}\left\{ |w_{\varrho}r(s)|_{\infty }\right\} +\sup_{0\leq s\leq
T_0}\left\{ \left\Vert \frac{w_{\varrho}g(s)}{\langle v\rangle }\right\Vert _{\infty
}\right\} +\int_{0}^{T_{0}}\|f(s)\|_{2}ds\right].\notag
\end{eqnarray}%
Let
\begin{equation*}
R\equiv \sup_{0\leq s\leq \infty }\left\{ e^{\frac{\lambda }{2}s}|\langle
v\rangle w_{\varrho}r(s)|_{\infty }\right\} +\sup_{0\leq s\leq \infty }\left\{ e^{%
\frac{\lambda }{2}s}\|w_{\varrho}g(s)\|_{\infty }\right\} +\|w_{\varrho}f_{0}\|_{\infty }.
\end{equation*}%
Since $\beta > 3/2$, Proposition \ref{dlinearl2} implies that
\begin{equation*}
\|f(t)\|_{2}\lesssim_{\varrho} e^{-\frac{\lambda }{2}t}R.
\end{equation*}%
Let $m$ be an integer.  We look at times $t=m T_0$.  By induction on $m$, we have from the definition of $R$ and (\ref{hT0}) for $\lambda \ll \nu_0,$
\begin{eqnarray*}
&&\|h([m+1]T_{0})\|_{\infty }+|h([m+1]T_{0})|_{\infty } \\
&\leq &e^{-\frac{\nu _{0}}{2}T_{0}}\|h(mT_{0})\|_{\infty } \\
&&+C(\varrho)e^{-\frac{\nu _{0}}{2}T_{0}}\bigg[ \sup_{0\leq s\leq
T_{0}}\left\{ |w_{\varrho}r(s+mT_{0})|_{\infty }\right\} +\sup_{0\leq s\leq T_0}\left\{
\left\Vert \frac{w_{\varrho}g(s+mT_{0})}{\langle v\rangle }\right\Vert _{\infty
}\right\}\\
 &&  \ \ \ \ \ \ \ \ \ \ \ \ \ \ \ \ \ \ \ \ \ \ \ \ \ \ \ \ \ \ \ \ \ \ \ \ \ \ \ \ \ \ \ \ \ \ \ \ \ \ \ \ \ \ \ +\int_{0}^{T_{0}}\|f(s+mT_{0})\|_{2}ds\bigg]  \\
&\leq &e^{-\frac{\nu _{0}}{2}T_{0}}\|h(mT_{0})\|_{\infty }+C(\varrho)e^{-%
\frac{\nu _{0}}{2}T_{0}-\frac{m\lambda T_{0}}{2}}\times  \\
&&\bigg[ \sup_{0\leq s\leq T_{0}}e^{\frac{m\lambda T_{0}}{2}}\left\{
|w_{\varrho}r(s+mT_{0})|_{\infty }\right\} +\sup_{0\leq s\leq t}e^{\frac{m\lambda T_{0}%
}{2}}\left\Vert \frac{w_{\varrho}g(s+mT_{0})}{\langle v\rangle }\right\Vert _{\infty
}\\
&& \ \ \ \ \ \ \ \ \ \ \ \ \ \ \ \ \ \ \ \ \ \ \ \ \ \ \ \ \ \ \ \ \ \ \ \ \ \ \ \ \ \ \ \ \ \ +\int_{0}^{T_{0}}e^{\frac{m\lambda T_{0}}{2}}\|f(s+mT_{0})\|_{2}ds\bigg]
\\
&\leq &e^{-\frac{\nu _{0}}{2}T_{0}}\|h(mT_{0})\|_{\infty
}+C(\varrho)(1+T_{0})e^{-\frac{\nu _{0}}{2}T_{0}-\frac{m\lambda T_{0}}{2}}R \\
&\leq &e^{-2\frac{\nu _{0}}{2}T_{0}}\|h([m-1]T_{0})\|_{\infty
}+C(\varrho)(1+T_{0})e^{-\frac{\nu _{0}}{2}T_{0}-\frac{m\lambda T_{0}}{2}%
}R\\
  &&+C(\varrho)(1+T_{0})e^{-\frac{\nu _{0}}{2}T_{0}}e^{-\frac{\nu _{0}}{2}%
T_{0}}e^{-\frac{[m-1]\lambda T_{0}}{2}}R \\
&=&e^{-2\frac{\nu _{0}}{2}T_{0}}\|h([m-1]T_{0})\|_{\infty
}+C(\varrho)(1+T_{0})e^{-\frac{\nu _{0}}{2}T_{0}-\frac{m\lambda T_{0}}{2}%
}R[1+e^{-\frac{\nu _{0}-\lambda }{2}T_{0}}] \\
&\leq &e^{-3\frac{\nu _{0}}{2}T_{0}}\|h([m-2]T_{0})\|_{\infty
}\\
&&+C(\varrho)(1+T_{0})e^{-\frac{\nu _{0}}{2}T_{0}-\frac{m\lambda T_{0}}{2}%
}R[1+e^{-\frac{\nu _{0}-\lambda }{2}T_{0}}+e^{-2\frac{\nu _{0}-\lambda }{2}%
T_{0}}] \\
&\vdots&\\
&\leq &e^{-m\frac{\nu _{0}}{2}T_{0}}\|h_{0}\|_{\infty
}+C(\varrho)(1+T_{0})e^{-\frac{\nu _{0}}{2}T_{0}-\frac{m\lambda T_{0}}{2}%
}R\sum_{j}e^{-j\frac{\nu _{0}-\lambda }{2}T_{0}} \\
&\leq &C_{k,\nu _{0},\lambda,\varrho }e^{-\frac{m\lambda T_{0}}{2}}R.
\end{eqnarray*}%
Combining with (\ref{hinfty}) for $0\leq t_{1}\leq T_{0}$, we deduce that
for any $t=mT_{0}+t_{1}$, $\|h(t)\|_{\infty} + |h(t)|_{\infty} \lesssim_{k,\nu_0,\lambda,\varrho} e^{-t}R$. We deduce our proposition with $w=w_{\varrho}$ for $\varrho>0$ sufficiently large but fixed. A simple scaling of $\varrho$ concludes the proof of the proposition.
\end{proof}

\vskip .3cm
By using the above $L^\infty$ estimate we can conclude the proof of Theorem \ref{main3}.

\vskip .3cm
\begin{proof}[Proof of Theorem \ref{main3}]
\bigskip We consider the following iteration sequence
\begin{eqnarray*}
\partial _{t}f^{\ell+1}+v\cdot \nabla _{x}f^{\ell+1}+Lf^{\ell+1}
&=& L_{\sqrt{\mu}f_s}f^{\ell}+\Gamma (f^{\ell},f^{\ell}), \\
f_{ {-}}^{\ell+1} &=&P_{\gamma }f^{\ell+1} +\frac{\mu _{\delta }-\mu }{\sqrt{%
\mu }}\int_{\gamma _{+}}f^{\ell} \sqrt{\mu} (n\cdot v) dv.
\end{eqnarray*}%
Clearly we have
\begin{equation*}
 \iint_{\Omega\times\mathbf{R}^3}\{L_{\sqrt{\mu}f_s}f^{\ell}+\Gamma (f^{\ell},f^{\ell})\} \sqrt{\mu }dx dv=0,\text{ \ \ \ }%
\int \frac{\mu _{\delta }-\mu }{\sqrt{\mu }}\left\{ \int_{\gamma
_{+}}f^{\ell}\right\} d\gamma =0.
\end{equation*}%
Recall $w_{\varrho}(v)= (1+\varrho^2 |v|^2)^{\frac{\beta}{2}}e^{\zeta |v|^2}$ in (\ref{weight}). Note that for $0\leq \zeta<\frac{1}{4},$
\begin{eqnarray*}
&& \ \ \ \ \left\Vert e^{\frac{\lambda  s }{2}}w_{\varrho}\left\{ \frac{1}{\langle v\rangle }%
[L_{\sqrt{\mu}f_s}f^{\ell} +\Gamma (f^{\ell},f^{\ell})(s)\right\} \right\Vert _{\infty
} \\
&&\lesssim \ \ \delta \sup_{0\leq s\leq t}\|e^{\frac{\lambda s}{2}%
}w_{\varrho} f^{\ell}(s)\|_{\infty }+\left\{ \sup_{0\leq s\leq t}\|e^{\frac{\lambda s}{2}%
}w_{\varrho} f^\ell(s)\|_{\infty }\right\} ^{2}.
\end{eqnarray*}
Using (\ref{mudeltaEXP}) for $0\leq \zeta < \frac{1}{4+2\delta},$
\begin{equation*}
\ \ \ \left\vert e^{\frac{\lambda t}{2}}w_{\varrho}\langle v\rangle \frac{\mu _{\delta }-\mu
}{\sqrt{\mu }}\left\{ \int_{\gamma _{+}}f^{\ell}\right\} \right\vert _{\infty } \ \
\lesssim \ \ \delta \sup_{0\leq s\leq t}|e^{\frac{\lambda s}{2}} f ^{\ell}(s)|_{\infty }.
\end{equation*}
By Proposition \ref{dlinearlinfty}, we deduce
\begin{eqnarray*}
&&\sup_{0\leq s\leq t}\|e^{\frac{\lambda s}{2}}w_{\varrho}f^{\ell+1}(s)\|_{\infty
}+\sup_{0\leq s\leq t}|e^{\frac{\lambda s}{2}}w_{\varrho}f^{\ell+1}(s)|_{\infty
} \\
&\lesssim &\|w_{\varrho}f_{0}\|_{\infty }+\delta \sup_{0\leq s\leq t}\|e^{\frac{%
\lambda s}{2}}w_{\varrho}f^{\ell}(s)\|_{\infty }+\delta \sup_{0\leq s\leq t}|e^{\frac{%
\lambda s}{2}}w_{\varrho}f^{\ell}(s)|_{\infty }\\&+&\left\{ \sup_{0\leq s\leq t}\|e^{\frac{%
\lambda s}{2}}w_{\varrho}f^\ell(s)\|_{\infty }\right\} ^{2}.
\end{eqnarray*}%
For $\delta $ small, there exists a $\varepsilon _{0}$ (uniform in $\delta $%
) such that, if the initial data satisfy (\ref{epsilon0}), then
\begin{equation*}
\sup_{0\leq s\leq t}\|e^{\frac{\lambda s}{2}}w_{\varrho}f^{\ell+1}(s)\|_{\infty
}+\sup_{0\leq s\leq t}|e^{\frac{\lambda s}{2}}w_{\varrho}f ^{\ell+1}(s)|_{\infty
}\lesssim \|w_{\varrho}f_{0}\|_{\infty }.
\end{equation*}%
By taking difference $f^{\ell+1}-f^{\ell}$, we deduce that%
\begin{eqnarray*}
&&\partial _{t}[f^{\ell+1}-f^{\ell}]+v\cdot \nabla
_{x}[f^{\ell+1}-f^{\ell}]+L[f^{\ell+1}-f^{\ell}] \\
&& \ \ \ \ = L_{\sqrt{\mu}f_s}[f^{\ell}-f^{\ell-1}]+\Gamma (f^{\ell}-f^{\ell-1},f^{\ell})+\Gamma
(f^{\ell-1},f^{\ell}-f^{\ell-1}), \\
&&[f^{\ell+1}-f^{\ell}]_- =P_{\gamma }[f^{\ell+1}-f^{\ell}]  +%
\frac{\mu _{\delta }-\mu }{\sqrt{\mu }}\int_{\gamma
_{+}}[f^{\ell}-f^{\ell-1}](n(x)\cdot v) dv,
\end{eqnarray*}%
with $f^{\ell+1}-f^{\ell}=0$ initially. Repeating the same argument, we obtain
\begin{eqnarray*}
&&\sup_{0\leq s\leq t}\|e^{\frac{\lambda s}{2}}w_{\varrho} [f^{\ell+1}-f^{\ell}](s)\|_{\infty
}+\sup_{0\leq s\leq t}|e^{\frac{\lambda s}{2}}w_{\varrho} [f ^{\ell+1}-f ^{\ell}](s)|_{\infty } \\
&\lesssim &[\delta +\sup_{0\leq s\leq t}\|e^{\frac{\lambda s}{2}%
}w_{\varrho} f^{\ell}(s)\|_{\infty }+\sup_{0\leq s\leq t}\|e^{\frac{\lambda s}{2}%
}w_{\varrho} f^{\ell-1}(s)\|_{\infty }]\sup_{0\leq s\leq t}\|e^{\frac{\lambda s}{2}%
}w_{\varrho} [f^{\ell}-f^{\ell-1}](s)\|_{\infty }.
\end{eqnarray*}%
This implies that $f^{\ell+1}$ is a Cauchy sequence. The uniqueness is standard.

\vskip .3cm
\noindent{\it Positivity.}
To conclude the positivity of $F_{s}$, we need to show
\begin{equation*}
F_{s}+\sqrt{\mu }f(t)\geq 0,
\end{equation*}%
if initially $F_{s}+\sqrt{\mu }f_{0}\geq 0$.   To this end, we first assume the cross section $B$ is bounded and hence $%
\nu $ is bounded. We note that  previous approximating sequence is not well suited to show the positivity. Therefore, we need to design a different iterative sequence. We use the following one:
\begin{eqnarray*}
\partial _{t}F^{\ell+1}+v\cdot \nabla_x F^{\ell+1}+\nu(F^{\ell})F^{\ell+1}
&=&Q_{\text{gain}}(F^{\ell},F^{\ell}), \\
F_{ {-}}^{\ell+1} &=&\mu _{\delta }\int_{\gamma _{+}}F^{\ell}  (n(x)\cdot v) dv,
\end{eqnarray*}%
where {$F^0=F_s+\sqrt{\mu}f_0$ and}
\begin{equation} \nu(F)=\int_{\mathbf{R}^3}dv_*\int_{\mathbf{S}^2}d\omega B(v-v_*,\omega) F(v_*),
\end{equation}
with initial condition $F^{\ell+1}(0)=F_{s}+\sqrt{\mu }f_{0}\geq 0$. Clearly,
such an iteration preserves the non-negativity. We need to show $F^{\ell}$ is
convergent to conclude the non negativity of the (unique!) limit $F(t)\geq 0$.  Writing $%
F^{\ell+1}=F_{s}+\sqrt{\mu }f^{\ell+1}$, we have
\begin{eqnarray}
&& \ \  \partial_t f^{\ell+1} + v\cdot\nabla_x f^{\ell+1} + \nu(v)f^{\ell+1} -Kf^\ell\label{n}\\
 &&\ \ \ \ \ \ \ \ \ \ \   \ \ \ \ \ = \Gamma_{\text{gain}}(f^\ell,f^\ell)-\nu(\sqrt{\mu}f^\ell)f^{\ell+1}-\nu(\sqrt{\mu}f_s)f^{\ell+1}-\nu(\sqrt{\mu}f^\ell)f_s\notag\\
&& \ \ \ \ \ \ \ \ \ \    \ \ \ \ \ \ \ \ \ \
+\frac{1}{\sqrt{\mu}}\left\{ Q_{\text{gain}}(\sqrt{\mu}f^\ell,\sqrt{\mu}f_s) + Q_{\text{gain}}(\sqrt{\mu}f_s,\sqrt{\mu}f^\ell)
\right\},\notag\\
&& \ \ \ \ \ \ \ \ \ \ \ \ \ \ \ \ f_{ {-}}^{\ell+1} =P_{\gamma }f^{\ell}  +\frac{\mu _{\delta }-\mu }{\sqrt{%
\mu }}\int_{\gamma _{+}}f^{\ell} \sqrt{\mu} (n(x)\cdot v)dv.\notag
\end{eqnarray}%
Taking inner product (Green's identity) with $f^{\ell+1}$, from
Ukai's trace Theorem, and $(Kf^{\ell},f^{\ell+1})\lesssim
\|f^{\ell}\|_{2 }^{2}+\|f^{\ell+1}\|_{2 }^{2}$ (bounded cross section $B$) and the boundary condition we have
\begin{eqnarray}
&&\|f^{\ell+1}(t)\|_{2}^2+\int_{0}^{t}\|f^{\ell+1}\|_{2}^{2}+ \int_{0}^{t}|f^{\ell+1}|_{2,+}^{2}\notag \\
&\lesssim &\|f_{0}\|_{2}^2 +[1 + C\delta^2 ]\int_{0}^{t}|P_{\gamma }f^{\ell}|_{2}^{2} +   \int_0^t \|f^\ell\|_{2}^2\label{1cdelta}
\\
&&+[\delta +\|wf^{\ell}\|_{\infty
}+ \|wf^{\ell}\|_{\infty
}^2]\int_{0}^{t}\|f^{\ell+1}\|_{2 }^{2}.\notag
\end{eqnarray}%
By (\ref{dukai}) and the equation (\ref{n}), assuming
$
\max_{1\leq l\leq \ell}\|wf^{l}\|_{\infty }<+\infty$ we have
\begin{eqnarray}
\int_{0}^{t}|P_{\gamma }f^{\ell}|_{2}^{2}&\lesssim&
\int_{0}^{t}\|f^{\ell}\|_{2}^{2}+ \int_{0}^{t}\|f^{\ell-1}\|_{2 }^{2}+ [ \delta  +\|wf^{\ell-1}\|_{\infty
} + \|wf^{\ell-1}\|_{\infty
}^2 ]\int_{0}^{t}\|f^{ \ell }\|_{2 }^{2}\notag\\
&\lesssim& \max_{\ell-1\leq l \leq \ell } \int_0^t \| f^l \|_2^2 ds
.\label{Pgmax}
\end{eqnarray}%
Splitting $1+C\delta^2=1-\delta + (\delta+ C\delta^2)$ in (\ref{1cdelta}) and using (\ref{Pgmax}) we have
\begin{eqnarray*}
&&\|f^{\ell+1}(t)\|_{2}^2+ \int_{0}^{t}\|f^{\ell+1}\|_{2
}^{2}+\int_{0}^{t}|f ^{\ell+1}|_{2,+}^{2} \\
&\lesssim&\|f_0\|_{2}^2+
[1-\delta] \int_0^t |P_{\gamma}f^\ell|_2^2 ds
+   \max_{\ell-1\leq l \leq \ell} \int_0^t \|f^l\|_2^2 ds.
\end{eqnarray*}%
Then we iterate to obtain
\begin{eqnarray}
&&\|f^{\ell+1}\|_2^2 +  \int_0^t \| f^{\ell+1}\|_{2}^2 + \int_0^t |f^{\ell+1} |_{2,+}^2\notag\\
&\leq& [1-\delta]^2\int_0^t |P_{\gamma}f^{\ell-1}|_2^2
+ [1+(1-\delta)]\left\{ \|f^0\|_2^2  + \max_{1\leq l \leq \ell+1} \int_0^t\|f^l(s)\|_2^2
\right\}\notag\\
&\vdots&  \notag\\
&\leq& (1-\delta)^{\ell+1}\int_0^t |P_{\gamma}f^{0}|_2^2
+ \frac{1-(1-\delta)^{\ell+1}}{\delta}\left\{ \|f^0\|_2^2 +  \max_{1\leq l \leq \ell+1} \int_0^t \|f^l(s)\|_2^2
\right\}.\notag\\
\label{2iterate}
\end{eqnarray}
Taking maximum over $1\leq l \leq \ell+1$,
\begin{eqnarray*}
\max_{1\leq l \leq \ell+1} \| f^l(t)\|_{2}^2 \lesssim \|f^0\|_2^2 +  \int_0^t |P_{\gamma}f^0|_2^2
+ \int_0^t \max_{1\leq l \leq \ell+1} \| f^l(s)\|_{2}^2 ds.
\end{eqnarray*}
By Gronwall's lemma, from $f^0=f_0,$
\begin{equation}
\max_{1\leq 1\leq n+1}\|f^{l}(t)\|_{2}\lesssim_t \|f_{0}\|_{2} + t |P_{\gamma}f_0|_{2}, \label{2bound}
\end{equation}%
uniformly bounded.

We now apply Lemma \ref{iterationlinfty} with
\begin{equation*}
g=L_{\sqrt{\mu}f_s}f^\ell+\Gamma (f^{\ell},f^{\ell}),\text{ \ \ }r=\frac{\mu
_{\delta }-\mu }{\sqrt{\mu }}\int_{\gamma _{+}}f^{\ell} \sqrt{\mu} d\gamma,
\end{equation*}%
for some $T_0$ large, with $f^{0}=f_{0}$, using the same arguments as in proof of (\ref{hT0}), to get
\begin{eqnarray*}
&&\sup_{0\leq s\leq T_0}e^{\frac{\nu _{0}}{2}s}|w_{\varrho}f^{\ell+1}(s,x,v)| \\
&\leq &\frac{1}{4}\max_{1\leq l\leq 2k}\sup_{0\leq s\leq T_0}\{e^{\frac{\nu
_{0}}{2}s}\|w_{\varrho}f^{\ell+1-l}(s)\|_{\infty }\}+\|w_{\varrho}f_{0}\|_{\infty } \\
&&+T_0\left[ \delta \max_{1\leq l\leq \ell}\sup_{0\leq s\leq T_0}\left\{ e^{\frac{%
\nu _{0}}{2}s}|w_{\varrho}f^{\ell+1-l}(s)|_{\infty }\right\} +\max_{0\leq l\leq
2k}\sup_{0\leq s\leq T_0}\left\{ e^{\frac{\nu _{0}}{2}s}|w_{\varrho}f^{\ell+1-l}(s)|_{%
\infty }\right\} ^{2}\right]  \\
&&+C(T_0)\max_{1\leq l\leq 2k}\int_{0}^{T_0}\|f^{\ell-1}(s)\|_{2}ds.
\end{eqnarray*}%
where $\ell\geq 2k$.  For
\begin{equation}
\delta +\max_{0\leq l\leq 2k}\sup_{0\leq s\leq t}\left\{ e^{\frac{\nu _{0}}{2%
}s}|w_{\varrho}f^{\ell+1-l}(s)|_{\infty }\right\} \ \ll \ 1,\label{assumption}
\end{equation}%
small, we obtain from $\beta>3,$
\begin{equation*}
\sup_{0\leq s\leq t}e^{\frac{\nu _{0}}{2}s}|w_{\varrho}f^{\ell+1}(s,x,v)|\leq \frac{1}{2}%
\max_{1\leq l\leq 2k}\sup_{0\leq s\leq t}\{e^{\frac{\nu _{0}}{2}%
s}\|w_{\varrho}f^{\ell+1-l}(s)\|_{\infty }\}+C_{T_0}\|w_{\varrho}f_{0}\|_{\infty }.
\end{equation*}%
Hence we obtain%
\begin{equation*}
\max_{1\leq l\leq 2k}\sup_{0\leq s\leq T_0}e^{\frac{\nu _{0}}{2}%
s}|w_{\varrho}f^{\ell+2-l}(s,x,v)|\leq \frac{1}{2}\max_{1\leq l\leq 2k}\sup_{0\leq s\leq
T_0}\{e^{\frac{\nu _{0}}{2}s}\|w_{\varrho}f^{\ell+1-l}(s)\|_{\infty
}\}+C_{k}\|w_{\varrho}f_{0}\|_{\infty },
\end{equation*}%
and
\begin{equation*}
\max_{1\leq l\leq 2k}\sup_{0\leq s\leq T_0}e^{\frac{\nu _{0}}{2}%
s}|w_{\varrho}f^{\ell+2-l}(s,x,v)|\leq \frac{1}{2^{\ell/2k}}\max_{1\leq l\leq 2k}\sup_{0\leq
s\leq T_0}\{e^{\frac{\nu _{0}}{2}s}\|w_{\varrho}f^{l}(s)\|_{\infty
}\}+2C_{k}\|w_{\varrho}f_{0}\|_{\infty }.
\end{equation*}%
But for $1\leq l\leq 2k$, we use Lemma with $k=1$ to get
\begin{equation}
\max_{1\leq l\leq 2k}\sup_{0\leq s\leq T_0}\{e^{\frac{\nu _{0}}{2}%
s}\|w_{\varrho}f^{l}(s)\|_{\infty }\}\leq C_{k}\|w_{\varrho}f_{0}\|_{\infty },\label{maxT0}
\end{equation}%
so that (\ref{assumption}) is valid as long as $\|w_{\varrho}f_{0}\|_{\infty }$ is\ sufficiently small.
We therefore obtain uniform bound
\begin{equation*}
\max_{1\leq l \leq \ell}\sup_{0\leq s\leq T_0}\|w_{\varrho}f^{l}(s)\|_{\infty }\leq C_{k}\|w_{\varrho}f_{0}\|_{\infty
}.
\end{equation*}%
This leads to $w_{\varrho} f^\ell \rightarrow w_{\varrho} f \in L^{\infty}$.  Furthermore, $f$ satisfies (\ref{n}) with $f^{\ell+1}=f^\ell=f$.  Therefore $f^{\ell+1}-f$ satisfies
\begin{eqnarray*}
&&\{\partial_t+ v\cdot\nabla_x + \nu(v)\} [f^{\ell+1}-f] - K[f^\ell-f]= R^{\prime}, \ \ \ \ [f^{\ell+1}-f](0)\equiv 0,\\
&& [f^{\ell+1}-f]_{+} = P_{\gamma}[f^\ell -f] + \frac{\mu_{\delta}-\mu}{\sqrt{\mu}} \int_{\gamma_+} [f^\ell-f]\sqrt{\mu} (n(x)\cdot v) dv,
\end{eqnarray*}
where
\begin{equation}
|(R^{\prime}, f^{\ell+1}-f)|
\lesssim \ \{\delta+ \|w_{\varrho}f^\ell\|_{\infty}+ \|w_{\varrho}f^\ell\|_{\infty}^2+ \|w_{\varrho}f \|_{\infty}+ \|w_{\varrho}f \|_{\infty}^2 \}\{ \|f^{\ell+1}-f\|_2 \}.
\end{equation}
Notice that from (\ref{maxT0})
\begin{equation}
\max_\ell \|w_{\varrho}f^\ell\|_{\infty}, \ \|w_{\varrho} f \|_{\infty} \leq C_{\varrho} \|w_{\varrho} f_0 \|_{\infty}.\notag
\end{equation}
With the same proof of (\ref{Pgmax})
\begin{equation*}
\int_0^t |P_{\gamma}[f^\ell-f]|_2^2 \lesssim \max_{\ell-1 \leq i \leq \ell} \| f^i -f \|_2^2 ds.
\end{equation*}
Combining the two above estimates we have
\begin{eqnarray*}
&&\|f^{\ell+1}(t)-f(t)\|_2^2 + \int_0^t \| f^{\ell+1}-f\|_2^2 + \int_0^t |f^{\ell+1}-f|_{2,+}^2\\
&&\lesssim [1-\delta] \int_0^t |P_{\gamma}[f^\ell-f]|_2^2 ds + \max_{\ell-1 \leq i \leq \ell}\int_0^t \|f^i-f\|_2^2 ds.
\end{eqnarray*}
We iterate as for  (\ref{2iterate}) to obtain
\begin{eqnarray*}
\max_{1\leq i \leq \ell+1 }\| f^i(t)-f(t)\|_2^2 \lesssim \int_0^t \max_{1\leq i \leq \ell+1} \|f^i-f\|_2^2 ds.
\end{eqnarray*}
Then applying Gronwall's lemma and choosing small $t_0>0$ we conclude that $f^j$ is Cauchy in $[0,t]$ and $f^j\rightarrow f$ by uniqueness and $F(t)\geq 0$ for $0\leq t\leq t_0$. We can repeat this process to show $F(t)\geq 0,$
for all $t\geq 0$.

Finally, we take $F_{0}\sim\mu $, so that $F_0\geq 0$
and hence $F(t)\geq 0$. Therefore $\lim_{t\rightarrow +\infty} F(t) =F_{s}\geq
0$.  This completes the proof in the case $B$ bounded. To remove this limitation, we use a cut-off procedure as before and reduce to the bounded case where previous result holds. Then we pass to the limit in the cut-off using the a priori bounds and uniqueness.
\end{proof}

\vskip 1cm
\noindent{\bf Acknowledgements}: {\it We are deeply indebted with prof. K. Aoki for permitting us to use the plots from the paper \cite{OAS}. Y. Guo thanks C. Mouhot for stimulating discussions and he thanks the support of Beijing International Mathematical Center.  R. E. and R. M. and C. K. wish to thank the Math Department of Brown University and ICERM where part of this research was performed, for their kind hospitality and support.

The research of Y. G. is supported in part by NSF \#0905255 and FRG grants as well as a Chinese NSF
grant \#10828103. The research of R. E. and R. M. is partially supported by MIUR and GNFM-INdAM. The research of C. K. is supported in part by Herchel Smith fund.}

\end{document}